\documentclass[12pt]{article}
\usepackage{array}
\newcolumntype{P}[1]{>{\centering\arraybackslash}p{#1}}
\newcolumntype{M}[1]{>{\centering\arraybackslash}m{#1}}

\usepackage{amsmath,amsthm,amsfonts}
\usepackage{enumerate}
\usepackage{natbib}
\usepackage{url} 
\DeclareMathOperator*{\plim}{plim}
\usepackage{graphicx,psfrag,epsf}

\usepackage{setspace}
 \DeclareGraphicsExtensions{.eps, .ps}
\usepackage{soul,color}
\usepackage{booktabs}
\usepackage{threeparttable}
\usepackage{amssymb}

\usepackage{newfloat}
\usepackage{multirow}
\usepackage{arydshln}
\usepackage{xcolor}
\usepackage{hyperref}

\graphicspath{{./art/}}

\usepackage{algorithm}
\usepackage{algpseudocode}

\algrenewcommand\algorithmicrequire{\textbf{Input:}}
\algrenewcommand\algorithmicensure{\textbf{Output:}}

\usepackage{tikz}
\usetikzlibrary{arrows.meta,
	calc,
	positioning,
	quotes,
    decorations.pathreplacing}
\usetikzlibrary{shapes.geometric}

\def\E{\mathbb{E}}

\newtheorem{theorem}{Theorem}[section]

\newtheorem{lemma}[theorem]{Lemma}
\newtheorem{assumption}[theorem]{Assumption}
\newtheorem{remark}[theorem]{Remark}

\newtheorem{proposition}[theorem]{Proposition}

\newcommand{\indep}{\perp \!\!\! \perp}
\newcommand{\R}{\mathbb{R}}

\newcommand{\Cov}{\mathrm{Cov}}
\newcommand{\Var}{\mathrm{Var}}
\newcommand{\Tr}{\mathrm{Tr}}

\newcommand{\vertiii}[1]{{\left\vert\kern-0.25ex\left\vert\kern-0.25ex\left\vert #1 
    \right\vert\kern-0.25ex\right\vert\kern-0.25ex\right\vert}}

\newcommand\independent{\protect\mathpalette{\protect\independenT}{\perp}}
\def\independenT#1#2{\mathrel{\rlap{$#1#2$}\mkern2mu{#1#2}}}
\newcommand\given{\,|\,}

\numberwithin{equation}{section}

\newcommand{\blind}{1}

\begin{document}

\def\spacingset#1{\renewcommand{\baselinestretch}%
{#1}\small\normalsize}

\spacingset{1}

\if1\blind
{
  \title{\bf Conditional Independence Testing in Time Series}

  \author{
    Jieru Shi\\
    Department of Statistical Science, University College London, UK\\
    \texttt{jieru.shi@ucl.ac.uk}\\[0.8em]
    and\\[0.8em]
    Rajen D. Shah\\
Statistical Laboratory,\\
Department of Pure Mathematics and Mathematical Statistics,\\
University of Cambridge, UK\\
\texttt{r.shah@statslab.cam.ac.uk}
  }

  \date{}
  \maketitle
}
\fi

\if0\blind
{
  \bigskip
  \bigskip
  \bigskip
  \begin{center}
    {\LARGE\bf Conditional Independence Testing in Time Series}
  \end{center}
  \medskip
}
\fi

\bigskip
\begin{abstract}
We consider the problem of testing Granger causality in time series, specifically, whether the future outcome $Y_{t+1}$ and the exposure history $\bar X_t$ are conditionally independent given the history of $\bar Y_t$ and a set of confounding variables $\bar Z_t$ up to time $t$. 
This testing procedure distinguishes true causal effects from associations driven by common external processes, supporting reliable decision-makings in applications such as finance, neuroscience, and climate science.
While traditional approaches assume a linear vector autoregressive (VAR) model and are vulnerable to misspecification, we instead address a model-free version of the problem.
We propose nonlinearly regressing both the outcome and exposure on the joint history of $Y$ and $Z$, and calculating a test statistic based on the sample covariance of residuals, we call the Generalised Temporal Covariance Measure (GTCM). To account for heteroscedasticity and improve power against local alternatives, we incorporate variance weights and employ a data-adaptive test based on polynomial lag expansions. 
The type I error control of the test relies on the relatively weak assumption that user-chosen regression procedures estimate conditional means at a sufficiently fast rate that is slow enough to accommodate nonparametric settings. By further assuming stability of the regression procedures and weak dependence in the time series, we can utilise the entire dataset to estimate the conditional means without splitting the time series into subsets.
\end{abstract}



\noindent%
{\it Keywords:}  Algorithm Stability; Causal Discovery; Weak Dependence; Double Machine Learning; Generalised Covariance Measure; Granger Causality.
\vfill


\newpage
\spacingset{1.45}

\section{Introduction}

It is of great interest to understand the causal structure of multivariate time series. For instance, in neuroscience, one may ask whether electrical activity in one brain region drives subsequent responses elsewhere during a cognitive task. In macroeconomics, one may ask whether changes in central bank interest rates affect future inflation, or whether both variables are responding to common macroeconomic conditions.

A framework for addressing these questions is provided by the notion of a structural causal models (SCMs) \citep{peters2017elements}. To fix ideas, suppose we have series $(X_t, Y_t, Z_t) \in \R \times \R \times \R^{d_Z}$ for $t=1,\ldots,K$. Here \(X_t\) denotes the source process, \(Y_t\) the target process, and \(Z_t\) collects other observed processes used for adjustment. We are interested in whether the past of \(X\) carries additional information about \(Y_{t+1}\) beyond the observed history of \(Y\) and \(Z\).

\sloppy An SCM in particular expresses \(Y_{t+1}\) through a structural equation of the form $Y_{t+1}
=
f_t(\bar Y_t,\bar X_t,\bar Z_t,\varepsilon_t),$ where \(\bar Y_t=(Y_t,Y_{t-1},\ldots,Y_1)\), and \(\bar X_t\) and \(\bar Z_t\) are defined analogously. Here \(\{\varepsilon_t\}_{t=1}^K\) are exogenous noise variables. If, in fact,
\begin{equation}
\label{eq:SCM_CI}
Y_{t+1} = f_t(\bar Y_t,\bar Z_t,\varepsilon_t)
\end{equation}
for all \(t\),
we may 
interpret this as the absence of a directed causal effect from \(X\) to
\(Y\); see Figure~\ref{fig:causal_diagram} for a causal directed acyclic graph where this is the case.\footnote{This interpretation relies on the usual causal sufficiency condition: after conditioning on the observed history of \(Y\) and \(Z\), there are no unmeasured common causes of the relevant history of \(X\) and \(Y_{t+1}\).}

\begin{figure}[htbp]
    \centering
   \begin{tikzpicture}[
        scale=0.65, 
        transform shape,
        node distance=1.5cm and 1.5cm,
        main_node/.style={
            circle, 
            draw=black, 
            thick, 
            minimum size=1.2cm, 
            inner sep=0pt, 
            font=\Large\bfseries,
            fill=white
        },
        conditioned_node/.style={
            main_node,
            fill=gray!15 
        },
        black_arrow/.style={
            -{Stealth[length=2.5mm, width=1.8mm]}, 
            thick, 
            draw=black,
            shorten >= 2pt,
            shorten <= 2pt
        },
        hypothesis_arrow/.style={
            black_arrow,
            draw=red!80!black,
            dashed
        },
        question_mark_text/.style={
            font=\huge\bfseries, 
            text=red!80!black, 
            above,
            yshift=2pt
        }
    ]

        \node[conditioned_node] (yt_1) {$Y_{t-1}$};
        \node[conditioned_node] (yt)   [above=of yt_1] {$Y_{t}$};
        \node[main_node]        (yt1)  [above=of yt]   {$Y_{t+1}$};

        \node[conditioned_node] (zt_1) [right=of yt_1] {$Z_{t-1}$};
        \node[conditioned_node] (zt)   [above=of zt_1] {$Z_{t}$};

        \node[main_node] (xt_1) [right=of zt_1] {$X_{t-1}$};
        \node[main_node] (xt)   [above=of xt_1] {$X_{t}$};

        \draw[black_arrow] (yt_1) -- (yt);
        \draw[black_arrow] (yt) -- (yt1);
        \draw[black_arrow] (zt_1) -- (zt);
        \draw[black_arrow] (xt_1) -- (xt);

        \draw[black_arrow] (yt_1) to[bend left=45] (yt1);

        \draw[black_arrow] (zt_1) -- (yt);
        \draw[black_arrow] (zt) -- (yt1);

        \draw[black_arrow] (zt_1) -- (xt);
        \draw[black_arrow] (xt_1) -- (zt);

    \end{tikzpicture}
    \caption{Local edge-testing problem for the directed link \(X_t\to Y_{t+1}\), after adjustment for $\bar Y_t$ and $\bar Z_t$.}
    \label{fig:causal_diagram}
\end{figure}
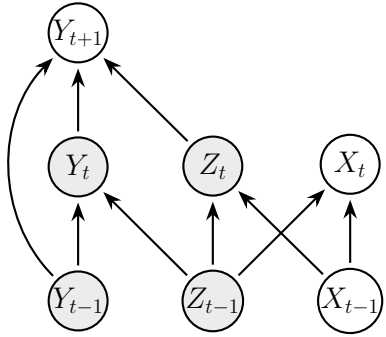
The structural restriction \eqref{eq:SCM_CI} implies the key conditional independence relation
\begin{equation}
\label{eq:CI_basic}
Y_{t+1}
\independent
\bar X_t
\given
\{\bar Y_t,\bar Z_t\}.
\end{equation}

It is therefore of great interest to test this conditional independence using the available data. The null hypothesis that \eqref{eq:CI_basic} holds is sometimes known as Granger non-causality, though the original paper \citep{granger1969investigating} introducing the notion of Granger causality assumed stationary linear autoregressive models: one compares the residual variance of the best linear predictor of \(Y\) with and without the past of \(X\). An issue with relying on such parametric models is that in contemporary datasets of interest, they may provide a poor representation of the data-generating process. As a result, misspecification can lead both to invalid Type I error control and to loss of power against alternatives outside the chosen parametric form.

On the other hand, it is known \citep{shah2020hardness} that testing conditional independence without imposing any restrictions is fundamentally hard when the conditioning variable is continuous.
The time series setting considered in this paper carries with it the additional challenge posed by temporal dependence, as well as the possibility that the joint distribution of $(X_t, Y_t, Z_t)$ may change across time. While much of the literature in time series analysis assumes stationarity, which would preclude this latter possibility, in many settings, such a restriction would be implausible.

In this work, we consider the problem of testing \eqref{eq:CI_basic}.
Our approach is based on the generalised covariance measure (GCM)
\citep{shah2020hardness}. For i.i.d. observations, the GCM targets conditional
independence statements of the form \(X \independent Y \given Z\) by regressing
\(X\) and \(Y\) separately on \(Z\), and forming a test statistic from the
empirical covariance of the resulting residuals. An attraction of the GCM is that its validity relies primarily on sufficiently accurate estimation of the two conditional mean functions, which can be performed using any chosen machine learning method.
A related approach, an instance of what has come to be known as double or debiased machine learning \citep{chen2022debiased}, can be used to estimate a target parameter in a partially linear model with i.i.d.\ data. In such settings, cross-fitting, a form of sample-splitting, is typically used to control certain so-called cross terms or empirical process terms. In the case of the GCM, such cross-fitting is in fact unnecessary as the terms can be controlled under the null of condititonal independence.

However, neither this argument nor the promise of independence afforded by cross-fitting carries over directly to our setting where we observe a single realisation of a multivariate time series $\{(X_t,Y_t,Z_t)\}_{t=1}^K$ exhibiting serial dependence and potentially time-varying conditional mean functions. 
A central theoretical contribution of this paper is demonstrating that the 
cross terms are nevertheless asymptotically negligible under weak temporal dependence and block-stability conditions on the regression methods.
This justifies regression on the full series using flexible machine learning methods while guaranteeing valid Type I error control, without requiring strict stationarity or sample-splitting schemes.

On the methodological side, to detect delayed and nonlinear effects, we propose computing test statistics corresponding to multiple lags and transformations of the $X$ series, which we aggregate in a Wald-type statistic. We further introduce an approach that aims to adapt to the unknown required lag length and complexity of transformations needed.
Finally, to account for conditional heteroscedasticity, a common feature of time series innovations as in ARCH/GARCH processes, we incorporate adaptive weighting to maximise testing power under non-constant variance.

\subsection{Related work}

Granger causality asks whether the history of one process improves prediction of another after conditioning on the remaining observed history \citep{granger1969investigating,granger1980testing}. Its interpretation depends on both the conditioning information and the predictive target. Noncausality in conditional mean is weaker than distributional noncausality expressed through conditional independence \citep{eichler2012graphical,shojaie2022granger}, and neither notion alone identifies an interventional effect in the presence of unmeasured common causes. We formulate a distributional conditional-independence null and test residual moment restrictions implied by that null.

Classical Granger procedures test restrictions on lag coefficients in finite-order vector autoregressions \citep{lutkepohl2005new}. These tests primarily concern linear predictability and, under correct conditional-mean specification, conditional-mean noncausality. Extensions address integrated and cointegrated systems \citep{toda1995statistical}, high-dimensional VARs through post-selection inference \citep{hecq2023granger}, and temporal instability through recursive or quantile-based procedures \citep{shi2018change,mayer2024quantile}. These methods retain relatively structured models for the relevant predictive relationships.

A substantial nonparametric literature instead tests nonlinear forms of Granger noncausality. Early distribution-based procedures include \citet{hiemstra1994testing} and the corrected test of \citet{diks2006new}. Subsequent methods use characteristic functions, Hellinger distance, copulas, conditional moments, and projections \citep{su2007consistent,su2008hellinger,bouezmarni2012copula,
nishiyama2011consistent,zhou2022projective,song2021nonparametric}. Information-theoretic approaches based on transfer entropy provide another route and coincide with variance-based Granger causality under Gaussianity \citep{schreiber2000measuring,barnett2009granger}. More recently, neural methods have been proposed for temporal graph learning and nonlinear prediction \citep{tank2022neural,nauta2019causal,cheng2023neural}. Of particular relevance, \citet{hui2025deep} develop a neural, doubly robust test of Granger noncausality in conditional mean for stationary mixing processes, using multiplier-bootstrap calibration without sample splitting. Its target and stationarity assumptions differ from those considered here.

Conditional-independence tests also form local components of causal-discovery procedures. PCMCI, PCMCI+ and LPCMCI estimate lagged or contemporaneous graphs under different assumptions on temporal ordering and latent confounding \citep{runge2019detecting,runge2020discovering,gerhardus2020high}. Their guarantees depend both on the graphical assumptions and on the local conditional-independence tests. In particular, testing one source lag conditional on the remaining lags and jointly testing an entire source history correspond to different null hypotheses.

For general conditional independence, kernel methods extend ideas related to
HSIC \citep{gretton2005measuring,gretton2007kernel,zhang2012kernel}, while nearest-neighbour conditional mutual information provides another nonparametric approach \citep{runge2018conditional}. Without structural restrictions, conditional-independence testing with continuous conditioning variables is subject to fundamental impossibility results \citep{shah2020hardness}. The generalised covariance measure (GCM) instead tests residual moment restrictions under suitable regression accuracy conditions \citep{shah2020hardness}. Weighted GCMs target cancellation of conditional residual covariances \citep{scheidegger2022weighted}, while the projected covariance measure learns directions sensitive to conditional-mean dependence \citep{lundborg2022projected}. Model-X procedures provide a complementary approach when the conditional distribution of the predictor is known or estimable \citep{candes2018panning,niu2024reconciling,shaer2023model}. These developments are closely related to the orthogonal-score perspective of double/debiased machine learning \citep{chernozhukov2018double}.

Several recent contributions address residual-based testing under dependence. \citet{flaxman2015gaussian} use Gaussian-process adjustment for dependent observations, while \citet{malinsky2019learning} consider nonstationary linear VARs with stochastic trends. Particularly close to our setting is the dynamic GCM of \citet{wieck2025conditional}, which uses time-varying regression and Gaussian calibration of residual partial sums for a single nonlinear, nonstationary trajectory, including Granger-type applications. Separately, \citet{chen2022debiased} show that, for independent observations, learner stability can justify debiased inference without sample splitting.

The proposed GTCM combines a Granger-history formulation with flexible residual adjustment, block-stability conditions for nuisance estimation on the full observed trajectory, predictable precision weighting, and calibrated aggregation over structured lag--feature subsets. The theory allows time-varying conditional laws and controls the additional terms arising from reuse of the same dependent trajectory for nuisance estimation and testing.

\subsection{Notation}
\label{sec:notation}

Let \(\{(X_t,Y_t,Z_t)\}_{t=1}^K\) denote the observed multivariate time series, with \(X_t,Y_t\in\mathbb R\) and \(Z_t\in\mathbb R^{d_Z}\). We use information-set notation to distinguish past, local and future observations in the dependence conditions below. Write \(H_t^{Y,Z}:=\{\bar Y_t,\bar Z_t\}\) for the history of \(Y\) and \(Z\) up to time \(t\), and \(H_t^{X,Y,Z}:=\{\bar X_t,\bar Y_t,\bar Z_t\}\); conditioning on these information sets is equivalent to conditioning on the corresponding history vectors. For future observations, let \(A_t^{Y,Z}:=\{Y_{t+1},\ldots,Y_K,Z_{t+1},\ldots,Z_K\}\) and \(A_t^{X,Y,Z}:=\{X_{t+1},\ldots,X_K,Y_{t+1},\ldots,Y_K,Z_{t+1},\ldots,Z_K\}\). Finally, for a finite window \([t,t+l]\), write \(Y_{t:(t+l)}:=(Y_t,\ldots,Y_{t+l})\) and \(L_{t:(t+l)}^{Y,Z}:=\{Y_{t:(t+l)},Z_{t:(t+l)}\}\), with analogous notation for other collections of processes.

\section{Method}

\subsection{The generalised temporal covariance measure}
\label{sec:gtcm}

We propose the \emph{generalised temporal covariance measure} (GTCM) for testing Granger-type conditional independence in a single realization of a dependent multivariate time series with time-varying conditional laws. Our starting point is the generalised covariance measure (GCM) of \citet{shah2020hardness}. For independent observations,
\begin{equation*}
U\independent V\given C
~\Rightarrow~
\E\left[
{U-\E(U\given C)}
{V-\E(V\given C)}
\right]
=0.
\end{equation*}
The GCM tests this moment condition using the empirical average of the product of residuals obtained by separately regressing \(U\) and \(V\) on \(C\). In the present setting, the null hypothesis is
\begin{equation}
H_0:
\quad
Y_{t+1}\independent \bar X_t
\given H_t^{Y,Z},
\label{eq:null_hyp-multi}
\end{equation}
where \(\bar X_t\) denotes the relevant history of the candidate source process \(X\), and \(H_t^{Y,Z}\) denotes the conditioning history. To construct a finite-dimensional moment condition, let \(\phi(\bar X_t)\in\mathbb R^d\) be a fixed dictionary of features of the source history. The dictionary may contain, for example, lagged values, interactions or nonlinear transformations.

Extending the GCM to this setting requires allowing the nuisance regressions to vary over time. In the independent-observation setting, the conditional-mean functions are common across observations, whereas here the relevant conditional moments may change with \(t\). We therefore define
\begin{equation}
f_t(H_t^{Y,Z})
=
\E[Y_{t+1}\given H_t^{Y,Z}],
\quad
g_t(H_t^{Y,Z})
=
\E[\phi(\bar X_t)\given H_t^{Y,Z}],
\label{eq:time_indexed_nuisances}
\end{equation}
with corresponding population residuals
\begin{equation}
\varepsilon_t
=
Y_{t+1}-f_t(H_t^{Y,Z}),
\quad
\boldsymbol{\xi}_t
=
\phi(\bar X_t)-g_t(H_t^{Y,Z}).
\label{eq:population_residuals}
\end{equation}
We take the feature dimension \(d\) to be fixed. The time indexing of \(f_t\) and \(g_t\) allows the relevant conditional means to evolve over the trajectory, so that such variation is accommodated under the null rather than interpreted as evidence against it.
Under \eqref{eq:null_hyp-multi}, conditional independence and the definitions
in \eqref{eq:population_residuals} imply
$$
    \E[
        \varepsilon_t\boldsymbol{\xi}_t
        \given H_t^{Y,Z}]
    =0.
$$
Consequently, for any suitable scalar weight \(W_t\) measurable with respect
to \(H_t^{Y,Z}\),
\begin{equation}
    \E\!\left[
        W_t\varepsilon_t\boldsymbol{\xi}_t
    \right]
    =0.
    \label{eq:vector_moment}
\end{equation}
Thus, under the null, the residual products have mean zero.
The weights determine the relative contribution of each time point;
Section~\ref{sec:weights} discusses their choice. In implementation, we replace the conditional means by estimates $\hat f_t$ and $\hat g_t$, giving
\begin{equation}
    \hat\varepsilon_t
    =
    Y_{t+1}-\hat f_t(H_t^{Y,Z}),
    \quad
    \hat{\boldsymbol{\xi}}_t
    =
    \phi(\bar X_t)-\hat g_t(H_t^{Y,Z}).
    \label{eq:estimated_residuals}
\end{equation}
Let \(\hat W_t\) estimate \(W_t\), with \(\hat W_t=W_t\)
when the weight is known. The GTCM score is
\begin{equation}
    T^{(K)}
    =
    \frac{1}{\sqrt K}
    \sum_{t=1}^K
    \hat W_t
    \hat\varepsilon_t
    \hat{\boldsymbol{\xi}}_t
    \in\mathbb R^d.
    \label{eq:vector_score}
\end{equation}

Each coordinate of \(T^{(K)}\) measures the remaining association between the outcome and one source feature after adjustment. Setting \(W_t=\hat W_t\equiv1\) gives the unweighted GTCM. With linear autoregressions for \(f_t\) and \(g_t\), this is a residual-covariance analogue of the classical Granger test. More flexible regressions can be used without changing the construction.

To studentise \(T^{(K)}\), we estimate its limiting covariance matrix by the
empirical covariance of the estimated score contributions,
\begin{equation}
\begin{split}
    \hat\Sigma
    =
    \frac{1}{K}\sum_{t=1}^K
    &\Big(
        \hat W_t\hat\varepsilon_t\hat{\boldsymbol{\xi}}_t
        -
        \frac{1}{K}\sum_{s=1}^K
        \hat W_s\hat\varepsilon_s\hat{\boldsymbol{\xi}}_s
    \Big) 
    \Big(
        \hat W_t\hat\varepsilon_t\hat{\boldsymbol{\xi}}_t
        -
        \frac{1}{K}\sum_{s=1}^K
        \hat W_s\hat\varepsilon_s\hat{\boldsymbol{\xi}}_s
    \Big)^{\top}.
    \label{eq:covariance_estimator}
\end{split}
\end{equation}
Under the conditions in Section~\ref{sec:theory}, $\hat\Sigma\xrightarrow{p}\Sigma$, where $\Sigma$ is the limiting covariance matrix of $T^{(K)}$. For fixed $d$, we aggregate these coordinates using
\begin{equation}
    Q^{(K)}
    =
    \{T^{(K)}\}^{\top}
    \hat\Sigma^{-1}
    T^{(K)},
    \label{eq:quadratic_stat}
\end{equation}
where \(\hat\Sigma\) estimates the null covariance matrix of \(T^{(K)}\). The quadratic form accounts for differences in scale and for dependence between the coordinates of the score.

The GCM argument does not extend to \eqref{eq:vector_score} without further justification. The residual products are temporally dependent and need not be identically distributed, while both nuisance regressions are estimated from the same trajectory used to form the test statistic. Standard sample splitting therefore does not yield independent training and evaluation samples. At the same time, fitting separate models using only the available conditioning history can be unstable at early time points and may discard substantial information. We therefore allow the nuisance estimators to be fitted using the full observed trajectory.

Section~\ref{sec:theory} establishes validity under moment conditions on the population residuals \(\varepsilon_t\) and \(\boldsymbol{\xi}_t\), conditions controlling temporal dependence, and stability conditions for the nuisance estimators. Under these conditions, replacing the population conditional means by their estimates is asymptotically negligible. The same conclusion holds when the predictable weights \(W_t\) are replaced by \(\hat W_t\), subject to an additional weight-estimation condition. Hence the test score in \eqref{eq:vector_score} has the same asymptotic null distribution as the corresponding oracle score constructed from the population residuals and weights.
\[
    T^{(K)}
    \xrightarrow{d}
    \mathcal N_d(0,\Sigma),
    ~
    Q^{(K)}
    \xrightarrow{d}
    \chi_d^2
\]
under \eqref{eq:null_hyp-multi}. The GTCM test therefore rejects the null at
asymptotic level \(\alpha\) when
\[
    Q^{(K)}>\chi^2_{d,1-\alpha}.
\]
The construction is summarised in Algorithm~\ref{alg:GTCM}.

\subsection{Precision weighting}
\label{sec:weights}

Any weight \(W_t\) measurable with respect to \(H_t^{Y,Z}\) preserves the null
moment restriction in \eqref{eq:vector_moment}; the corresponding regularity
conditions are given in Assumption~\ref{ass:population_weights}. Different
choices of \(W_t\), however, can lead to different variances and local power.
Weighting has also been considered for the GCM by
\citet{scheidegger2022weighted}, where functions of the conditioning
variables are used to reduce cancellation between conditional residual covariances of opposite signs. Our objective is instead to use predictable weights to improve efficiency under time-varying conditional response variance.

Such heteroscedasticity is common in financial, physiological and economic
time series. Without weighting, periods with large conditional residual
variance contribute on the same scale as more stable periods and may inflate
the variance of the aggregated score relative to their contribution to the
local mean shift. Motivated by weighted partial linear regression, we consider
the population precision weight
\begin{equation}
    W_t^\star
    =
    \big\{
        \E[
            \varepsilon_t^2
            \given H_t^{Y,Z}
        ]
    \big\}^{-1}.
    \label{eq:optimal_weight}
\end{equation}
Assume additionally that
\(\E[\varepsilon_t^2\given H_t^{Y,Z}]\geq c_3^2\) almost surely for
some \(c_3>0\), uniformly in \(t\). Together with the conditional
upper-moment bound in Assumption~\ref{ass:errors}(i), This ensures that \(W_t^\star\) satisfies the boundedness condition in Assumption~\ref{ass:population_weights}. For the local alternatives in Proposition~\ref{prop:vector-precision-optimal}, \(W_t^\star\) is the optimal scalar predictable weight when the conditional variance of \(\varepsilon_t\) is determined by \(H_t^{Y,Z}\). It therefore assigns greater weight to periods with smaller conditional response variance.

In practice, \(W_t^\star\) is replaced by an estimate \(\hat W_t\).
Appendix~\ref{app:weight-assumptions} gives conditions under which this
substitution is asymptotically negligible.

\subsection{Adaptive testing over feature subsets}
\label{subsec:max_statistic}

A rich feature dictionary allows sensitivity to delayed and nonlinear departures from the null, but including many inactive coordinates can reduce power. Conversely, restricting attention to a small dictionary risks excluding the relevant signal. We therefore consider a prespecified collection \(\mathcal S\) of structured, nonempty subsets of \([d]=\{1,\ldots,d\}\). For polynomial lag dictionaries, for example, the candidate subsets may vary jointly in lag depth and polynomial degree.

For each subset $\pmb{s}\in\mathcal S$, let $T_{\pmb{s}}^{(K)}$ denote the corresponding subvector of the GTCM score and let $\hat\Sigma_{\pmb{s}}$ denote the associated covariance submatrix. Define
\begin{equation}
    Q_{\pmb{s}}^{(K)}
    =
    (T_{\pmb{s}}^{(K)})^\top
    \hat\Sigma_{\pmb{s}}^{-1}
    T_{\pmb{s}}^{(K)} .
\end{equation}
Because the candidate subsets may have different dimensions, we place their quadratic statistics on a common scale using their asymptotic null tail probabilities,
\begin{equation}
    p_{\pmb{s}}^{(K)}
    = 1-
    F_{\chi^2_{|\pmb{s}|}}
    \big(
        Q_{\pmb{s}}^{(K)}
    \big).
\end{equation}
We then define
\begin{equation}
    M^{(K)}
    =
    \max_{\pmb{s}\in\mathcal S}
    \big(-\log p_{\pmb{s}}^{(K)}\big),
    \label{eq:max_evidence_stat}
\end{equation}
so that the test adapts to the subset yielding the strongest departure from the null after standardisation for subset dimension.

Calibration of \(M^{(K)}\) must account for both the adaptive subset search and the dependence among the subset statistics. Under the null, \(T^{(K)}\) is asymptotically \(\mathcal N_d(0,\Sigma)\), so Algorithm~\ref{alg:bootstrap_subset} approximates its joint null distribution by drawing from \(\mathcal N_d(0,\hat\Sigma)\). For each draw, the subset statistics and their maximum over \(\mathcal S\) are recomputed in the same way as for the observed data. Comparing the observed maximum with the resulting simulated distribution therefore calibrates the adaptive search while preserving the dependence among overlapping subsets.

The resulting procedure separates construction of the score from adaptation over the feature dictionary. Algorithm~\ref{alg:GTCM} forms the residual score for the conditional-independence null, whereas Algorithm~\ref{alg:bootstrap_subset} calibrates the maximum over candidate lag--basis subsets.

\begin{remark}[Connection with causal discovery in time series]
In a time-indexed causal graph, \eqref{eq:null_hyp-multi} corresponds to testing whether the source history contains additional information about \(Y_{t+1}\) after conditioning on \(H_t^{Y,Z}\). GTCM may therefore be used as a conditional-independence test within constraint-based procedures such as PC, with temporal ordering restricting the admissible edge directions. 
\end{remark}

\section{Theory}
\label{sec:theory}

We establish the null distribution and local power of the GTCM when the same observed trajectory is used for nuisance estimation and testing. Unlike standard cross-fitted i.i.d. arguments \citep{chernozhukov2018double}, temporal dependence means that sample splitting does not separate nuisance-training observations from the score contributions. Our theory therefore permits full-trajectory nuisance fitting and controls the resulting reuse through conditions on nuisance prediction error, temporal dependence, and stability of the fitted regressions under local block deletion. These conditions ensure that the effect of using the same dependent observations for nuisance estimation and testing is asymptotically negligible.

\subsection{Prediction errors and residual-weighted prediction}

We first impose conditions on the accuracy of the nuisance regressions. Because the GTCM score is formed from products of residuals, prediction error in one regression enters the score multiplied by the residual from the other. Ordinary prediction error alone is therefore insufficient to characterize the effect of nuisance estimation, particularly under temporal dependence. We therefore consider both ordinary and residual-weighted prediction errors. For \(q\in\{f,g\}\), define the empirical mean-squared prediction error and the corresponding residual-weighted prediction error by
\begin{equation*}
    A_q=\frac1K\sum_{t=1}^K
\big\|q_t(H_t^{Y,Z})-\hat q_t(H_t^{Y,Z})\big\|_2^2, ~B_g
=\frac1K\sum_{t=1}^K\varepsilon_t^2
\big\|g_t(H_t^{Y,Z})-\hat g_t(H_t^{Y,Z})\big\|_2^2.
\end{equation*}
Define \(B_f\) by replacing \(g\) with \(f\) and
\(\varepsilon_t^2\) with \(\|\boldsymbol\xi_t\|_2^2\).

\begin{assumption}[Prediction error and residual-weighted prediction]
\label{ass:errors}
The following conditions hold.
\begin{enumerate}
\item[(i)] There exist constants \(c_1,c_2>0\) and \(\delta>0\)
such that, uniformly in \(t\),
\[
\E\big[|\varepsilon_t|^{2+\delta}\given H_t^{X,Y,Z}\big]\le c_1^2~~\text{almost surely},
\quad
\E\big[\|\boldsymbol\xi_t\|_2^4\big]<c_2.
\]

\item[(ii)] For each \(q\in\{f,g\}\), the prediction errors satisfy
\begin{equation}
\label{eq:Kmsemse}
A_q=o_p(1),\quad
A_fA_g=o_p(K^{-1}),\quad
\E[A_q^2]=O(1).
\end{equation}

\item[(iii)] For each \(q\in\{f,g\}\), the residual-weighted prediction
errors satisfy
\begin{equation}
\label{eq:residual-weighted-errors}
(l_K+r_K)\E[B_q]=o(1), \quad \E[B_q^2]=O(1).
\end{equation}
\end{enumerate}
\end{assumption}

Part~(i) rules out degeneracy of the response residual and imposes mild
moment conditions while allowing its conditional variance to vary with the
history. In part~(ii), \(A_q=o_p(1)\) requires each nuisance regression to
be consistent in mean-squared prediction error. The product condition
\(A_fA_g=o_p(K^{-1})\) controls the leading second-order interaction between
the two nuisance estimation errors and allows their rates to trade off. For
example, it holds when both regressions have root mean-squared prediction
error \(o_p(K^{-1/4})\), while faster convergence of one permits slower
convergence of the other. This is the usual product-rate condition arising
in orthogonal and GCM-based inference \citep{shah2020hardness}.

The residual-weighted errors \(B_g\) and \(B_f\) reflect how nuisance estimation errors enter the score: an error in \(g_t\) is multiplied by \(\varepsilon_t\), whereas an error in \(f_t\) is multiplied by \(\boldsymbol\xi_t\). Condition~\eqref{eq:residual-weighted-errors} therefore controls prediction error after accounting for the corresponding residual contribution and the temporal dependence window, while the second-moment bound in the same condition controls larger residual-weighted
errors.

\subsection{Weak temporal dependence}

Prediction accuracy alone does not control dependence between score contributions across time. We therefore introduce backward and forward window lengths \(l_K,r_K\in\mathbb N^+\) to separate a local temporal neighbourhood around time \(t\) from observations treated as distant, with \(l_K+r_K=O(\log K)\). The following assumption requires sufficiently distant observations to contribute asymptotically negligible additional information and ensures stabilization of the second-order quantities governing the score.

\begin{assumption}[Weak temporal dependence]
\label{ass:mixing}
Let \(W_t\) satisfy Assumption~\ref{ass:population_weights}.
The following conditions hold.
\begin{enumerate}
\item[(i)] \textbf{Local block dependence.}
The local conditional-independence relations are
\[
\phi(\bar X_t)\indep H^{Y,Z}_{t-l_K}
\given L^{Y,Z}_{(t-l_K+1):t},\quad
Y_{t+1}\indep H^{Y,Z}_{t-l_K}
\given L^{Y,Z}_{(t-l_K+1):t}.
\]
\item[(ii)] \textbf{Distant information about the response residual.}
\begin{equation}
\label{eq:epsilon-mixing}
\sum_{t=1}^K
\E\big[
\E[\varepsilon_t\given H_t^{X,Y,Z},A_{t+r_K}^{X,Y,Z}]^2
\big]=o(K^{-1}).
\end{equation}

\item[(iii)] \textbf{Distant information about the source residual.}
\begin{equation}
\label{eq:xi-mixing}
\sum_{t=1}^K
\E\Big[
\max_{\substack{1\le s\le K:\\|s-t|\ge l_K+r_K}}
\big\|
\E\big[\boldsymbol\xi_t\given
H_t^{Y,Z},\phi(\bar X_s),A_{t+r_K}^{Y,Z}\big]
\big\|_2^2
\Big]=o(K^{-1}),
\end{equation}

\item[(iv)] \textbf{Variance stabilization.}
There exists a finite positive-definite matrix
\(\Sigma\in\mathbb R^{d\times d}\) such that
\begin{equation}
\label{eq:oracle-unconditional-covariance-limit}
\frac1K\sum_{t=1}^K
\E\big[W_t^2\varepsilon_t^2
\boldsymbol\xi_t\boldsymbol\xi_t^\top\big]
\longrightarrow\Sigma.
\end{equation}
\sloppy
Moreover for $\mathcal V_t
:=
W_t^2\E[\varepsilon_t^2\given H_t^{X, Y,Z}]
\operatorname{vec}(\boldsymbol\xi_t\boldsymbol\xi_t^\top),$ we have
\begin{equation}
\label{eq:2-mixing}
    \lim_{h\to\infty}
\sup_{t \geq 1}
{\left\|
\Cov(\mathcal V_t,\mathcal V_{t+h})
\right\|_F}
=0 .
\end{equation}
\end{enumerate}
\end{assumption}

Under \(H_0\), the oracle score contributions satisfy \(\E[W_t\varepsilon_t\boldsymbol\xi_t\given H_t^{X,Y,Z}]=0\) and therefore form a martingale-difference sequence. Part~(iv), together with Assumption~\ref{ass:errors}, ensures stabilization of the score covariance and the regularity required for the martingale central limit theorem \citep{dvoretzky1972asymptotic}. Parts~(i)--(iii) instead control the dependence introduced by nuisance estimation on the full trajectory. Part~(i) uses \(l_K\) to isolate dependence on the remote past through a local backward block, while parts~(ii)--(iii) use \(r_K\) to ensure that sufficiently distant future observations contain asymptotically negligible additional information about the current residuals; the remaining local dependence is handled by the stability conditions below.

\begin{remark}
\label{rem:example-mixing}
Appendix~\ref{app:mixing-example} verifies conditions~(ii) and~(iii) for two
standard dependent processes: a jointly stationary Gaussian AR(\(p\)) process
and a joint MA(\(q\)) process for \(\{X_t,Y_t,Z_t\}_{t=1}^K\). These examples
show that the assumptions are satisfied by familiar classes of linear time
series.
\end{remark}

\subsection{Stability of the nuisance regressions}
\label{sec:stability}

The weak dependence conditions control the influence of observations sufficiently separated from a given score contribution, but nearby observations may remain strongly dependent and are reused in nuisance estimation. Such reuse can induce additional dependence between the fitted nuisance functions and the score. We therefore impose stability of the regression algorithms under deletion of local temporal blocks, using the separation length \(r_K\) introduced above.

For \(t=1,\ldots,K\), define the forward deletion block
\begin{equation}
    \Delta_t=\{t+1,\ldots,t+r_K\}\cap\{1,\ldots,K+1\}.
\end{equation}
For \(q\in\{f,g\}\), let \(\hat q^{-t}\) and
\(\hat q^{-t,t'}\) denote fits obtained after deleting the observations
in \(\Delta_t\) and \(\Delta_t\cup\Delta_{t'}\), respectively.
Deletion removes every training row whose response or predictors use a
deleted observation. For randomized learners, the full and deleted fits
are constructed using the same underlying randomization, so that their
difference reflects only the change in the training sample. In particular,
\(\hat q^{-t,t}=\hat q^{-t}\). Define
\begin{align*}
S_q
&=\frac1K\sum_{t=1}^K
\big\|\hat q_t(H_t^{Y,Z})-\hat q_t^{-t}(H_t^{Y,Z})\big\|_2^2,
~D_q
=\frac1K\sum_{t=1}^K\max_{1\le t'\le K}
\big\|\hat q_t^{-t}(H_t^{Y,Z})
      -\hat q_t^{-t,t'}(H_t^{Y,Z})\big\|_2^2.
\end{align*}
For the source regression, define the residual-weighted deletion errors
\begin{align*}
\tilde S_g
&=\frac1K\sum_{t=1}^K\varepsilon_t^2
\big\|\hat g_t(H_t^{Y,Z})-\hat g_t^{-t}(H_t^{Y,Z})\big\|_2^2,
~\tilde D_g=\frac1K\sum_{t=1}^K\varepsilon_t^2
\max_{1\le t'\le K}
\big\|\hat g_t^{-t}(H_t^{Y,Z})
      -\hat g_t^{-t,t'}(H_t^{Y,Z})\big\|_2^2.
\end{align*}
Define \(\tilde S_f,\tilde D_f\) by replacing \(g\) with \(f\) and
\(\varepsilon_t^2\) with \(\|\boldsymbol\xi_t\|_2^2\).
For the scalar regression \(f\), the Euclidean norm reduces to absolute
value. Thus \(S_q\) measures the average squared change in the fitted value
at \(H_t^{Y,Z}\) after deleting \(\Delta_t\), while \(D_q\) measures the
largest additional change induced by a second block deletion.
The quantities \(\tilde S_q\) and \(\tilde D_q\) give the corresponding
changes after weighting by the residual factor appearing in the score.
We suppress their dependence on \(K\).

\begin{assumption}[Block stability of nuisance estimators]
\label{ass:stability}
For each \(q\in\{f,g\}\),
\begin{equation}
\label{eq:stability}
S_q=o_p(K^{-1}),\quad
(l_K+r_K)\E[\tilde S_q]+K\E[\tilde D_q]=o(1),
\end{equation}
and
\begin{equation}
\label{eq:stability-moments}
\E[S_q^2+D_q^2+\tilde S_q^{\,2}+\tilde D_q^{\,2}]
=O(1).
\end{equation}
\end{assumption}

Assumption~\ref{ass:stability} is a block-deletion analogue of algorithmic
stability. For independent data, leave-one-out stability controls the effect
of reusing observations for nuisance estimation and inference
\citep{barber2021predictive,chen2022debiased}; here a local block is removed
because nearby observations may remain strongly dependent. The condition
\(S_q=o_p(K^{-1})\) requires the average change in fitted values induced by
the first deletion to be negligible on the scale relevant for the score.

The residual-weighted terms distinguish the two deletion steps. In
\(\tilde S_q\), the full fit may use observations from the local block
\(\Delta_t\), and can therefore remain more strongly associated with the
residual at time \(t\). In \(\tilde D_q\), both fits already exclude
\(\Delta_t\) and differ only through an additional deletion, providing
greater temporal separation. The factors \(l_K+r_K\) and \(K\) reflect,
respectively, the local and separated pairs arising in the score expansion.
The second-moment condition \eqref{eq:stability-moments} controls occasional
larger deletion perturbations.

Temporal separation and algorithmic stability are complementary. Increasing
\(r_K\) reduces the dependence remaining outside the deleted block, while
requiring stability after deletion of a larger local neighbourhood. The
resulting condition is related to stability under nested training samples
\citep[Definitions~3.1 and~3.4]{liang2025algorithmic}, but is formulated here
for contiguous block deletions and averaged squared changes in fitted values.

\begin{remark}
\label{rem:example-stability}
Appendix~\ref{app:stability-example} verifies \eqref{eq:stability} for ordinary least squares and kernel time-varying coefficient regression with \(O(\log K)\) adjustment lags. The results hold under the stated moment and Gram-matrix conditions, together with bandwidth and local regularity conditions for the kernel estimator.
\end{remark}

\subsection{Null distribution and size control}
\label{sec:null_limits}

The preceding conditions give the following justification
for the chi-squared test in Section~\ref{sec:gtcm}.

\begin{theorem}[Null distribution of the GTCM statistic]
\label{thm:vector_gtcm}
Suppose that the null hypothesis \eqref{eq:null_hyp-multi} holds, \(d\) is
fixed, and Assumptions~\ref{ass:mixing}, \ref{ass:stability},
\ref{ass:errors} and \ref{ass:estimated_weights} hold. Then the score and the
covariance estimator in \eqref{eq:covariance_estimator} satisfy
\begin{equation}
    \label{eq:martingaleCLT}
    T^{(K)}
    \xrightarrow{d}
    \mathcal N_d(0,\Sigma),\quad
    \hat\Sigma\xrightarrow{p}\Sigma.
\end{equation}
Consequently,
\begin{equation}
    Q^{(K)}
    =
    \{T^{(K)}\}^{\top}
    \hat\Sigma^{-1}
    T^{(K)}
    \xrightarrow{d}
    \chi_d^2.
    \label{eq:martingaleCLT2}
\end{equation}
In particular, for every \(\alpha\in(0,1)\),
\begin{equation}
    \Pr\left\{
        Q^{(K)}>\chi^2_{d,1-\alpha}
    \right\}
    \rightarrow
    \alpha.
    \label{eq:gtcm_size}
\end{equation}
\end{theorem}

Theorem~\ref{thm:vector_gtcm} shows that, under the stated conditions, estimating the conditional means and precision weights from the full trajectory does not affect the limiting null distribution. Thus, valid asymptotic inference can be obtained without sample splitting or sequential refitting of the nuisance regressions.

The adaptive procedure in Section~\ref{subsec:max_statistic} additionally maximises over lag--feature subsets, so the chi-squared limit no longer applies directly. Its Monte Carlo calibration reproduces this selection step for each simulated score, and Proposition~\ref{prop:adaptive_calibration} establishes the asymptotic validity of the test.

\subsection{Local alternatives and asymptotic power}
\label{subsec:local_power}

For fixed alternatives with a non-zero limiting residual
moment, the signal in \(T^{(K)}\) grows at order \(\sqrt K\) and power
tends to one. To describe sensitivity to weaker signals, we consider
alternatives whose magnitude decreases at rate \(K^{-1/2}\). At this
scale, signal and sampling variability remain comparable, giving a
non-trivial limiting power curve.

Consider a response with a baseline component
\(m_t(H_t^{Y,Z})\) and a small additional contribution from the source
features:
\begin{equation*}
Y_{t+1}
=
m_t(H_t^{Y,Z})
+
\frac{\tau^\top\phi(\bar X_t)}{\sqrt K}
+
\zeta_t,
\quad
\E[\zeta_t\given H_t^{Y,Z},\bar X_t]=0,
\end{equation*}
where \(\tau\in\mathbb R^d\) fixes the direction and relative strength
of the source contribution. Subtracting the conditional mean
\(f_t(H_t^{Y,Z})\) removes the baseline and the part of the source
contribution already predictable from the adjustment history, giving
\begin{equation*}
\varepsilon_t
=
Y_{t+1}-f_t(H_t^{Y,Z})
=
\frac{\tau^\top\{\phi(\bar X_t)-g_t(H_t^{Y,Z})\}}{\sqrt K}
+\zeta_t
=
\frac{\boldsymbol\xi_t^\top\tau}{\sqrt K}+\zeta_t.
\end{equation*}
This data-generating process implies the sequence of local alternatives
\begin{equation}
\label{eq:local_alt}
    H_{1,K}(\tau):
    ~
    \E[
        \varepsilon_t
        \given
        H_t^{Y,Z},\bar X_t
    ]
    =
    \frac{\boldsymbol{\xi}_t^\top\tau}{\sqrt K}.
\end{equation}
The residual source component therefore shifts the mean of
\(T^{(K)}\) by an amount that remains of constant order as \(K\) grows.
Since the alternative changes with \(K\), the population quantities also
depend on \(K\); this dependence is suppressed in the notation.
\begin{theorem}[Local asymptotic power]
\label{thm:local_power}
Fix \(\tau\in\mathbb R^d\) and consider the sequence
\(H_{1,K}(\tau)\) in \eqref{eq:local_alt}. Suppose, uniformly along this
sequence, that Assumptions~\ref{ass:errors}, \ref{ass:mixing},
\ref{ass:stability}, and~\ref{ass:estimated_weights} continue to hold with
\(\zeta_t\) in place of \(\varepsilon_t\) wherever applicable. Then
\begin{equation}
\label{eq:local_score_expansion}
T^{(K)}
=
\frac{1}{\sqrt K}\sum_{t=1}^K
W_t\zeta_t\boldsymbol{\xi}_t
+
\Big(
\frac{1}{K}\sum_{t=1}^K
W_t\boldsymbol{\xi}_t\boldsymbol{\xi}_t^\top
\Big)\tau
+
o_p(1).
\end{equation}
Suppose also that \eqref{eq:local-drift-mean-limit} holds
for a finite matrix \(\Gamma\). Proposition~\ref{prop:Sigma_xi} then gives \(\Gamma=\plim_{K\to\infty}K^{-1}\sum_{t=1}^K
W_t\boldsymbol\xi_t\boldsymbol\xi_t^\top\). Then
\begin{equation}
\label{eq:local_score_limit}
T^{(K)}
\xrightarrow{d}
\mathcal N_d(\Gamma\tau,\Sigma),\quad
\hat\Sigma\xrightarrow{p}\Sigma.
\end{equation}
Consequently,
\begin{equation}
\label{eq:local_chisq_limit}
Q^{(K)}
\xrightarrow{d}
\chi_d^2\{\lambda(\tau)\},
\quad
\lambda(\tau)
=
\tau^\top\Gamma^\top\Sigma^{-1}\Gamma\tau.
\end{equation}
Therefore, for every $\alpha\in(0,1)$,
\begin{equation}
\label{eq:local_power_curve}
\lim_{K\to\infty}
\Pr_{H_{1,K}(\tau)}
\big(
Q^{(K)}>\chi^2_{d,1-\alpha}
\big)
=
\Pr \big(
\chi_d^2\{\lambda(\tau)\}>
\chi^2_{d,1-\alpha}
\big).
\end{equation}
In particular, if \(\Gamma\tau\neq0\), then
\(\lambda(\tau)>0\), and the limiting power is strictly greater than
\(\alpha\).
\end{theorem}

The expansion in \eqref{eq:local_score_expansion} separates the stochastic
fluctuation of the centred oracle score from the deterministic local mean
shift. The matrix \(\Gamma\) maps a departure represented by \(\tau\) into
the limiting residual moment, so that $\Gamma\tau\neq 0$ is precisely the detectability condition for the chosen feature dictionary. 

The noncentrality parameter \(\lambda(\tau)\) also describes the role of weighting. The weights affect both the mean shift through \(\Gamma\) and the variability of the score through \(\Sigma\); their effect on power is therefore determined by the resulting signal-to-noise ratio. The adaptive statistic addresses a distinct source of power loss. When the departure is concentrated in a small subset of lag--feature directions, aggregation over the full dictionary may include many coordinates carrying little signal. Maximisation over structured subsets is intended to retain sensitivity to such alternatives, with its null calibration justified in Proposition~\ref{prop:adaptive_calibration}.

\section{Simulation study}
\label{sec:simulation}

The simulations separate two potential sources of power gain: precision weighting under heteroscedasticity, and adaptive aggregation when the signal concentrates only on a small part of a broad dictionary. We consider three data-generating processes that vary the complexity of the outcome history and the timing and form of the exposure effect. Every scenario uses 1000 Monte Carlo replications, \(K=2500\), and nominal level \(0.05\). We omit sequence $Z$ here for simplicity, so the null hypothesis is
\begin{equation}
\label{eq:sim-null}
H_0:~
Y_{t+1} \indep \bar X_t
\given \bar Y_t.
\end{equation}
The primary design combines a nonlinear shared history with a
linear exposure effect at lag 1:
\begin{align}
Y_t={}&0.45Y_{t-1}-0.10Y_{t-2}
 +0.15\tanh(2Y_{t-1}Y_{t-2})
 +\beta X_{t-1}+\varepsilon_t,
\label{eq:sim-shared-y}\\
X_t={}&0.20X_{t-1}-0.05X_{t-2}
 +0.45\tanh(2Y_tY_{t-1})+u_t.
\label{eq:sim-shared-x}
\end{align}
Under the null (i.e., $\beta=0$), the same nonlinear function of the outcome history predicts both \(X_t\) and \(Y_{t+1}\). The design therefore requires the nuisance regressions to account for this shared dependence: if it is inadequately learned, the residuals remain correlated and can lead to spurious rejection of $H_0$.

The innovations are independent, with
\(u_t\sim N(0,0.20)\) and, conditional on the past,
\(\varepsilon_t\sim N(0,\sigma_t^2)\). The variance \(\sigma_t^2\) changes
with the state of the series, and
\(R_\sigma=\max_t\sigma_t^2/\min_t\sigma_t^2\in\{4,8,16\}\).
We set \(\beta=\delta/\sqrt K\), where
\(\delta\in\{0,2.5,3.5\}\) indexes the strength of the alternative and
\(\delta=0\) gives the null. Further details appear in Supplementary
Appendix~
\ref{sec:simulation-common-details}.

For the primary design, the specified-alternative benchmark uses only the active linear term \(X_{t-1}\). We compare this with tests based on a 40-dimensional dictionary comprising polynomial terms of degrees one to five over eight exposure lags. Each procedure is implemented both unweighted and with estimated precision weighting. Generalised random forests \citep{athey2019generalized} are used to estimate the nuisance regressions and precision weights. The adaptive procedure selects among candidate subsets of $\phi(\bar X_t) = \{ X_t^r, X_{t-1}^r,\ldots,X_{t-7}^r: r=1,\ldots,5\},$ where each subset contains all terms up to a given lag depth and polynomial degree, corresponding to a leading rectangle of the lag--degree grid. Calibration uses 1999 Gaussian draws. further implementation details are given in Supplementary Appendix~\ref{sec:simulation-common-details}.

Figure~\ref{fig:simulation-shared-combined}(a) illustrates both sources of power gain. The 40-dimensional test loses power by aggregating over many inactive coordinates, much of which is recovered by the adaptive procedure. Precision weighting yields a further gain that increases with \(R_\sigma\). All six procedures are well calibrated under the null.

For comparison with commonly used Granger-causality procedures, we benchmark the proposed GTCM tests against the classical linear Granger \(F\)-test \citep{granger1969investigating}, a linear VAR Granger test \citep{lutkepohl2005new}, an HC3-robust dynamic linear Wald test \citep{mackinnon1985some}, and the weighted GCM with GAM nuisance regressions \citep{scheidegger2022weighted}. Figure~\ref{fig:simulation-shared-combined}(b) uses the data-generating process in \eqref{eq:sim-shared-y}--\eqref{eq:sim-shared-x} with the mild variance ratio \(R_\sigma=4\). The proposed procedures maintain nominal size, whereas the competing methods exhibit substantial over-rejection under the null. Their rejection rates under the alternative therefore do not represent size-controlled power.

\begin{figure}[htbp]
\centering
\includegraphics[width=0.99\linewidth]
{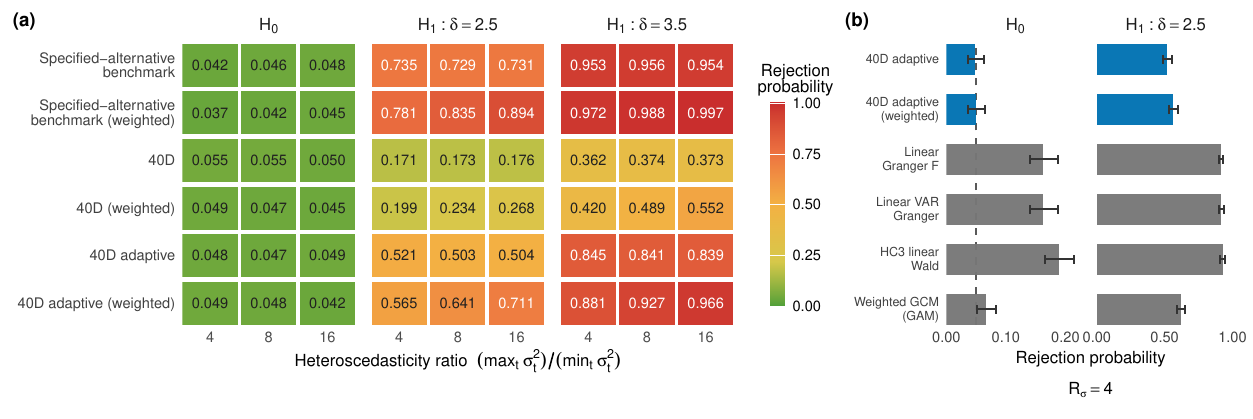}
\caption{Rejection rates for the nonlinear shared-history
design, based on 1000 Monte Carlo replications. Panel (a) shows the six
proposed procedures; columns within each
block give \(R_\sigma=4,8,16\). The specified-alternative benchmark uses the
known linear lag-1 coordinate, 40D uses the full dictionary, and
``weighted'' denotes estimated precision weighting. Panel (b) compares the
two adaptive procedures with four competitors at \(R_\sigma=4\), under
\(H_0\) and under \(H_1\) with \(\delta=2.5\). In panel (b), whiskers are
exact 95\% binomial confidence intervals, and the dashed line in the null
panel marks the nominal level.}
\label{fig:simulation-shared-combined}
\end{figure}

To assess robustness beyond the primary design, we consider two additional data-generating processes. The nonlinear distributed-signal design spreads the exposure effect over lags $1--3$ and uses a nonlinear signal that is not represented exactly by the polynomial dictionary, thereby assessing power under feature misspecification. The delayed linear-signal design places the effect only at lags \(4\) and \(5\), so that the adaptive procedure must search beyond several earlier lags containing no signal. Both designs otherwise follow the primary experiment. Full specifications and additional results are given in Supplementary Appendix~\ref{sec:simulation-additional-results}.

Figure~\ref{fig:simulation-additional-dgps} reports results for the two additional designs at \(R_\sigma=16\). The results reinforce the main findings from the primary experiment: both precision weighting and adaptive aggregation yield clear power gains across these alternative signal structures, while all proposed procedures remain well calibrated under the null. Results across the full \(R_\sigma\)-grid are given in Supplementary Appendix~\ref{sec:simulation-additional-results}.

\begin{figure}[htbp]
\centering
\includegraphics[width=0.98\linewidth]
{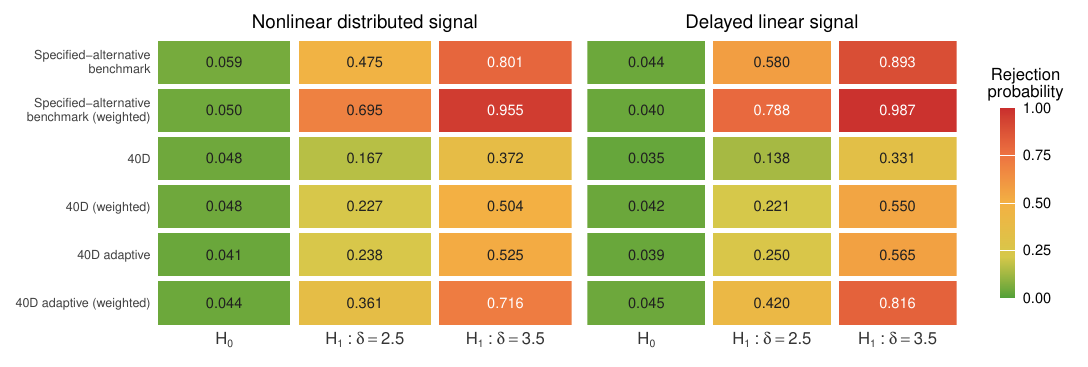}
\caption{Empirical rejection rates for the nonlinear
distributed-signal and delayed linear-signal designs at \(R_\sigma=16\),
based on 1000 Monte Carlo replications. The specified-alternative benchmark
uses the active signal coordinates for each design, whereas the 40D
procedures use the full lag--degree dictionary. ``Weighted'' denotes
estimated precision weighting.}
\label{fig:simulation-additional-dgps}
\end{figure}

\section{Wearable activity and short-term glucose dynamics}
\label{sec:ohiot1dm}

We apply GTCM to longitudinal glucose and wearable data from individuals with Type I diabetes. Physical activity is associated with short-term changes in glucose, but it is less clear whether wearable activity measurements provide information beyond that already contained in routinely recorded clinical and behavioural history. We therefore test whether recent wearable-measured activity is informative about the next five-minute glucose change after conditioning on recent glucose, insulin delivery, carbohydrate intake, sleep, and time of day.

We analyse the six participants in each of the 2018 and 2020 releases of the OhioT1DM study \citep{marling2018ohiot1dm,marling2020ohiot1dm}. Each participant was followed for eight weeks while using an insulin pump and a continuous glucose monitor (CGM). The activity measurements differ importantly between releases: the 2018 participants wore a Basis Peak fitness band that recorded five-minute step counts, whereas the 2020 participants wore an Empatica device that recorded acceleration rather than step counts. For the latter cohort, we use five-minute mean acceleration magnitude constructed from the available one-minute measurements. The two cohorts are therefore analysed separately. Let $G_t$ denote the CGM measurement and let $X_t$ denote the corresponding cohort-specific activity measurement at five-minute grid time \(t\). Define
\begin{equation}
    Y_t = G_t-G_{t-1},
\end{equation}
whenever the two consecutive CGM measurements are observed, so that $Y_{t+1}$ represents the subsequent five-minute glucose change. Insulin, meal, and sleep records are aligned to the same grid before the histories are constructed.

The conditioning history includes recent glucose dynamics, insulin delivery, carbohydrate intake, sleep and circadian timing. Activity history is represented by a 40-dimensional polynomial dictionary constructed from the current activity measurement and recent lags. After removing observations with incomplete required histories, the participant-specific sample sizes range from $K=11{,}484$ to $13{,}888$ in 2018 and from $K=3{,}231$ to $10{,}194$ in 2020.

For cohort-level inference, we combine the six participant-specific \(p\)-values using Fisher's method for each GTCM construction. In the 2018 cohort, the combined \(p\)-values are \(3.79\times10^{-37}\), \(1.07\times10^{-58}\), \(5.19\times10^{-17}\), and \(4.71\times10^{-18}\) for the full unweighted, full precision-weighted, adaptive unweighted, and adaptive precision-weighted procedures, respectively; the corresponding values in 2020 are \(0.001851\), \(0.000280\), \(0.000239\), and \(1.61\times10^{-5}\). We do not report comparisons with standard Granger procedures, as Figure~\ref{fig:simulation-shared-combined}(b) shows that several such methods can be miscalibrated under nonlinear shared histories.

We further assess the sensitivity of these findings to measurement noise by adding independent Gaussian noise to the glucose-change response and repeating the participant- and cohort-level analyses over $500$ repetitions. Table~\ref{tab:ohiot1dm_noise} shows a progressive attenuation of the evidence as the noise level increases, with substantially greater robustness in the 2018 cohort. At c=2, the cohort-level test rejects in 97.2\% of repetitions for 2018 and 10.0\% for 2020; under the strongest contamination, both rejection rates approach the nominal level. This analysis quantifies how rapidly the empirical evidence deteriorates as the response becomes less informative.

\begin{table}[htbp]
\centering
\caption{Cohort-level evidence and attenuation under noise added to the glucose-change response}
\label{tab:ohiot1dm_noise}
\begin{threeparttable}
\begin{tabular}{lcccc}
\toprule
& \multicolumn{2}{c}{2018 Basis cohort}
& \multicolumn{2}{c}{2020 Empatica cohort} \\
\cmidrule(lr){2-3}\cmidrule(lr){4-5}
Noise multiplier \(c\)
& Rejection rate & Median \(p\)
& Rejection rate & Median \(p\) \\
\midrule
Observed data
    & --    & \(4.71\times10^{-18}\)
    & --    & \(1.61\times10^{-5}\) \\
1
    & 1.000 & \(1.15\times10^{-11}\)
    & 0.326 & 0.115 \\
2
    & 0.972 & \(1.42\times10^{-4}\)
    & 0.100 & 0.327 \\
4
    & 0.326 & 0.121
    & 0.064 & 0.480 \\
8
    & 0.080 & 0.397
    & 0.048 & 0.528 \\
\bottomrule
\end{tabular}

\begin{tablenotes}[flushleft]
\footnotesize
\item \textit{Note:}
The observed-data row reports the single cohort-level global-null test, for
which a rejection rate is not defined. For \(c>0\), the rejection rate is
the proportion of 500 noise repetitions in which the cohort-level global null
is rejected at level 0.05, and the median is taken over the corresponding
500 Fisher-combined \(p\)-values. The multiplier \(c\) denotes the standard
deviation of the added noise relative to the observed standard deviation of
the glucose-change response.
\end{tablenotes}
\end{threeparttable}
\end{table}

Both the observed-data and noise-contamination analyses concern incremental temporal information conditional on the recorded history, rather than a causal effect of physical activity on glucose. Supplementary Appendix~\ref{sec:ohiot1dm_supplement} reports participant-level results and two additional checks. Analyses over randomly located windows containing \(20\%\), \(40\%\), \(60\%\) and \(80\%\) of each record assess how the evidence accumulates with increasing observation length and whether it is confined to particular periods. Independent pseudo-exposures provide a data-based calibration check, assessing whether the proposed procedure maintains Type-I error control in these observed time series when the null holds by construction.

\section*{Supplementary material}
\label{SM}

{Appendices~\ref{app:null-proof} and
\ref{app:local-power-proof} contain the proofs of
Theorems~\ref{thm:vector_gtcm} and \ref{thm:local_power}, respectively.
The examples in Remarks~\ref{rem:example-mixing} and
\ref{rem:example-stability} are developed in
Appendices~\ref{app:mixing-example} and \ref{app:stability-example}, with
additional technical results in Appendices~\ref{app:power-results} and
\ref{app:complementary-lemmas}. Simulation details and comparisons are given
in Appendix~\ref{app:implementation}, and additional OhioT1DM analyses in
Appendix~\ref{sec:ohiot1dm_supplement}. Code for the simulation studies and
the real-data analysis is available at~\url{https://github.com/Herashi/Granger-causality}.}

\appendix

\section{Algorithms}

\begin{algorithm}[htbp]
\caption{Generalised temporal covariance measure}
\label{alg:GTCM}
\begin{algorithmic}[1]

\Require $\{(Y_t,X_t,Z_t):t=1,\ldots,K+1\}$, 
         $H_t^{Y,Z}$, $\phi$, $\alpha$;
         weighting option (unweighted or precision);
         bounds $0<c_W<C_W<\infty$

\State Fit $\hat f_t$ and $\hat g_t$ using all $K$ observations

\For{$t=1,\ldots,K$}
    \State $\hat\varepsilon_t
        \gets Y_{t+1}-\hat f_t(H_t^{Y,Z})$
    \State $\hat{\boldsymbol\xi}_t
        \gets \phi(\bar X_t)-\hat g_t(H_t^{Y,Z})$
\EndFor

\If{the weighting option is unweighted}
    \State $\hat W_t\gets 1$, $t=1,\ldots,K$
\Else
    \State Fit
    $\hat\sigma_t^2
        \gets \widehat{\E}(\hat\varepsilon_t^2\given H_t^{Y,Z})$
    \Statex \hspace{\algorithmicindent}
        using all $K$ observations
    \State $\hat W_t
        \gets
        \{\max(C_W^{-1},
        \min(c_W^{-1},\hat\sigma_t^2))\}^{-1}$,
        $t=1,\ldots,K$
\EndIf

\State $T^{(K)}
    \gets K^{-1/2}
    \sum_{t=1}^K
    \hat W_t\hat\varepsilon_t\hat{\boldsymbol\xi}_t$

\State Compute $\hat\Sigma$ using \eqref{eq:covariance_estimator}

\State $Q^{(K)}
    \gets
    (T^{(K)})^\top
    \hat\Sigma^{-1}
    T^{(K)}$

\State Reject $H_0$ if
    $Q^{(K)}>\chi^2_{d,1-\alpha}$

\Ensure $T^{(K)}$, $\hat\Sigma$, and $Q^{(K)}$

\end{algorithmic}
\end{algorithm}

\begin{algorithm}[htbp]
\caption{Monte Carlo calibration of the adaptive maximum GTCM}
\label{alg:bootstrap_subset}
\begin{algorithmic}[1]

\Require $T^{(K)}$, $\hat\Sigma$,
         $\mathcal S\subseteq 2^{\{1,\ldots,d\}}$, and $B$

\For{each $\pmb{s}\in\mathcal S$}
    \State Extract
        $T_{\pmb{s}}^{(K)}$
        and $\hat\Sigma_{\pmb{s}}$
    \State $Q_{\pmb{s}}^{(K)}
        \gets
        (T_{\pmb{s}}^{(K)})^\top
        \hat\Sigma_{\pmb{s}}^{-1}
        T_{\pmb{s}}^{(K)}$
    \State $p_{\pmb{s}}^{(K)}
        \gets
        1-F_{\chi^2_{|\pmb{s}|}}
        \{Q_{\pmb{s}}^{(K)}\}$
\EndFor

\State $M^{(K)}
    \gets
    \max_{\pmb{s}\in\mathcal S}
    \{-\log p_{\pmb{s}}^{(K)}\}$

\State $\hat{\pmb{s}}
    \gets
    \arg\max_{\pmb{s}\in\mathcal S}
    \{-\log p_{\pmb{s}}^{(K)}\}$

\For{$b=1,\ldots,B$}
    \State Draw $T_b\sim\mathcal N_d(0,\hat\Sigma)$
    \For{each $\pmb{s}\in\mathcal S$}
        \State Extract
            $T_{b,\pmb{s}}$
            and $\hat\Sigma_{\pmb{s}}$
        \State $Q_{b,\pmb{s}}
            \gets
            (T_{b,\pmb{s}})^\top
            \hat\Sigma_{\pmb{s}}^{-1}
            T_{b,\pmb{s}}$
        \State $p_{b,\pmb{s}}
            \gets
            1-F_{\chi^2_{|\pmb{s}|}}
            \{Q_{b,\pmb{s}}\}$
    \EndFor
    \State $M_b
        \gets
        \max_{\pmb{s}\in\mathcal S}
        \{-\log p_{b,\pmb{s}}\}$
\EndFor

\State $p_{\mathrm{adj}}
    \gets
    (B+1)^{-1}
    \left\{
        1+\sum_{b=1}^B
        \mathbf 1(M_b\geq M^{(K)})
    \right\}$

\Ensure $p_{\mathrm{adj}}$ and $\hat{\pmb{s}}$

\end{algorithmic}
\end{algorithm}

\section{Conditions on the weights}
\label{app:weight-assumptions}

We distinguish regularity of the population weighting rule from accuracy and
stability of its estimator. The former ensures that weighting preserves a
well-behaved score, while the latter controls the additional error introduced
when the weights are estimated from the same trajectory.

\begin{assumption}[Population weight regularity]
\label{ass:population_weights}
The weights \(W_t\) satisfy the following conditions.
\begin{enumerate}
    \item[(i)] There exist constants \(0<c_W<C_W<\infty\) such that, for all
    \(t\), \(W_t\) is \(H_t^{Y,Z}\)-measurable and
    \(c_W\leq W_t\leq C_W\) almost surely.

    \item[(ii)] For
\(\mathcal V_t
=W_t\operatorname{vec}(\boldsymbol\xi_t\boldsymbol\xi_t^\top)\),
condition~\eqref{eq:2-mixing} holds.
\end{enumerate}
\end{assumption}

Assumption~\ref{ass:population_weights} requires the predictable weights to
remain bounded away from zero and infinity, and imposes the same weak
second-order dependence used for covariance stabilization.

For the estimated weights, block deletion is defined as in
Section~\ref{sec:stability}. The quantities \(\hat W_t^{-t}\) and
\(\hat W_t^{-t,t'}\) are obtained by refitting the complete weight-estimation
procedure after deleting the corresponding blocks, with
\(\hat W_t^{-t,t}=\hat W_t^{-t}\). All estimates are evaluated at the original
history \(H_t^{Y,Z}\). For precision weighting, this refitting includes the
outcome regression, the residuals used to fit the conditional variance, and
any associated preprocessing or tuning.

Use the definitions of \(A_q,B_q,S_q,D_q,\tilde S_q,\tilde D_q\)
in Section~\ref{sec:theory} with \(q=W\), taking
\(W_t,\hat W_t,\hat W_t^{-t},\hat W_t^{-t,t'}\) as the corresponding
evaluated values. For \(B_W,\tilde S_W,\tilde D_W\), the
residual multiplier is \(\varepsilon_t^2\|\boldsymbol\xi_t\|_2^2\).
Thus \(A_W,S_W,D_W\) measure ordinary weight error and deletion stability,
whereas \(B_W,\tilde S_W,\tilde D_W\) measure the same errors after weighting
by the squared norm of the oracle score contribution.

\begin{assumption}[Estimated weight accuracy and stability]
\label{ass:estimated_weights}
Let \(W_t\) satisfy Assumption~\ref{ass:population_weights}. With the
windows \(l_K,r_K\) of Assumption~\ref{ass:mixing}, the following
conditions hold.
\begin{enumerate}
    \item[(i)] \textit{Boundedness and information.} Almost surely, for
    every \(K\) and \(1\leq t,t'\leq K\),
    \[
    c_W\leq\hat W_t,\,\hat W_t^{-t},\,\hat W_t^{-t,t'}\leq C_W.
    \]
    The weight procedure uses response and adjustment observations.
    Whenever \(|s-t|\geq l_K+r_K\), the weights
    \(W_s,W_t,\hat W_s^{-s,t},\hat W_t^{-t},\hat W_t^{-s,t}\)
    are measurable with respect to
    \(H_t^{Y,Z},A_{t+r_K}^{Y,Z}\), together with any independent
    fitting randomness.

    \item[(ii)] \textit{Accuracy and block stability.} The fitted weights satisfy
    \begin{equation}
    \label{eq:weight-stability-moments}
    A_W=o_p(1),\quad S_W=o_p(K^{-1}), \quad \E[D_W^2]=o(K^{-2}).
    \end{equation}

    \item[(iii)] \textit{Residual-weighted accuracy.}
    The fitted weights satisfy
    \begin{equation}
    \label{eq:weight-consistency}
(l_K+r_K)\E[B_W+\tilde S_W]
+K\E[\tilde D_W]=o(1).
    \end{equation}

    \item[(iv)] \textit{Residual-weighted moment control.}
    The second moments satisfy
    \begin{equation}
    \label{eq:weight-weighted-second-moments}
    \E[B_W^2+\tilde S_W^{\,2}+\tilde D_W^{\,2}]=O(1).
    \end{equation}

\end{enumerate}
\end{assumption}

Assumption~\ref{ass:estimated_weights} parallels the conditions imposed on
the nuisance regressions. Part~(ii) requires consistency of the fitted
weights and stability under local block deletion, while boundedness in
part~(i) provides uniform control of the weights. Parts~(iii)--(iv) impose
the corresponding residual-weighted conditions, with multiplier
\(\varepsilon_t^2\|\boldsymbol\xi_t\|_2^2\), reflecting how weight-estimation
errors enter the score through the oracle residual product. The rates in
\eqref{eq:weight-consistency} parallel the residual-weighted prediction
condition in \eqref{eq:residual-weighted-errors} and the deletion rates in
\eqref{eq:stability}, while the second-moment bounds in
\eqref{eq:weight-weighted-second-moments} correspond to those in
\eqref{eq:residual-weighted-errors} and \eqref{eq:stability-moments}.
The information condition in part~(i), together with block deletion, provides
the separation needed to control interactions between weight estimation and
temporally distant score terms. The condition \(A_W=o_p(1)\) additionally
ensures consistency of the estimated weighting rule.

For precision weighting, \(W_t=W_t^\star\) remains time-varying through
the conditional variance in \eqref{eq:optimal_weight}. Under \(H_0\),
Assumption~\ref{ass:errors}(i) gives
\[
c_1^{-4/(2+\delta)}\leq W_t^\star\leq c_2^{-2},\quad
\E[W_t^\star\varepsilon_t^2\given H_t^{Y,Z}]=1.
\]
Conditional independence further gives
\[
\E\big[(W_t^\star)^2\varepsilon_t^2
\boldsymbol\xi_t\boldsymbol\xi_t^\top\given H_t^{X,Y,Z}\big]
=W_t^\star\boldsymbol\xi_t\boldsymbol\xi_t^\top.
\]
Hence the population precision weights satisfy the required boundedness
provided \(c_W\) and \(C_W\) are chosen to contain this range, and the above
identities simplify the corresponding covariance expressions.

The fitted precision weights are obtained from an estimated conditional
variance and remain subject to Assumption~\ref{ass:estimated_weights}.
In practice, \(\hat\sigma_t^2\) may be truncated to
\([C_W^{-1},c_W^{-1}]\) before inversion. Since the reciprocal map is
Lipschitz on this interval, accuracy and deletion bounds for the variance
estimator transfer directly to the weights, including their residual-weighted
counterparts. In particular,
\[
|\hat W_t-W_t^\star|
\leq C_W^2\big|\hat\sigma_t^2-
\E[\varepsilon_t^2\given H_t^{Y,Z}]\big|.
\]
The same Lipschitz bound applies to the differences between deleted
variance fits.

\bibliographystyle{chicago}
\bibliography{paper-ref}

\bigskip

\newpage
\appendix

\section{Proof of Theorem \ref{thm:vector_gtcm}}
\label{app:null-proof}

\subsection{Vector decomposition and oracle martingale limit}

Substituting the residual and weight decompositions into $T^{(K)}$ and
collecting terms gives the following decomposition:
\begin{equation}
\label{eq:hatR_decompose}
\begin{aligned}
T^{(K)}
={}&
\underbrace{\frac{1}{\sqrt K}\sum_{t=1}^K
W_t\varepsilon_t\boldsymbol\xi_t}_{\mathrm{I}}
+
\underbrace{\frac{1}{\sqrt K}\sum_{t=1}^K
W_t\varepsilon_t
\big\{g_t(H_t^{Y,Z})-\hat g_t(H_t^{Y,Z})\big\}}_{\mathrm{II}}
\\
&+
\underbrace{\frac{1}{\sqrt K}\sum_{t=1}^K
W_t\boldsymbol\xi_t
\big\{f_t(H_t^{Y,Z})-\hat f_t(H_t^{Y,Z})\big\}}_{\mathrm{III}}
\\
&+
\underbrace{\frac{1}{\sqrt K}\sum_{t=1}^K
W_t
\big\{f_t(H_t^{Y,Z})-\hat f_t(H_t^{Y,Z})\big\}
\big\{g_t(H_t^{Y,Z})-\hat g_t(H_t^{Y,Z})\big\}}_{\mathrm{IV}}
\\
&+
\underbrace{\frac{1}{\sqrt K}\sum_{t=1}^K
(\hat W_t-W_t)\varepsilon_t\boldsymbol\xi_t}_{\mathrm{V}}
\\
&+
\underbrace{\frac{1}{\sqrt K}\sum_{t=1}^K
(\hat W_t-W_t)\varepsilon_t
\big\{g_t(H_t^{Y,Z})-\hat g_t(H_t^{Y,Z})\big\}}_{\mathrm{VI}}
\\
&+
\underbrace{\frac{1}{\sqrt K}\sum_{t=1}^K
(\hat W_t-W_t)\boldsymbol\xi_t
\big\{f_t(H_t^{Y,Z})-\hat f_t(H_t^{Y,Z})\big\}}_{\mathrm{VII}}
\\
&+
\underbrace{\frac{1}{\sqrt K}\sum_{t=1}^K
(\hat W_t-W_t)
\big\{f_t(H_t^{Y,Z})-\hat f_t(H_t^{Y,Z})\big\}
\big\{g_t(H_t^{Y,Z})-\hat g_t(H_t^{Y,Z})\big\}}_{\mathrm{VIII}} .
\end{aligned}
\end{equation}
For Term~\(\mathrm I\), appealing to Assumption~\ref{ass:errors}(i) and a
 martingale central limit theorem we obtain a Gaussian limit (see Lemma~\ref{lem:leading-term}). We show that the remaining terms are $o_P(1)$. For Terms~\(\mathrm{II}\)--\(\mathrm{III}\), we
use the weighted nuisance-error bounds in Assumption \ref{ass:errors}(iii), the projection conditions in
Assumption~\ref{ass:mixing}(ii)--(iii), and block stability in Assumption \ref{ass:stability}. Control of Terms
\(\mathrm{IV}\) and \(\mathrm{VIII}\) follow from the product
rate condition in Assumption \ref{ass:errors}(ii). Terms~\(\mathrm V\)--\(\mathrm{VII}\) use boundedness, consistency
and block stability of the estimated weights stated in Assumption \ref{ass:estimated_weights}, together with the corresponding
controls for the population-weighted terms sated in Assumption \ref{ass:population_weights}. 
We first establish the multivariate limit of term \(\mathrm{I}\), and then
control the remaining terms coordinatewise.

\begin{lemma}[Oracle martingale limit]
\label{lem:leading-term}
Under \(H_0\), Assumptions
\ref{ass:errors}(i), \ref{ass:mixing}(iv) and \ref{ass:population_weights} imply
\begin{equation}
\label{eq:leading-null}
\frac1{\sqrt K}\sum_{t=1}^K W_t\varepsilon_t\boldsymbol\xi_t
\xrightarrow{d}\mathcal N_d(0,\Sigma).
\end{equation}
Under the local alternatives in Theorem~\ref{thm:local_power}, the same
conclusion holds with \(\zeta_t\) in place of \(\varepsilon_t\), provided
the corresponding variance conditions and
$\E\big[|\zeta_t|^{2+\delta}\given H_t^{X,Y,Z}\big]
\leq c_1^2$ hold uniformly along the sequence.
\end{lemma}

\begin{proof}
Under \(H_0\),
\(W_t\boldsymbol\xi_t\) is
\(H_t^{X,Y,Z}\)-measurable and so
for \(s<t\),
\begin{align*}
\E[W_t\varepsilon_t\boldsymbol\xi_t\given H_t^{X,Y,Z}]
&=W_t\boldsymbol\xi_t\E[\varepsilon_t\given H_t^{X,Y,Z}]=0,\\
\E[W_tW_s\varepsilon_t\varepsilon_s\boldsymbol\xi_t\boldsymbol\xi_s^\top]
&=\E\big[W_s\varepsilon_s
\E[W_t\varepsilon_t\boldsymbol\xi_t\given H_t^{X,Y,Z}]
\boldsymbol\xi_s^\top\big]=0.
\end{align*}
Thus $W_t\varepsilon_t\boldsymbol\xi_t$ is a martingale different sequence relative to the filtration given by $H_t^{X,Y,Z}$. 
Recall
\[
\mathcal V_t =
\operatorname{vec}\big(
\E[W_t^2\varepsilon_t^2\boldsymbol\xi_t\boldsymbol\xi_t^\top
\given H_t^{X,Y,Z}]\big) = W_t^2\E[\varepsilon_t^2\given H_t^{Y,Z}]
\operatorname{vec}(\boldsymbol\xi_t\boldsymbol\xi_t^\top).
\]
Using Assumption~\ref{ass:mixing}(iv),
\begin{align*}
\E\Big\|\frac1K\sum_{t=1}^K(\mathcal V_t-\E\mathcal V_t)\Big\|_2^2
&=\frac1{K^2}\sum_{t=1}^K\sum_{s=1}^K
\Tr\{\Cov(\mathcal V_t,\mathcal V_s)\}\\
&\leq\frac1K\max_{1\leq t\leq K}\E\|\mathcal V_t\|_2^2
+\frac{2d}{K^2}\sum_{h=1}^{K-1}(K-h)
\max_{1\leq t\leq K-h}\|\Cov(\mathcal V_t,\mathcal V_{t+h})\|_F\\
&\leq\frac1K\max_{1\leq t\leq K}\E\|\mathcal V_t\|_2^2
+\frac{2d}{K}\sum_{h=1}^{K-1}
\max_{1\leq t\leq K-h}\|\Cov(\mathcal V_t,\mathcal V_{t+h})\|_F\\
& \rightarrow0.
\end{align*}
Hence
\begin{equation}
\label{eq:oracle-conditional-covariance-stabilization}
\frac1K\sum_{t=1}^K
\big\{
\E[W_t^2\varepsilon_t^2\boldsymbol\xi_t\boldsymbol\xi_t^\top
\given H_t^{X,Y,Z}]
-\E[W_t^2\varepsilon_t^2\boldsymbol\xi_t\boldsymbol\xi_t^\top]
\big\}\xrightarrow{p}0,
\end{equation}
and \eqref{eq:oracle-unconditional-covariance-limit} gives
\begin{equation}
\label{eq:oracle-conditional-covariance-limit}
\frac1K\sum_{t=1}^K
\E[W_t^2\varepsilon_t^2\boldsymbol\xi_t\boldsymbol\xi_t^\top
\given H_t^{X,Y,Z}]
\xrightarrow{p}\Sigma.
\end{equation}

To check the conditional Lyapunov condition, reduce \(\delta\) if necessary so that
\(0<\delta\leq2\), adjusting the finite upper-moment constant of Assumption~\ref{ass:errors}(i). By
\eqref{eq:null_hyp-multi} and 
Assumption
\ref{ass:errors}(i),
\begin{align*}
\max_{1\leq t\leq K}\E\|\boldsymbol\xi_t\|_2^4
&=\max_{1\leq t\leq K}\E\|\operatorname{vec}(\boldsymbol\xi_t\boldsymbol\xi_t^\top)\|_2^2
=O(1),\\
\E[|\varepsilon_t|^{2+\delta}\|\boldsymbol\xi_t\|_2^{2+\delta}]
&=\E\big[
\E[|\varepsilon_t|^{2+\delta}\given H_t^{X,Y,Z}]
\|\boldsymbol\xi_t\|_2^{2+\delta}
\big]\\
&\leq c_1^2\E \big[\|\boldsymbol\xi_t\|_2^{2+\delta}\big]\\
& \leq c_1^2\E\big[\|\boldsymbol\xi_t\|_2^4\big]^{(2+\delta)/4}.
\end{align*}
Consequently,
\begin{align*}
&\E\Big[
\frac1{K^{1+\delta/2}}\sum_{t=1}^K
\E[\|W_t\varepsilon_t\boldsymbol\xi_t\|_2^{2+\delta}
\given H_t^{X,Y,Z}]
\Big]\\
&\quad=\frac1{K^{1+\delta/2}}\sum_{t=1}^K
\E[\|W_t\varepsilon_t\boldsymbol\xi_t\|_2^{2+\delta}]\\
&\quad\leq\frac{C_W^{2+\delta}c_1^2}{K^{\delta/2}}
\max_{1\leq t\leq K}\big(\E\|\boldsymbol\xi_t\|_2^4\big)^{(2+\delta)/4}\\
& \quad\rightarrow0.
\end{align*}
Markov's inequality gives conditional Lyapunov and hence conditional
Lindeberg. This and \eqref{eq:oracle-conditional-covariance-limit} thus mean the conditions for the martingale central limit theorem \citep{dvoretzky1972asymptotic} are satisfied and prove \eqref{eq:leading-null}.

\end{proof}

\subsection{Coordinatewise control of the nuisance remainders}
\label{app:nuisance-remainders}

Fix \(j\in\{1,\ldots,d\}\). We use the quantities
\(A_q,B_q,S_q,D_q,\tilde S_q,\tilde D_q\) defined in
Section~\ref{sec:theory}, with \(q\in\{f,g\}\). All fitted values are evaluated at their original histories. For later reference, we first establish several useful bounds. The three-term decomposition
\begin{equation*}
    \begin{split}
q_t(H_t^{Y,Z})-\hat q_t^{-t,t'}(H_t^{Y,Z})
={}&q_t(H_t^{Y,Z})-\hat q_t(H_t^{Y,Z})\\
&+\hat q_t(H_t^{Y,Z})-\hat q_t^{-t}(H_t^{Y,Z})\\
&+\hat q_t^{-t}(H_t^{Y,Z})-\hat q_t^{-t,t'}(H_t^{Y,Z})
\end{split}
\end{equation*}
and the squared triangle inequality
give
\begin{equation}
\label{eq:deleted-prediction-envelope}
\frac1K\sum_{t=1}^K\max_{1\leq t'\leq K}
\|q_t(H_t^{Y,Z})-\hat q_t^{-t,t'}(H_t^{Y,Z})\|_2^2
\leq3(A_q+S_q+D_q).
\end{equation}
The corresponding residual-weighted bounds are
\begin{equation}
\label{eq:deleted-residual-envelope}
\begin{aligned}
\frac1K\sum_{t=1}^K\varepsilon_t^2
\max_{1\leq t'\leq K}
\|g_t(H_t^{Y,Z})-\hat g_t^{-t,t'}(H_t^{Y,Z})\|_2^2
&\leq3(B_g+\tilde S_g+\tilde D_g),\\
\frac1K\sum_{t=1}^K\|\boldsymbol\xi_t\|_2^2
\max_{1\leq t'\leq K}
|f_t(H_t^{Y,Z})-\hat f_t^{-t,t'}(H_t^{Y,Z})|^2
&\leq3(B_f+\tilde S_f+\tilde D_f).
\end{aligned}
\end{equation}
Assumptions~\ref{ass:stability} and \ref{ass:errors} give
\begin{equation}
\label{eq:deleted-error-moments}
\begin{aligned}
&\E[(A_q+S_q+D_q)^2]
+\E[(B_q+\tilde S_q+\tilde D_q)^2]\\
&\quad\leq3\big\{\E[A_q^2]+\E[B_q^2]+\E[S_q^2]+\E[D_q^2]
+\E[\tilde S_q^2]+\E[\tilde D_q^2]\big\}\\
&\quad=O(1),
\end{aligned}
\end{equation}
and
\begin{equation}
\label{eq:deletion-first-moments}
\begin{aligned}
&\E[S_q]+\E[D_q]+\E[\tilde S_q]+\E[\tilde D_q]\\
&\quad\leq\{\E[S_q^2]\}^{1/2}+\{\E[D_q^2]\}^{1/2}
+\{\E[\tilde S_q^2]\}^{1/2}+\{\E[\tilde D_q^2]\}^{1/2}\\
&\quad=O(1),
\end{aligned}
\end{equation}
These results will be invoked directly in the remainder of the proof without further comment.
\begin{lemma}[Exposure-feature regression remainder]
\label{lem:exposure-remainder-moments}
Under \(H_0\), suppose Assumptions~\ref{ass:mixing},
\ref{ass:stability}, \ref{ass:errors} and
\ref{ass:population_weights} hold. 
Then \(\mathrm{II}=o_p(1)\).
\end{lemma}

\begin{proof}
First,
\begin{align*}
&\frac1{\sqrt K}\sum_{t=1}^K W_t\varepsilon_t
\big\{g_{t,j}(H_t^{Y,Z})-\hat g_{t,j}(H_t^{Y,Z})\big\}\\
&\quad=\frac1{\sqrt K}\sum_{t=1}^K W_t\varepsilon_t
\big\{g_{t,j}(H_t^{Y,Z})-\hat g_{t,j}^{-t}(H_t^{Y,Z})\big\}\\
&\qquad{}+\frac1{\sqrt K}\sum_{t=1}^K W_t\varepsilon_t
\big\{\hat g_{t,j}^{-t}(H_t^{Y,Z})-\hat g_{t,j}(H_t^{Y,Z})\big\}.
\end{align*}
By the bounded weight assumption, Cauchy--Schwarz and
Assumption~\ref{ass:stability},
\begin{equation}
\label{eq:exposure-first-deletion-bound}
\begin{aligned}
&\Big|\frac1{\sqrt K}\sum_{t=1}^K W_t\varepsilon_t
\big\{\hat g_{t,j}^{-t}(H_t^{Y,Z})
      -\hat g_{t,j}(H_t^{Y,Z})\big\}\Big|\\
&\quad\leq C_W
\Big(\frac1K\sum_{t=1}^K\varepsilon_t^2\Big)^{1/2}
(KS_g)^{1/2}\\
&\quad=o_p(1).
\end{aligned}
\end{equation}

Squaring the remaining term gives
\begin{equation}
\label{eq:exposure-square-expansion}
\begin{aligned}
&\Big(\frac1{\sqrt K}\sum_{t=1}^K W_t\varepsilon_t
\big\{g_{t,j}(H_t^{Y,Z})
      -\hat g_{t,j}^{-t}(H_t^{Y,Z})\big\}\Big)^2\\
&\quad=\frac1K\sum_{t=1}^K\sum_{t'=1}^K
W_tW_{t'}\varepsilon_t\varepsilon_{t'}
\big\{g_{t,j}(H_t^{Y,Z})
      -\hat g_{t,j}^{-t}(H_t^{Y,Z})\big\}\times
\big\{g_{t',j}(H_{t'}^{Y,Z})
      -\hat g_{t',j}^{-t'}(H_{t'}^{Y,Z})\big\}.
\end{aligned}
\end{equation}
For \(|t-t'|<l_K+r_K\), taking expectations and using
\(2|ab|\leq a^2+b^2\) gives
\begin{align*}
&\frac1K\sum_{\substack{1\leq t,t'\leq K:\\
                      |t-t'|<l_K+r_K}}
\E\Big[\Big|W_tW_{t'}\varepsilon_t\varepsilon_{t'}
\big\{g_{t,j}(H_t^{Y,Z})
      -\hat g_{t,j}^{-t}(H_t^{Y,Z})\big\}\times
\big\{g_{t',j}(H_{t'}^{Y,Z})
      -\hat g_{t',j}^{-t'}(H_{t'}^{Y,Z})\big\}\Big|\Big]\\
&\quad\leq\frac{C_W^2}{2K}
\sum_{\substack{1\leq t,t'\leq K:\\
                |t-t'|<l_K+r_K}}
\E\Big[
\varepsilon_t^2
\big\{g_{t,j}(H_t^{Y,Z})
      -\hat g_{t,j}^{-t}(H_t^{Y,Z})\big\}^2+
\varepsilon_{t'}^2
\big\{g_{t',j}(H_{t'}^{Y,Z})
      -\hat g_{t',j}^{-t'}(H_{t'}^{Y,Z})\big\}^2
\Big]\\
&\quad\leq\frac{C_W^2\{2(l_K+r_K)-1\}}K
\sum_{t=1}^K
\E\big[
\varepsilon_t^2
\big\{g_{t,j}(H_t^{Y,Z})
      -\hat g_{t,j}^{-t}(H_t^{Y,Z})\big\}^2
\big]\\
&\quad\leq
2C_W^2\{2(l_K+r_K)-1\}\E[B_g+\tilde S_g]\\
&\quad=o(1).
\end{align*}
For the separated pairs, symmetry and the common twice-deleted fit give
the decomposition
\begin{align}
&\frac1K\sum_{\substack{1\leq t,t'\leq K:\\
                      |t-t'|\geq l_K+r_K}}
W_tW_{t'}\varepsilon_t\varepsilon_{t'}
\big\{g_{t,j}(H_t^{Y,Z})
      -\hat g_{t,j}^{-t}(H_t^{Y,Z})\big\}\times
\big\{g_{t',j}(H_{t'}^{Y,Z})
      -\hat g_{t',j}^{-t'}(H_{t'}^{Y,Z})\big\}\notag\\
&\quad=\frac2K\sum_{\substack{1\leq t<t'\leq K:\\
                            t'-t\geq l_K+r_K}}
W_tW_{t'}\varepsilon_t\varepsilon_{t'}
\big\{g_{t,j}(H_t^{Y,Z})
      -\hat g_{t,j}^{-t,t'}(H_t^{Y,Z})\big\}\times
\big\{g_{t',j}(H_{t'}^{Y,Z})
      -\hat g_{t',j}^{-t,t'}(H_{t'}^{Y,Z})\big\}\notag\\
&\qquad{}+\frac2K\sum_{\substack{1\leq t,t'\leq K:\\
                               |t-t'|\geq l_K+r_K}}
W_tW_{t'}\varepsilon_t\varepsilon_{t'}
\big\{g_{t,j}(H_t^{Y,Z})
      -\hat g_{t,j}^{-t,t'}(H_t^{Y,Z})\big\}\times
\big\{\hat g_{t',j}^{-t,t'}(H_{t'}^{Y,Z})
      -\hat g_{t',j}^{-t'}(H_{t'}^{Y,Z})\big\}\notag\\
&\qquad{}+\frac1K\sum_{\substack{1\leq t,t'\leq K:\\
                               |t-t'|\geq l_K+r_K}}
W_tW_{t'}\varepsilon_t\varepsilon_{t'}
\big\{\hat g_{t,j}^{-t,t'}(H_t^{Y,Z})
      -\hat g_{t,j}^{-t}(H_t^{Y,Z})\big\}\times
\big\{\hat g_{t',j}^{-t,t'}(H_{t'}^{Y,Z})
      -\hat g_{t',j}^{-t'}(H_{t'}^{Y,Z})\big\}.
\label{eq:exposure-pair-expansion}
\end{align}

For the first sum, \(t<t'\). Both evaluation histories lie in
\(H_{t'}^{X,Y,Z}\), and both fits omit \(\Delta_{t'}\).
Thus
\begin{align*}
&\E\Big[
W_tW_{t'}\varepsilon_t\varepsilon_{t'}
\big\{g_{t,j}(H_t^{Y,Z})
      -\hat g_{t,j}^{-t,t'}(H_t^{Y,Z})\big\}
\times
\big\{g_{t',j}(H_{t'}^{Y,Z})
      -\hat g_{t',j}^{-t,t'}(H_{t'}^{Y,Z})\big\}
\given H_{t'}^{X,Y,Z},A_{t'+r_K}^{X,Y,Z}
\Big]\\
&\quad=
W_tW_{t'}\varepsilon_t
\big\{g_{t,j}(H_t^{Y,Z})
      -\hat g_{t,j}^{-t,t'}(H_t^{Y,Z})\big\}
\times
\big\{g_{t',j}(H_{t'}^{Y,Z})
      -\hat g_{t',j}^{-t,t'}(H_{t'}^{Y,Z})\big\}\times
\E[\varepsilon_{t'}\given
H_{t'}^{X,Y,Z},A_{t'+r_K}^{X,Y,Z}].
\end{align*}
Taking expectations, then applying Cauchy--Schwarz over the indices and
H\"older with exponents \(4,2,4\), yields
\begin{align}
&\frac2K\sum_{\substack{1\leq t<t'\leq K:\\
                      t'-t\geq l_K+r_K}}
\Big|\E\Big[
W_tW_{t'}\varepsilon_t\varepsilon_{t'}
\big\{g_{t,j}(H_t^{Y,Z})
      -\hat g_{t,j}^{-t,t'}(H_t^{Y,Z})\big\}\times
\big\{g_{t',j}(H_{t'}^{Y,Z})
      -\hat g_{t',j}^{-t,t'}(H_{t'}^{Y,Z})\big\}
\Big]\Big|\notag\\
&\quad\leq
\frac2K\sum_{\substack{1\leq t<t'\leq K:\\
                      t'-t\geq l_K+r_K}}
\E\Big[\Big|
W_tW_{t'}\varepsilon_t
\big\{g_{t,j}(H_t^{Y,Z})
      -\hat g_{t,j}^{-t,t'}(H_t^{Y,Z})\big\}\notag\\*
&\hspace{44mm}{}\times
\big\{g_{t',j}(H_{t'}^{Y,Z})
      -\hat g_{t',j}^{-t,t'}(H_{t'}^{Y,Z})\big\}
\E[\varepsilon_{t'}\given
H_{t'}^{X,Y,Z},A_{t'+r_K}^{X,Y,Z}]
\Big|\Big]\notag\\
&\quad\leq2C_W^2\E\Bigg[
\Big\{\frac1K\sum_{t=1}^K\varepsilon_t^2
\max_{1\leq t'\leq K}
\big(g_{t,j}(H_t^{Y,Z})
     -\hat g_{t,j}^{-t,t'}(H_t^{Y,Z})\big)^2
\Big\}^{1/2}\notag\\*
&\hspace{44mm}{}\times
\Big\{K\sum_{t=1}^K
\E[\varepsilon_t\given
H_t^{X,Y,Z},A_{t+r_K}^{X,Y,Z}]^2
\Big\}^{1/2}\notag\\*
&\hspace{44mm}{}\times
\Big\{\frac1K\sum_{t=1}^K
\max_{1\leq t'\leq K}
\big(g_{t,j}(H_t^{Y,Z})
     -\hat g_{t,j}^{-t,t'}(H_t^{Y,Z})\big)^2
\Big\}^{1/2}
\Bigg]\notag\\
&\quad\leq6C_W^2
\Big\{K\sum_{t=1}^K
\E\big[
\E[\varepsilon_t\given
H_t^{X,Y,Z},A_{t+r_K}^{X,Y,Z}]^2
\big]\Big\}^{1/2}\notag\\*
&\hspace{44mm}{}\times
\big\{\E[(B_g+\tilde S_g+\tilde D_g)^2]\big\}^{1/4}
\big\{\E[(A_g+S_g+D_g)^2]\big\}^{1/4}\notag\\
&\quad=o(1)\times O(1)\times O(1)\notag\\
&\quad=o(1).
\label{eq:exposure-separated-bound}
\end{align}
The last line uses \eqref{eq:deleted-prediction-envelope},
\eqref{eq:deleted-residual-envelope},
\eqref{eq:deleted-error-moments}
and Assumption~\ref{ass:mixing}(ii).

For the second sum in \eqref{eq:exposure-pair-expansion}, project
\(\varepsilon_{t'}\), the residual attached to the singly deleted fit.
Conditioning gives, for either ordering,
\begin{align*}
&\E\Big[
W_tW_{t'}\varepsilon_t\varepsilon_{t'}
\big\{g_{t,j}(H_t^{Y,Z})
      -\hat g_{t,j}^{-t,t'}(H_t^{Y,Z})\big\}
\times
\big\{\hat g_{t',j}^{-t,t'}(H_{t'}^{Y,Z})
      -\hat g_{t',j}^{-t'}(H_{t'}^{Y,Z})\big\}
\given H_{t'}^{X,Y,Z},A_{t'+r_K}^{X,Y,Z}
\Big]\\
&\quad=
W_tW_{t'}\varepsilon_t
\big\{g_{t,j}(H_t^{Y,Z})
      -\hat g_{t,j}^{-t,t'}(H_t^{Y,Z})\big\}
\times
\big\{\hat g_{t',j}^{-t,t'}(H_{t'}^{Y,Z})
      -\hat g_{t',j}^{-t'}(H_{t'}^{Y,Z})\big\}\times
\E[\varepsilon_{t'}\given
H_{t'}^{X,Y,Z},A_{t'+r_K}^{X,Y,Z}].
\end{align*}
Consequently,
\begin{align}
&\frac2K\sum_{\substack{1\leq t,t'\leq K:\\
                      |t-t'|\geq l_K+r_K}}
\Big|\E\Big[
W_tW_{t'}\varepsilon_t\varepsilon_{t'}
\big\{g_{t,j}(H_t^{Y,Z})
      -\hat g_{t,j}^{-t,t'}(H_t^{Y,Z})\big\}\times
\big\{\hat g_{t',j}^{-t,t'}(H_{t'}^{Y,Z})
      -\hat g_{t',j}^{-t'}(H_{t'}^{Y,Z})\big\}
\Big]\Big|\notag\\
&\quad\leq
\frac2K\sum_{\substack{1\leq t,t'\leq K:\\
                      |t-t'|\geq l_K+r_K}}
\E\Big[\Big|
W_tW_{t'}\varepsilon_t
\big\{g_{t,j}(H_t^{Y,Z})
      -\hat g_{t,j}^{-t,t'}(H_t^{Y,Z})\big\}\notag\\*
&\hspace{34mm}{}\times
\big\{\hat g_{t',j}^{-t,t'}(H_{t'}^{Y,Z})
      -\hat g_{t',j}^{-t'}(H_{t'}^{Y,Z})\big\}
\E[\varepsilon_{t'}\given
H_{t'}^{X,Y,Z},A_{t'+r_K}^{X,Y,Z}]
\Big|\Big]\notag\\
&\quad\leq2C_W^2\E\Bigg[
\Big\{\frac1K\sum_{t=1}^K\varepsilon_t^2
\max_{1\leq t'\leq K}
\big(g_{t,j}(H_t^{Y,Z})
     -\hat g_{t,j}^{-t,t'}(H_t^{Y,Z})\big)^2
\Big\}^{1/2}\notag\\*
&\hspace{34mm}{}\times
\Big\{K\sum_{t=1}^K
\E[\varepsilon_t\given
H_t^{X,Y,Z},A_{t+r_K}^{X,Y,Z}]^2
\Big\}^{1/2}D_g^{1/2}
\Bigg]\notag\\
&\quad\leq2\sqrt3\,C_W^2
\Big\{K\sum_{t=1}^K
\E\big[
\E[\varepsilon_t\given
H_t^{X,Y,Z},A_{t+r_K}^{X,Y,Z}]^2
\big]\Big\}^{1/2}\notag\\*
&\hspace{34mm}{}\times
\big\{\E[(B_g+\tilde S_g+\tilde D_g)^2]\big\}^{1/4}
\{\E[D_g^2]\}^{1/4}\notag\\
&\quad=o(1)\times O(1)\times O(1)\notag\\
&\quad=o(1).
\label{eq:exposure-mixed-bound}
\end{align}

For the third sum, taking expectations and applying Cauchy--Schwarz gives
\begin{align*}
&\frac1K\sum_{\substack{1\leq t,t'\leq K:\\
                      |t-t'|\geq l_K+r_K}}
\E\Big[\Big|
W_tW_{t'}\varepsilon_t\varepsilon_{t'}
\big\{\hat g_{t,j}^{-t,t'}(H_t^{Y,Z})
      -\hat g_{t,j}^{-t}(H_t^{Y,Z})\big\}\times
\big\{\hat g_{t',j}^{-t,t'}(H_{t'}^{Y,Z})
      -\hat g_{t',j}^{-t'}(H_{t'}^{Y,Z})\big\}
\Big|\Big]\\
&\quad\leq\frac{C_W^2}{K}
\E\Big[
\Big(\sum_{t=1}^K|\varepsilon_t|
\max_{1\leq t'\leq K}
\big|\hat g_{t,j}^{-t,t'}(H_t^{Y,Z})
     -\hat g_{t,j}^{-t}(H_t^{Y,Z})\big|
\Big)^2
\Big]\\
&\quad\leq C_W^2
\E\Big[
\sum_{t=1}^K\varepsilon_t^2
\max_{1\leq t'\leq K}
\big|\hat g_{t,j}^{-t,t'}(H_t^{Y,Z})
     -\hat g_{t,j}^{-t}(H_t^{Y,Z})\big|^2
\Big]\\
&\quad\leq C_W^2K\E[\tilde D_g]\\
&\quad=o(1).
\end{align*}

Taking expectations in \eqref{eq:exposure-square-expansion},
and combining the nearby-pair bound with
\eqref{eq:exposure-pair-expansion},
\eqref{eq:exposure-separated-bound},
\eqref{eq:exposure-mixed-bound} and the preceding bound, gives
\begin{align*}
0
&\leq\E\Big[
\Big(\frac1{\sqrt K}\sum_{t=1}^K W_t\varepsilon_t
\big\{g_{t,j}(H_t^{Y,Z})
      -\hat g_{t,j}^{-t}(H_t^{Y,Z})\big\}
\Big)^2
\Big]\\
&\leq
2C_W^2\{2(l_K+r_K)-1\}\E[B_g+\tilde S_g]
+C_W^2K\E[\tilde D_g]+o(1)\\
&=o(1).
\end{align*}
Markov's inequality and
\eqref{eq:exposure-first-deletion-bound} therefore imply
\begin{equation}
\label{eq:termII-converge}
\begin{aligned}
\frac1{\sqrt K}\sum_{t=1}^K W_t\varepsilon_t
\big\{g_{t,j}(H_t^{Y,Z})
      -\hat g_{t,j}(H_t^{Y,Z})\big\}
=o_p(1).
\end{aligned}
\end{equation}
Summing over \(j=1,\ldots,d\) proves the result.
\end{proof}

\begin{lemma}[Outcome-regression remainder]
\label{lem:outcome-remainder-moments}
Under \(H_0\), suppose that Assumptions~\ref{ass:mixing},
\ref{ass:stability}, \ref{ass:errors} and
\ref{ass:population_weights} hold,
then \(\mathrm{III}=o_p(1)\). 
\end{lemma}

\begin{proof}
First,
\begin{align*}
&\frac1{\sqrt K}\sum_{t=1}^K W_t\boldsymbol\xi_{t,j}
\big\{f_t(H_t^{Y,Z})-\hat f_t(H_t^{Y,Z})\big\}\\
&\quad=\frac1{\sqrt K}\sum_{t=1}^K W_t\boldsymbol\xi_{t,j}
\big\{f_t(H_t^{Y,Z})-\hat f_t^{-t}(H_t^{Y,Z})\big\}\\
&\qquad{}+\frac1{\sqrt K}\sum_{t=1}^K W_t\boldsymbol\xi_{t,j}
\big\{\hat f_t^{-t}(H_t^{Y,Z})-\hat f_t(H_t^{Y,Z})\big\}.
\end{align*}
Under the bounded weight assumption, Cauchy--Schwarz and
the stability assumption give
\begin{equation}
\label{eq:outcome-first-deletion-bound}
\begin{aligned}
&\Big|\frac1{\sqrt K}\sum_{t=1}^K W_t\boldsymbol\xi_{t,j}
\big\{\hat f_t^{-t}(H_t^{Y,Z})-\hat f_t(H_t^{Y,Z})\big\}\Big|\\
&\quad\leq C_W
\Big(\frac1K\sum_{t=1}^K\|\boldsymbol\xi_t\|_2^2\Big)^{1/2}
(KS_f)^{1/2}\\
&\quad=o_p(1),
\end{aligned}
\end{equation}
where the empirical second moment of \(\boldsymbol\xi_t\) is
\(O_p(1)\) by Assumption~\ref{ass:errors}(i).

Squaring the remaining term gives
\begin{equation}
\label{eq:outcome-square-expansion}
\begin{aligned}
&\Big(\frac1{\sqrt K}\sum_{t=1}^K W_t\boldsymbol\xi_{t,j}
\big\{f_t(H_t^{Y,Z})-\hat f_t^{-t}(H_t^{Y,Z})\big\}\Big)^2\\
&\quad=\frac1K\sum_{t=1}^K\sum_{t'=1}^K
W_tW_{t'}\boldsymbol\xi_{t,j}\boldsymbol\xi_{t',j}
\big\{f_t(H_t^{Y,Z})-\hat f_t^{-t}(H_t^{Y,Z})\big\}\times
\big\{f_{t'}(H_{t'}^{Y,Z})
-\hat f_{t'}^{-t'}(H_{t'}^{Y,Z})\big\}.
\end{aligned}
\end{equation}
For \(|t-t'|<l_K+r_K\),
\begin{align*}
&\frac1K\sum_{\substack{1\leq t,t'\leq K:\\
|t-t'|<l_K+r_K}}
\E\Big[
\Big|W_tW_{t'}\boldsymbol\xi_{t,j}\boldsymbol\xi_{t',j}
\big\{f_t(H_t^{Y,Z})-\hat f_t^{-t}(H_t^{Y,Z})\big\}\times
\big\{f_{t'}(H_{t'}^{Y,Z})
-\hat f_{t'}^{-t'}(H_{t'}^{Y,Z})\big\}\Big|
\Big]\\
&\quad\leq\frac{C_W^2}{2K}
\sum_{\substack{1\leq t,t'\leq K:\\
|t-t'|<l_K+r_K}}
\E\Big[
\boldsymbol\xi_{t,j}^2
\big\{f_t(H_t^{Y,Z})-\hat f_t^{-t}(H_t^{Y,Z})\big\}^2+
\boldsymbol\xi_{t',j}^2
\big\{f_{t'}(H_{t'}^{Y,Z})
-\hat f_{t'}^{-t'}(H_{t'}^{Y,Z})\big\}^2
\Big]\\
&\quad\leq\frac{C_W^2\{2(l_K+r_K)-1\}}K
\E\Big[
\sum_{t=1}^K\boldsymbol\xi_{t,j}^2
\big\{f_t(H_t^{Y,Z})-\hat f_t^{-t}(H_t^{Y,Z})\big\}^2
\Big]\\
&\quad\leq\frac{2C_W^2\{2(l_K+r_K)-1\}}K
\E\Big[
\sum_{t=1}^K\boldsymbol\xi_{t,j}^2
\big\{f_t(H_t^{Y,Z})-\hat f_t(H_t^{Y,Z})\big\}^2+
\sum_{t=1}^K\boldsymbol\xi_{t,j}^2
\big\{\hat f_t(H_t^{Y,Z})-\hat f_t^{-t}(H_t^{Y,Z})\big\}^2
\Big]\\
&\quad\leq
2C_W^2\{2(l_K+r_K)-1\}\E[B_f+\tilde S_f]\\
&\quad=o(1).
\end{align*}
For the separated pairs, symmetry and the common twice-deleted fit
give the exact decomposition
\begin{align}
&\frac1K\sum_{\substack{1\leq t,t'\leq K:\\
|t-t'|\geq l_K+r_K}}
W_tW_{t'}\boldsymbol\xi_{t,j}\boldsymbol\xi_{t',j}
\big\{f_t(H_t^{Y,Z})-\hat f_t^{-t}(H_t^{Y,Z})\big\}
\times
\big\{f_{t'}(H_{t'}^{Y,Z})
-\hat f_{t'}^{-t'}(H_{t'}^{Y,Z})\big\}
\notag\\
&\quad=\frac2K\sum_{\substack{1\leq t<t'\leq K:\\
t'-t\geq l_K+r_K}}
W_tW_{t'}\boldsymbol\xi_{t,j}\boldsymbol\xi_{t',j}
\big\{f_t(H_t^{Y,Z})-\hat f_t^{-t,t'}(H_t^{Y,Z})\big\}
\times
\big\{f_{t'}(H_{t'}^{Y,Z})
-\hat f_{t'}^{-t,t'}(H_{t'}^{Y,Z})\big\}
\notag\\
&\qquad{}+\frac2K\sum_{\substack{1\leq t,t'\leq K:\\
|t-t'|\geq l_K+r_K}}
W_tW_{t'}\boldsymbol\xi_{t,j}\boldsymbol\xi_{t',j}
\big\{f_t(H_t^{Y,Z})-\hat f_t^{-t,t'}(H_t^{Y,Z})\big\}
\times
\big\{\hat f_{t'}^{-t,t'}(H_{t'}^{Y,Z})
-\hat f_{t'}^{-t'}(H_{t'}^{Y,Z})\big\}
\notag\\
&\qquad{}+\frac1K\sum_{\substack{1\leq t,t'\leq K:\\
|t-t'|\geq l_K+r_K}}
W_tW_{t'}\boldsymbol\xi_{t,j}\boldsymbol\xi_{t',j}
\big\{\hat f_t^{-t,t'}(H_t^{Y,Z})
-\hat f_t^{-t}(H_t^{Y,Z})\big\}
\times
\big\{\hat f_{t'}^{-t,t'}(H_{t'}^{Y,Z})
-\hat f_{t'}^{-t'}(H_{t'}^{Y,Z})\big\}.
\label{eq:outcome-pair-expansion}
\end{align}

For the first sum, \(t<t'\). The outcome fits omit
\(\Delta_{t'}\), and the earlier exposure residual is determined
by \(\phi(\bar X_t)\) and \(H_t^{Y,Z}\).
When fitting randomization is present, all conditioning information
sets below also include this independent randomization; the residual
conditional expectations are unchanged.
Thus
\begin{align*}
&\E\Big[
W_tW_{t'}\boldsymbol\xi_{t,j}\boldsymbol\xi_{t',j}
\big\{f_t(H_t^{Y,Z})-\hat f_t^{-t,t'}(H_t^{Y,Z})\big\}\\
&\hspace{28mm}{}\times
\big\{f_{t'}(H_{t'}^{Y,Z})
-\hat f_{t'}^{-t,t'}(H_{t'}^{Y,Z})\big\}
\given H_{t'}^{Y,Z},\phi(\bar X_t),A_{t'+r_K}^{Y,Z}
\Big]\\
&\quad=
W_tW_{t'}\boldsymbol\xi_{t,j}
\big\{f_t(H_t^{Y,Z})-\hat f_t^{-t,t'}(H_t^{Y,Z})\big\}\\
&\hspace{28mm}{}\times
\big\{f_{t'}(H_{t'}^{Y,Z})
-\hat f_{t'}^{-t,t'}(H_{t'}^{Y,Z})\big\}\times
\E[\boldsymbol\xi_{t',j}\given
H_{t'}^{Y,Z},\phi(\bar X_t),A_{t'+r_K}^{Y,Z}].
\end{align*}
Taking expectations, then applying Cauchy--Schwarz over the indices
and H\"older with exponents \(4,2,4\), yields
\begin{align}
&\frac2K\sum_{\substack{1\leq t<t'\leq K:\\
t'-t\geq l_K+r_K}}
\bigg|\E\Big[
W_tW_{t'}\boldsymbol\xi_{t,j}\boldsymbol\xi_{t',j}
\big\{f_t(H_t^{Y,Z})-\hat f_t^{-t,t'}(H_t^{Y,Z})\big\}
\times
\big\{f_{t'}(H_{t'}^{Y,Z})
-\hat f_{t'}^{-t,t'}(H_{t'}^{Y,Z})\big\}
\Big]\bigg|
\notag\\
&\quad\leq\frac2K\sum_{\substack{1\leq t<t'\leq K:\\
t'-t\geq l_K+r_K}}
\E\Big[
\Big|W_tW_{t'}\boldsymbol\xi_{t,j}
\big\{f_t(H_t^{Y,Z})-\hat f_t^{-t,t'}(H_t^{Y,Z})\big\}
\notag\\*
&\hspace{38mm}{}\times
\big\{f_{t'}(H_{t'}^{Y,Z})
-\hat f_{t'}^{-t,t'}(H_{t'}^{Y,Z})\big\}
\times
\E[\boldsymbol\xi_{t',j}\given
H_{t'}^{Y,Z},\phi(\bar X_t),A_{t'+r_K}^{Y,Z}]
\Big|
\Big]
\notag\\
&\quad\leq2C_W^2\E\Bigg[
\Big\{\frac1K\sum_{t=1}^K\boldsymbol\xi_{t,j}^2
\max_{1\leq t'\leq K}
\big(f_t(H_t^{Y,Z})-\hat f_t^{-t,t'}(H_t^{Y,Z})\big)^2
\Big\}^{1/2}
\notag\\*
&\hspace{48mm}{}\times
\Bigg\{K\sum_{t=1}^K
\max_{\substack{1\leq s\leq K:\\|s-t|\geq l_K+r_K}}
\big\|\E[\boldsymbol\xi_t\given
H_t^{Y,Z},\phi(\bar X_s),A_{t+r_K}^{Y,Z}]\big\|_2^2
\Bigg\}^{1/2}
\notag\\*
&\hspace{48mm}{}\times
\Big\{\frac1K\sum_{t=1}^K
\max_{1\leq t'\leq K}
\big(f_t(H_t^{Y,Z})-\hat f_t^{-t,t'}(H_t^{Y,Z})\big)^2
\Big\}^{1/2}
\Bigg]
\notag\\
&\quad\leq6C_W^2
\Bigg\{K\sum_{t=1}^K
\E\Big[
\max_{\substack{1\leq s\leq K:\\|s-t|\geq l_K+r_K}}
\big\|\E[\boldsymbol\xi_t\given
H_t^{Y,Z},\phi(\bar X_s),A_{t+r_K}^{Y,Z}]\big\|_2^2
\Big]\Bigg\}^{1/2}
\notag\\*
&\hspace{48mm}{}\times
\big\{\E[(B_f+\tilde S_f+\tilde D_f)^2]\big\}^{1/4}
\big\{\E[(A_f+S_f+D_f)^2]\big\}^{1/4}
\notag\\
&\quad=o(1).
\label{eq:outcome-separated-bound}
\end{align}
The last line uses \eqref{eq:deleted-prediction-envelope},
\eqref{eq:deleted-residual-envelope}, \eqref{eq:deleted-error-moments}
and Assumption~\ref{ass:mixing}(iii).

For the second sum in \eqref{eq:outcome-pair-expansion}, project
\(\boldsymbol\xi_{t',j}\), the residual attached to the singly
deleted fit. Conditioning gives, for either ordering,
\begin{align*}
&\E\Big[
W_tW_{t'}\boldsymbol\xi_{t,j}\boldsymbol\xi_{t',j}
\big\{f_t(H_t^{Y,Z})-\hat f_t^{-t,t'}(H_t^{Y,Z})\big\}\\
&\hspace{28mm}{}\times
\big\{\hat f_{t'}^{-t,t'}(H_{t'}^{Y,Z})
-\hat f_{t'}^{-t'}(H_{t'}^{Y,Z})\big\}
\given H_{t'}^{Y,Z},\phi(\bar X_t),A_{t'+r_K}^{Y,Z}
\Big]\\
&\quad=
W_tW_{t'}\boldsymbol\xi_{t,j}
\big\{f_t(H_t^{Y,Z})-\hat f_t^{-t,t'}(H_t^{Y,Z})\big\}\\
&\hspace{28mm}{}\times
\big\{\hat f_{t'}^{-t,t'}(H_{t'}^{Y,Z})
-\hat f_{t'}^{-t'}(H_{t'}^{Y,Z})\big\}\times
\E[\boldsymbol\xi_{t',j}\given
H_{t'}^{Y,Z},\phi(\bar X_t),A_{t'+r_K}^{Y,Z}].
\end{align*}
Consequently,
\begin{align}
&\frac2K\sum_{\substack{1\leq t,t'\leq K:\\
|t-t'|\geq l_K+r_K}}
\bigg|\E\Big[
W_tW_{t'}\boldsymbol\xi_{t,j}\boldsymbol\xi_{t',j}
\big\{f_t(H_t^{Y,Z})-\hat f_t^{-t,t'}(H_t^{Y,Z})\big\}
\times
\big\{\hat f_{t'}^{-t,t'}(H_{t'}^{Y,Z})
-\hat f_{t'}^{-t'}(H_{t'}^{Y,Z})\big\}
\Big]\bigg|
\notag\\
&\quad\leq\frac2K\sum_{\substack{1\leq t,t'\leq K:\\
|t-t'|\geq l_K+r_K}}
\E\Big[
\Big|W_tW_{t'}\boldsymbol\xi_{t,j}
\big\{f_t(H_t^{Y,Z})-\hat f_t^{-t,t'}(H_t^{Y,Z})\big\}
\notag\\*
&\hspace{28mm}{}\times
\big\{\hat f_{t'}^{-t,t'}(H_{t'}^{Y,Z})
-\hat f_{t'}^{-t'}(H_{t'}^{Y,Z})\big\}
\times
\E[\boldsymbol\xi_{t',j}\given
H_{t'}^{Y,Z},\phi(\bar X_t),A_{t'+r_K}^{Y,Z}]
\Big|
\Big]
\notag\\
&\quad\leq2C_W^2\E\Bigg[
\Big\{\frac1K\sum_{t=1}^K\boldsymbol\xi_{t,j}^2
\max_{1\leq t'\leq K}
\big(f_t(H_t^{Y,Z})-\hat f_t^{-t,t'}(H_t^{Y,Z})\big)^2
\Big\}^{1/2}
\notag\\*
&\hspace{24mm}{}\times
\Bigg\{K\sum_{t=1}^K
\max_{\substack{1\leq s\leq K:\\|s-t|\geq l_K+r_K}}
\big\|\E[\boldsymbol\xi_t\given
H_t^{Y,Z},\phi(\bar X_s),A_{t+r_K}^{Y,Z}]\big\|_2^2
\Bigg\}^{1/2}D_f^{1/2}
\Bigg]
\notag\\
&\quad\leq2\sqrt3\,C_W^2
\Bigg\{K\sum_{t=1}^K
\E\Big[
\max_{\substack{1\leq s\leq K:\\|s-t|\geq l_K+r_K}}
\big\|\E[\boldsymbol\xi_t\given
H_t^{Y,Z},\phi(\bar X_s),A_{t+r_K}^{Y,Z}]\big\|_2^2
\Big]\Bigg\}^{1/2}
\notag\\*
&\hspace{24mm}{}\times
\big\{\E[(B_f+\tilde S_f+\tilde D_f)^2]\big\}^{1/4}
\{\E[D_f^2]\}^{1/4}
\notag\\
&\quad=o(1).
\label{eq:outcome-mixed-bound}
\end{align}
Thus the sums of the absolute expectations of the first two
pairwise contributions are \(o(1)\).

For the third sum, Cauchy--Schwarz directly gives
\begin{align*}
&\frac1K\sum_{\substack{1\leq t,t'\leq K:\\
|t-t'|\geq l_K+r_K}}
\E\Big[
\big|W_tW_{t'}\boldsymbol\xi_{t,j}\boldsymbol\xi_{t',j}
\big\{\hat f_t^{-t,t'}(H_t^{Y,Z})
-\hat f_t^{-t}(H_t^{Y,Z})\big\}\times
\big\{\hat f_{t'}^{-t,t'}(H_{t'}^{Y,Z})
-\hat f_{t'}^{-t'}(H_{t'}^{Y,Z})\big\}
\big|
\Big]\\
&\quad\leq\frac{C_W^2}{K}
\E\Big[
\Big(\sum_{t=1}^K|\boldsymbol\xi_{t,j}|
\max_{1\leq t'\leq K}
\big|\hat f_t^{-t,t'}(H_t^{Y,Z})
-\hat f_t^{-t}(H_t^{Y,Z})\big|\Big)^2
\Big]\\
&\quad\leq C_W^2
\E\Big[
\sum_{t=1}^K\boldsymbol\xi_{t,j}^2
\max_{1\leq t'\leq K}
\big|\hat f_t^{-t,t'}(H_t^{Y,Z})
-\hat f_t^{-t}(H_t^{Y,Z})\big|^2
\Big]\\
&\quad\leq C_W^2K\E[\tilde D_f]\\
&\quad=o(1).
\end{align*}
Taking expectations in \eqref{eq:outcome-square-expansion}
and combining these bounds gives
\begin{align*}
0
&\leq
\E\Bigg[
\Big(\frac1{\sqrt K}\sum_{t=1}^K
W_t\boldsymbol\xi_{t,j}
\big\{f_t(H_t^{Y,Z})-\hat f_t^{-t}(H_t^{Y,Z})\big\}
\Big)^2
\Bigg]\\
&\leq
2C_W^2\{2(l_K+r_K)-1\}\E[B_f+\tilde S_f]
+C_W^2K\E[\tilde D_f]+o(1)\\
&=o(1).
\end{align*}
Markov's inequality and \eqref{eq:outcome-first-deletion-bound}
therefore imply
\begin{equation}
\label{eq:termIII-converge}
\frac1{\sqrt K}\sum_{t=1}^K W_t\boldsymbol\xi_{t,j}
\big\{f_t(H_t^{Y,Z})-\hat f_t(H_t^{Y,Z})\big\}
=o_p(1).
\end{equation}
Since \(d\) is fixed, applying this argument to
\(j=1,\ldots,d\) proves the result.
\end{proof}

\begin{lemma}[Product of nuisance errors]
Under Assumptions~\ref{ass:errors} and \ref{ass:population_weights},
term~\(\mathrm{IV}\) is \(o_p(1)\).
\end{lemma}

\begin{proof}
The triangle and Cauchy--Schwarz inequalities give
\begin{equation}
\label{eq:termIV-product-bound}
\begin{aligned}
&\Big\|\frac1{\sqrt K}\sum_{t=1}^K W_t
\{f_t(H_t^{Y,Z})-\hat f_t(H_t^{Y,Z})\}
\{g_t(H_t^{Y,Z})-\hat g_t(H_t^{Y,Z})\}\Big\|_2\\
&\quad\leq\frac{C_W}{\sqrt K}\sum_{t=1}^K
|f_t(H_t^{Y,Z})-\hat f_t(H_t^{Y,Z})|
\|g_t(H_t^{Y,Z})-\hat g_t(H_t^{Y,Z})\|_2\\
&\quad\leq\frac{C_W}{\sqrt K}
\Big\{\sum_{t=1}^K|f_t(H_t^{Y,Z})-\hat f_t(H_t^{Y,Z})|^2\Big\}^{1/2}
\Big\{\sum_{t=1}^K\|g_t(H_t^{Y,Z})-\hat g_t(H_t^{Y,Z})\|_2^2\Big\}^{1/2}\\
&\quad=C_W\sqrt{K A_fA_g}\\
&\quad =o_p(1),
\end{aligned}
\end{equation}
by \eqref{eq:Kmsemse}.
\end{proof}

{\textbf{Estimated-weight remainders}}

\begin{lemma}[Estimated-weight remainder]
\label{lem:estimated-weight-remainder}
Under \(H_0\), suppose that Assumptions~\ref{ass:mixing}(ii), (iv),
\ref{ass:errors}(i), \ref{ass:population_weights}(i) and
\ref{ass:estimated_weights} hold, 
then \(\mathrm V=o_p(1)\).
\end{lemma}

\begin{proof}
Adding and subtracting the once-deleted weights, we obtain 
\[
\mathrm V
=\frac1{\sqrt K}\sum_{t=1}^K
(\hat W_t-\hat W_t^{-t})\varepsilon_t\boldsymbol\xi_t
+\frac1{\sqrt K}\sum_{t=1}^K
(\hat W_t^{-t}-W_t)\varepsilon_t\boldsymbol\xi_t.
\]
Under \(H_0\), conditional independence and
Assumption~\ref{ass:errors}(i) give
\[
\E\Big[\frac1K\sum_{t=1}^K
\varepsilon_t^2\|\boldsymbol\xi_t\|_2^2\Big]
\leq c_1^{4/(2+\delta)}\frac1K\sum_{t=1}^K
\E\|\boldsymbol\xi_t\|_2^2=O(1),
\quad
\frac1K\sum_{t=1}^K
\varepsilon_t^2\|\boldsymbol\xi_t\|_2^2=O_p(1).
\]
The probability stability rates implied by
\eqref{eq:weight-stability-moments} are
\begin{equation}
\label{eq:weight-V-probability-stability}
\sum_{t=1}^K(\hat W_t-\hat W_t^{-t})^2=o_p(1),
\quad
\sum_{t=1}^K\max_{1\leq t'\leq K}
(\hat W_t^{-t}-\hat W_t^{-t,t'})^2=o_p(1).
\end{equation}
For the first deletion, Cauchy--Schwarz gives directly
\begin{equation}
\label{eq:weight-V-first-deletion}
\begin{aligned}
&\Big\|\frac1{\sqrt K}\sum_{t=1}^K
(\hat W_t-\hat W_t^{-t})
\varepsilon_t\boldsymbol\xi_t\Big\|_2\\
&\quad\leq
\Big\{\sum_{t=1}^K(\hat W_t-\hat W_t^{-t})^2\Big\}^{1/2}
\Big\{\frac1K\sum_{t=1}^K
\varepsilon_t^2\|\boldsymbol\xi_t\|_2^2\Big\}^{1/2}\\
&\quad=o_p(1).
\end{aligned}
\end{equation}
It remains to control the once-deleted sum. Fix
\(j\in\{1,\ldots,d\}\). Assumption~\ref{ass:errors}(i) gives
\begin{equation}
\label{eq:weight-V-exposure-moment}
\E\Big\{\frac1K\sum_{t=1}^K\boldsymbol\xi_{t,j}^2\Big\}^{2}
\leq\frac1K\sum_{t=1}^K\E[\boldsymbol\xi_{t,j}^4]
\leq\max_{1\leq t\leq K}\E\|\boldsymbol\xi_t\|_2^4=O(1).
\end{equation}
The identity
\[
\hat W_t^{-t,t'}-W_t
=(\hat W_t-W_t)+(\hat W_t^{-t}-\hat W_t)
+(\hat W_t^{-t,t'}-\hat W_t^{-t})
\]
and \eqref{eq:weight-weighted-second-moments} yield
\begin{align}
&\E\bigg[\Big\{\frac1K\sum_{t=1}^K
\varepsilon_t^2\boldsymbol\xi_{t,j}^2
\max_{1\leq t'\leq K}
(\hat W_t^{-t,t'}-W_t)^2\Big\}^{2}\bigg]\notag\\
&\quad\leq9\E[(B_W+\tilde S_W+\tilde D_W)^2]\notag\\
&\quad\leq27\E[B_W^2+\tilde S_W^{\,2}+\tilde D_W^{\,2}]
\notag\\
&\quad=O(1).
\label{eq:weight-V-deleted-moment}
\end{align}
The square of the once-deleted sum is
\begin{equation}
\label{eq:weight-V-square-expansion}
\begin{aligned}
&\Big\{\frac1{\sqrt K}\sum_{t=1}^K
(\hat W_t^{-t}-W_t)\varepsilon_t\boldsymbol\xi_{t,j}\Big\}^{2}\\
&\quad=\frac1K\sum_{t=1}^K\sum_{t'=1}^K
(\hat W_t^{-t}-W_t)(\hat W_{t'}^{-t'}-W_{t'})
\varepsilon_t\varepsilon_{t'}
\boldsymbol\xi_{t,j}\boldsymbol\xi_{t',j}.
\end{aligned}
\end{equation}
For the diagonal and the nearby pairs, \(2|ab|\leq a^2+b^2\) gives
\begin{align*}
&\frac1K\sum_{\substack{1\leq t,t'\leq K:\\
|t-t'|<l_K+r_K}}
\E\Big[
\Big|(\hat W_t^{-t}-W_t)(\hat W_{t'}^{-t'}-W_{t'})
\varepsilon_t\varepsilon_{t'}
\boldsymbol\xi_{t,j}\boldsymbol\xi_{t',j}\Big|
\Big]\\
&\quad\leq\frac{2(l_K+r_K)-1}{K}
\E\Big[\sum_{t=1}^K
(\hat W_t^{-t}-W_t)^2
\varepsilon_t^2\boldsymbol\xi_{t,j}^2\Big]\\
&\quad\leq2\{2(l_K+r_K)-1\}\E[B_W+\tilde S_W]\\
&\quad=o(1).
\end{align*}
For separated pairs, expand
\begin{align}
&\frac1K\sum_{\substack{1\leq t,t'\leq K:\\
|t-t'|\geq l_K+r_K}}
(\hat W_t^{-t}-W_t)(\hat W_{t'}^{-t'}-W_{t'})
\varepsilon_t\varepsilon_{t'}
\boldsymbol\xi_{t,j}\boldsymbol\xi_{t',j}\notag\\
&\quad=\frac2K\sum_{\substack{1\leq t<t'\leq K:\\
t'-t\geq l_K+r_K}}
(\hat W_t^{-t,t'}-W_t)(\hat W_{t'}^{-t,t'}-W_{t'})
\varepsilon_t\varepsilon_{t'}
\boldsymbol\xi_{t,j}\boldsymbol\xi_{t',j}\notag\\
&\qquad{}+\frac2K\sum_{\substack{1\leq t,t'\leq K:\\
|t-t'|\geq l_K+r_K}}
(\hat W_t^{-t,t'}-W_t)
(\hat W_{t'}^{-t'}-\hat W_{t'}^{-t,t'})
\varepsilon_t\varepsilon_{t'}
\boldsymbol\xi_{t,j}\boldsymbol\xi_{t',j}\notag\\
&\qquad{}+\frac1K\sum_{\substack{1\leq t,t'\leq K:\\
|t-t'|\geq l_K+r_K}}
(\hat W_t^{-t}-\hat W_t^{-t,t'})
(\hat W_{t'}^{-t'}-\hat W_{t'}^{-t,t'})
\varepsilon_t\varepsilon_{t'}
\boldsymbol\xi_{t,j}\boldsymbol\xi_{t',j}.
\label{eq:weight-V-pair-expansion}
\end{align}
When fitting randomization is present, all conditioning information
sets below also include this independent randomization; the residual
conditional expectations are unchanged.
Considering the first sum, when \(t<t'\) we have
\begin{align*}
&\E\big[
(\hat W_t^{-t,t'}-W_t)(\hat W_{t'}^{-t,t'}-W_{t'})
\varepsilon_t\varepsilon_{t'}
\boldsymbol\xi_{t,j}\boldsymbol\xi_{t',j}
\given H_{t'}^{X,Y,Z},A_{t'+r_K}^{X,Y,Z}
\big]\\
&\quad=
(\hat W_t^{-t,t'}-W_t)(\hat W_{t'}^{-t,t'}-W_{t'})
\varepsilon_t\boldsymbol\xi_{t,j}\boldsymbol\xi_{t',j}
\E[\varepsilon_{t'}\given
H_{t'}^{X,Y,Z},A_{t'+r_K}^{X,Y,Z}].
\end{align*}
Since both weight errors are bounded by \(C_W-c_W\),
Cauchy--Schwarz over the indices and H\"older's inequality give
\begin{align}
&\frac2K\sum_{\substack{1\leq t<t'\leq K:\\
t'-t\geq l_K+r_K}}
\Big|\E\big[
(\hat W_t^{-t,t'}-W_t)(\hat W_{t'}^{-t,t'}-W_{t'})
\varepsilon_t\varepsilon_{t'}
\boldsymbol\xi_{t,j}\boldsymbol\xi_{t',j}
\big]\Big|\notag\\
&\quad\leq2(C_W-c_W)\E\bigg[
\Big\{\frac1K\sum_{t=1}^K
\varepsilon_t^2\boldsymbol\xi_{t,j}^2
\max_{1\leq t'\leq K}
(\hat W_t^{-t,t'}-W_t)^2\Big\}^{1/2}\notag\\*
&\hspace{33mm}{}\times
\Big\{K\sum_{t=1}^K
\E[\varepsilon_t\given
H_t^{X,Y,Z},A_{t+r_K}^{X,Y,Z}]^2\Big\}^{1/2}
\Big\{\frac1K\sum_{t=1}^K
\boldsymbol\xi_{t,j}^2\Big\}^{1/2}
\bigg]\notag\\
&\quad\leq2(C_W-c_W)
\E\bigg[\Big\{\frac1K\sum_{t=1}^K
\varepsilon_t^2\boldsymbol\xi_{t,j}^2
\max_{1\leq t'\leq K}
(\hat W_t^{-t,t'}-W_t)^2\Big\}^{2}\bigg]^{1/4}
\notag\\*
&\hspace{33mm}{}\times
\Big\{K\sum_{t=1}^K
\E\big[\E[\varepsilon_t\given
H_t^{X,Y,Z},A_{t+r_K}^{X,Y,Z}]^2\big]\Big\}^{1/2}
\Big\{\max_{1\leq t\leq K}
\E\|\boldsymbol\xi_t\|_2^4\Big\}^{1/4}\notag\\
&\quad=o(1).
\label{eq:weight-V-separated-bound}
\end{align}
Here \eqref{eq:weight-V-deleted-moment} and
\eqref{eq:weight-V-exposure-moment} bound the first and third factors,
while Assumption~\ref{ass:mixing}(ii) makes the middle factor vanish.
For the second sum in \eqref{eq:weight-V-pair-expansion}, project
the residual attached to the singly deleted weight. Conditioning gives,
for either ordering,
\begin{align*}
&\E\big[
(\hat W_t^{-t,t'}-W_t)
(\hat W_{t'}^{-t'}-\hat W_{t'}^{-t,t'})
\varepsilon_t\varepsilon_{t'}
\boldsymbol\xi_{t,j}\boldsymbol\xi_{t',j}
\given H_{t'}^{X,Y,Z},A_{t'+r_K}^{X,Y,Z}
\big]\\
&\quad=
(\hat W_t^{-t,t'}-W_t)
(\hat W_{t'}^{-t'}-\hat W_{t'}^{-t,t'})
\varepsilon_t\boldsymbol\xi_{t,j}\boldsymbol\xi_{t',j}
\E[\varepsilon_{t'}\given
H_{t'}^{X,Y,Z},A_{t'+r_K}^{X,Y,Z}].
\end{align*}
The deletion difference is also bounded by \(C_W-c_W\). Therefore
\begin{align}
&\frac2K\sum_{\substack{1\leq t,t'\leq K:\\
|t-t'|\geq l_K+r_K}}
\Big|\E\big[
(\hat W_t^{-t,t'}-W_t)
(\hat W_{t'}^{-t'}-\hat W_{t'}^{-t,t'})
\varepsilon_t\varepsilon_{t'}
\boldsymbol\xi_{t,j}\boldsymbol\xi_{t',j}
\big]\Big|\notag\\
&\quad\leq2(C_W-c_W)
\E\bigg[\Big\{\frac1K\sum_{t=1}^K
\varepsilon_t^2\boldsymbol\xi_{t,j}^2
\max_{1\leq t'\leq K}
(\hat W_t^{-t,t'}-W_t)^2\Big\}^{2}\bigg]^{1/4}
\notag\\*
&\hspace{33mm}{}\times
\Big\{K\sum_{t=1}^K
\E\big[\E[\varepsilon_t\given
H_t^{X,Y,Z},A_{t+r_K}^{X,Y,Z}]^2\big]\Big\}^{1/2}
\Big\{\max_{1\leq t\leq K}
\E\|\boldsymbol\xi_t\|_2^4\Big\}^{1/4}\notag\\
&\quad=o(1).
\label{eq:weight-V-mixed-bound}
\end{align}
Thus the sums of the absolute expectations of the first two
pairwise contributions are \(o(1)\).
For the third sum in \eqref{eq:weight-V-pair-expansion},
Cauchy--Schwarz gives
\begin{align}
&\E\Big[
\Big|\frac1K\sum_{\substack{1\leq t,t'\leq K:\\
|t-t'|\geq l_K+r_K}}
(\hat W_t^{-t}-\hat W_t^{-t,t'})
(\hat W_{t'}^{-t'}-\hat W_{t'}^{-t,t'})
\varepsilon_t\varepsilon_{t'}
\boldsymbol\xi_{t,j}\boldsymbol\xi_{t',j}
\Big|
\Big]\notag\\
&\quad\leq\frac1K
\E\Big[
\Big(\sum_{t=1}^K
|\varepsilon_t|\|\boldsymbol\xi_t\|_2
\max_{1\leq t'\leq K}
|\hat W_t^{-t}-\hat W_t^{-t,t'}|\Big)^2
\Big]\notag\\
&\quad\leq
\E\Big[
\sum_{t=1}^K\varepsilon_t^2\|\boldsymbol\xi_t\|_2^2
\max_{1\leq t'\leq K}
(\hat W_t^{-t}-\hat W_t^{-t,t'})^2
\Big]\notag\\
&\quad=K\E[\tilde D_W]\notag\\
&\quad=o(1).
\label{eq:weight-V-double-deletion-probability}
\end{align}
Taking expectations in \eqref{eq:weight-V-square-expansion}
and combining these bounds gives
\begin{align*}
0
&\leq
\E\Big[
\Big\{\frac1{\sqrt K}\sum_{t=1}^K
(\hat W_t^{-t}-W_t)
\varepsilon_t\boldsymbol\xi_{t,j}\Big\}^{2}
\Big]\\
&\leq
2\{2(l_K+r_K)-1\}\E[B_W+\tilde S_W]
+K\E[\tilde D_W]+o(1)\\
&=o(1).
\end{align*}
Markov's inequality, fixed \(d\), and
\eqref{eq:weight-V-first-deletion} therefore give
\begin{equation}
\label{eq:weight-V-convergence}
\mathrm V
=\frac1{\sqrt K}\sum_{t=1}^K
(\hat W_t^{-t}-W_t)\varepsilon_t\boldsymbol\xi_t+o_p(1)
=o_p(1).
\end{equation}
\end{proof}

\begin{lemma}[Weight--regression interaction remainders]
\label{lem:weight-regression-remainders}
Suppose \(H_0\), Assumptions~\ref{ass:mixing}, \ref{ass:stability},
\ref{ass:errors}, \ref{ass:population_weights}(i) and
\ref{ass:estimated_weights}(i)--(ii) hold, together with the
measurability conditions stated in the local-evaluation remark.
Then \(\mathrm{VI}=o_p(1)\) and \(\mathrm{VII}=o_p(1)\).
\end{lemma}

\begin{proof}
Fix \(j\in\{1,\ldots,d\}\). For term~\(\mathrm{VI}\),
\begin{align}
&\frac1{\sqrt K}\sum_{t=1}^K
(\hat W_t-W_t)\varepsilon_t
\big\{g_{t,j}(H_t^{Y,Z})-\hat g_{t,j}(H_t^{Y,Z})\big\}
\notag\\
&\quad=\frac1{\sqrt K}\sum_{t=1}^K
(\hat W_t-\hat W_t^{-t})\varepsilon_t
\big\{g_{t,j}(H_t^{Y,Z})-\hat g_{t,j}(H_t^{Y,Z})\big\}
\notag\\
&\qquad{}+\frac1{\sqrt K}\sum_{t=1}^K
(\hat W_t^{-t}-W_t)\varepsilon_t
\big\{\hat g_{t,j}^{-t}(H_t^{Y,Z})
-\hat g_{t,j}(H_t^{Y,Z})\big\}
\notag\\
&\qquad{}+\frac1{\sqrt K}\sum_{t=1}^K
(\hat W_t^{-t}-W_t)\varepsilon_t
\big\{g_{t,j}(H_t^{Y,Z})
-\hat g_{t,j}^{-t}(H_t^{Y,Z})\big\}.
\label{eq:weight-interaction-first-deletion}
\end{align}
The first two sums satisfy
\begin{align}
&\Big|\frac1{\sqrt K}\sum_{t=1}^K
(\hat W_t-\hat W_t^{-t})\varepsilon_t
\big\{g_{t,j}(H_t^{Y,Z})-\hat g_{t,j}(H_t^{Y,Z})\big\}\Big|
\leq(KS_WB_g)^{1/2}=o_p(1),
\label{eq:weight-interaction-first-weight-bound}\\
&\Big|\frac1{\sqrt K}\sum_{t=1}^K
(\hat W_t^{-t}-W_t)\varepsilon_t
\big\{\hat g_{t,j}^{-t}(H_t^{Y,Z})
-\hat g_{t,j}(H_t^{Y,Z})\big\}\Big|\notag\\
&\quad\leq(C_W-c_W)
\Big\{\Big(\frac1K\sum_{t=1}^K\varepsilon_t^2\Big)
KS_g\Big\}^{1/2}
=o_p(1).
\label{eq:weight-interaction-first-nuisance-bound}
\end{align}
Here \eqref{eq:weight-stability-moments}, \eqref{eq:stability},
\eqref{eq:Kmsemse} and the bounded residual second moments give
the probability limits.

For the square of the last sum in
\eqref{eq:weight-interaction-first-deletion}, the nearby pairs
are bounded as in Lemma~\ref{lem:exposure-remainder-moments}:
\begin{align*}
&\frac1K\sum_{\substack{1\leq t,t'\leq K:\\
|t-t'|<l_K+r_K}}
\E\Big[
\Big|(\hat W_t^{-t}-W_t)(\hat W_{t'}^{-t'}-W_{t'})
\varepsilon_t\varepsilon_{t'}
\big\{g_{t,j}(H_t^{Y,Z})
-\hat g_{t,j}^{-t}(H_t^{Y,Z})\big\}\\
&\hspace{48mm}{}\times
\big\{g_{t',j}(H_{t'}^{Y,Z})
-\hat g_{t',j}^{-t'}(H_{t'}^{Y,Z})\big\}
\Big|
\Big]\\
&\quad\leq
2(C_W-c_W)^2\{2(l_K+r_K)-1\}\E[B_g+\tilde S_g]\\
&\quad=o(1).
\end{align*}
For separated pairs, the additional deletion gives
\begin{align}
&(\hat W_t^{-t}-W_t)
\big\{g_{t,j}(H_t^{Y,Z})
-\hat g_{t,j}^{-t}(H_t^{Y,Z})\big\}
-(\hat W_t^{-t,t'}-W_t)
\big\{g_{t,j}(H_t^{Y,Z})
-\hat g_{t,j}^{-t,t'}(H_t^{Y,Z})\big\}
\notag\\
&\quad=
(\hat W_t^{-t}-\hat W_t^{-t,t'})
\big\{g_{t,j}(H_t^{Y,Z})
-\hat g_{t,j}^{-t,t'}(H_t^{Y,Z})\big\}
+
(\hat W_t^{-t}-W_t)
\big\{\hat g_{t,j}^{-t,t'}(H_t^{Y,Z})
-\hat g_{t,j}^{-t}(H_t^{Y,Z})\big\}.
\label{eq:weight-interaction-second-deletion}
\end{align}
Apply the pair expansion in \eqref{eq:exposure-pair-expansion}
to these products. As above, conditioning includes any independent
fitting randomization.
Projecting \(\varepsilon_{t'}\) onto
\(H_{t'}^{X,Y,Z},A_{t'+r_K}^{X,Y,Z}\), then applying the same
Cauchy--Schwarz and H\"older bounds as in
\eqref{eq:exposure-separated-bound}--\eqref{eq:exposure-mixed-bound},
gives
\begin{align}
&\frac2K\sum_{\substack{1\leq t<t'\leq K:\\
t'-t\geq l_K+r_K}}
\Big|\E\Big[
(\hat W_t^{-t,t'}-W_t)(\hat W_{t'}^{-t,t'}-W_{t'})
\varepsilon_t\varepsilon_{t'}\notag\\*
&\hspace{40mm}{}\times
\big\{g_{t,j}(H_t^{Y,Z})
-\hat g_{t,j}^{-t,t'}(H_t^{Y,Z})\big\}
\big\{g_{t',j}(H_{t'}^{Y,Z})
-\hat g_{t',j}^{-t,t'}(H_{t'}^{Y,Z})\big\}
\Big]\Big|\notag\\
&\quad{}+\frac2K\sum_{\substack{1\leq t,t'\leq K:\\
|t-t'|\geq l_K+r_K}}
\Big|\E\Big[
(\hat W_t^{-t,t'}-W_t)
(\hat W_{t'}^{-t'}-\hat W_{t'}^{-t,t'})
\varepsilon_t\varepsilon_{t'}\notag\\*
&\hspace{40mm}{}\times
\big\{g_{t,j}(H_t^{Y,Z})
-\hat g_{t,j}^{-t,t'}(H_t^{Y,Z})\big\}
\big\{g_{t',j}(H_{t'}^{Y,Z})
-\hat g_{t',j}^{-t,t'}(H_{t'}^{Y,Z})\big\}
\Big]\Big|\notag\\
&\quad{}+\frac2K\sum_{\substack{1\leq t,t'\leq K:\\
|t-t'|\geq l_K+r_K}}
\Big|\E\Big[
(\hat W_t^{-t,t'}-W_t)(\hat W_{t'}^{-t'}-W_{t'})
\varepsilon_t\varepsilon_{t'}\notag\\*
&\hspace{40mm}{}\times
\big\{g_{t,j}(H_t^{Y,Z})
-\hat g_{t,j}^{-t,t'}(H_t^{Y,Z})\big\}
\big\{\hat g_{t',j}^{-t,t'}(H_{t'}^{Y,Z})
-\hat g_{t',j}^{-t'}(H_{t'}^{Y,Z})\big\}
\Big]\Big|\notag\\
&\quad\leq(C_W-c_W)^2
\Big\{K\sum_{t=1}^K
\E\big[\E[\varepsilon_t\given
H_t^{X,Y,Z},A_{t+r_K}^{X,Y,Z}]^2\big]\Big\}^{1/2}
\notag\\*
&\hspace{20mm}{}\times
\big\{\E[(B_g+\tilde S_g+\tilde D_g)^2]\big\}^{1/4}
\Big(
12\big\{\E[(A_g+S_g+D_g)^2]\big\}^{1/4}
+2\sqrt3\,\{\E[D_g^2]\}^{1/4}
\Big)\notag\\
&\quad=o(1).
\label{eq:weight-interaction-separated-bound}
\end{align}
Thus the sum of the absolute expectations of these projected
pairwise contributions is \(o(1)\).
The product of the two deletion differences is bounded without
conditioning, using \eqref{eq:weight-interaction-second-deletion}:
\begin{align}
&\frac1K\E\Bigg[
\Big(\sum_{t=1}^K|\varepsilon_t|
\max_{1\leq t'\leq K}\Big|
(\hat W_t^{-t}-W_t)
\big\{g_{t,j}(H_t^{Y,Z})
-\hat g_{t,j}^{-t}(H_t^{Y,Z})\big\}\notag\\*
&\hspace{25mm}{}-
(\hat W_t^{-t,t'}-W_t)
\big\{g_{t,j}(H_t^{Y,Z})
-\hat g_{t,j}^{-t,t'}(H_t^{Y,Z})\big\}
\Big|\Big)^2
\Bigg]\notag\\
&\quad\leq2K\E\Big[
D_W\Big\{\frac1K\sum_{t=1}^K\varepsilon_t^2
\max_{1\leq t'\leq K}
\big(g_{t,j}(H_t^{Y,Z})
-\hat g_{t,j}^{-t,t'}(H_t^{Y,Z})\big)^2\Big\}
\Big]\notag\\*
&\hspace{16mm}{}+
2(C_W-c_W)^2
\E\Big[
\sum_{t=1}^K\varepsilon_t^2
\max_{1\leq t'\leq K}
\big|\hat g_{t,j}^{-t,t'}(H_t^{Y,Z})
-\hat g_{t,j}^{-t}(H_t^{Y,Z})\big|^2
\Big]\notag\\
&\quad\leq
6K\E[D_W(B_g+\tilde S_g+\tilde D_g)]
+2(C_W-c_W)^2K\E[\tilde D_g]\notag\\
&\quad\leq
6K\E[D_W(B_g+\tilde S_g)]
+8(C_W-c_W)^2K\E[\tilde D_g]\notag\\
&\quad=o(1).
\label{eq:weight-interaction-double-difference-bound}
\end{align}
The penultimate inequality uses \(D_W\leq(C_W-c_W)^2\);
the last line follows from the additional expectation assumptions.
Taking expectations in the square expansion of the once-deleted
sum therefore gives
\begin{align*}
0
&\leq
\E\Bigg[
\Big\{\frac1{\sqrt K}\sum_{t=1}^K
(\hat W_t^{-t}-W_t)\varepsilon_t
\big(g_{t,j}(H_t^{Y,Z})
-\hat g_{t,j}^{-t}(H_t^{Y,Z})\big)
\Big\}^2
\Bigg]\\
&\leq
2(C_W-c_W)^2\{2(l_K+r_K)-1\}\E[B_g+\tilde S_g]\\
&\qquad{}+
6K\E[D_W(B_g+\tilde S_g)]
+8(C_W-c_W)^2K\E[\tilde D_g]+o(1)\\
&=o(1).
\end{align*}
Markov's inequality, together with
\eqref{eq:weight-interaction-first-weight-bound}--
\eqref{eq:weight-interaction-first-nuisance-bound}, gives
\begin{equation}
\label{eq:termVI-converge}
\frac1{\sqrt K}\sum_{t=1}^K
(\hat W_t-W_t)\varepsilon_t
\big\{g_{t,j}(H_t^{Y,Z})
-\hat g_{t,j}(H_t^{Y,Z})\big\}
=o_p(1).
\end{equation}

For term~\(\mathrm{VII}\), the identities
\eqref{eq:weight-interaction-first-deletion} and
\eqref{eq:weight-interaction-second-deletion} hold with
\(g_{t,j}\) replaced by \(f_t\) and \(\varepsilon_t\) by
\(\boldsymbol\xi_{t,j}\). The first two sums satisfy
\begin{align*}
&\Big|\frac1{\sqrt K}\sum_{t=1}^K
(\hat W_t-\hat W_t^{-t})\boldsymbol\xi_{t,j}
\big\{f_t(H_t^{Y,Z})-\hat f_t(H_t^{Y,Z})\big\}\Big|
\leq(KS_WB_f)^{1/2}=o_p(1),\\
&\Big|\frac1{\sqrt K}\sum_{t=1}^K
(\hat W_t^{-t}-W_t)\boldsymbol\xi_{t,j}
\big\{\hat f_t^{-t}(H_t^{Y,Z})
-\hat f_t(H_t^{Y,Z})\big\}\Big|\\
&\quad\leq(C_W-c_W)
\Big\{\Big(\frac1K\sum_{t=1}^K\boldsymbol\xi_{t,j}^2\Big)
KS_f\Big\}^{1/2}
=o_p(1).
\end{align*}
For separated pairs, project \(\boldsymbol\xi_{t',j}\) onto
\(H_{t'}^{Y,Z},\phi(\bar X_t),A_{t'+r_K}^{Y,Z}\), as in
\eqref{eq:outcome-separated-bound}--\eqref{eq:outcome-mixed-bound}.
Assumption~\ref{ass:mixing}(iii) gives
\[
K\sum_{t=1}^K
\E\Big[
\max_{\substack{1\leq s\leq K:\\|s-t|\geq l_K+r_K}}
\big\|\E[\boldsymbol\xi_t\given
H_t^{Y,Z},\phi(\bar X_s),A_{t+r_K}^{Y,Z}]\big\|_2^2
\Big]
=o(1).
\]
The expectation of the sum of the absolute nearby-pair
contributions is bounded by
\[
2(C_W-c_W)^2\{2(l_K+r_K)-1\}\E[B_f+\tilde S_f]
=o(1).
\]
The projected-pair bound is the analogue of
\eqref{eq:weight-interaction-separated-bound} with \(f\) in place
of \(g\), using the source-residual projection above. Thus the
sum of the absolute expectations of the projected pairwise
contributions is \(o(1)\).
For the product of the deletion differences,
\begin{align}
&\frac1K\E\Bigg[
\Big(\sum_{t=1}^K|\boldsymbol\xi_{t,j}|
\max_{1\leq t'\leq K}\Big|
(\hat W_t^{-t}-W_t)
\big\{f_t(H_t^{Y,Z})-\hat f_t^{-t}(H_t^{Y,Z})\big\}
\notag\\*
&\hspace{40mm}{}-
(\hat W_t^{-t,t'}-W_t)
\big\{f_t(H_t^{Y,Z})-\hat f_t^{-t,t'}(H_t^{Y,Z})\big\}
\Big|\Big)^2
\Bigg]\notag\\
&\quad\leq
6K\E[D_W(B_f+\tilde S_f+\tilde D_f)]
+2(C_W-c_W)^2K\E[\tilde D_f]\notag\\
&\quad\leq
6K\E[D_W(B_f+\tilde S_f)]
+8(C_W-c_W)^2K\E[\tilde D_f]\notag\\
&\quad=o(1).
\label{eq:weight-outcome-interaction-bound}
\end{align}
Here we again use \(D_W\leq(C_W-c_W)^2\) and the additional
expectation assumptions.
Combining these bounds in the square expansion gives
\begin{align*}
0
&\leq
\E\Bigg[
\Big\{\frac1{\sqrt K}\sum_{t=1}^K
(\hat W_t^{-t}-W_t)\boldsymbol\xi_{t,j}
\big(f_t(H_t^{Y,Z})-\hat f_t^{-t}(H_t^{Y,Z})\big)
\Big\}^2
\Bigg]\\
&\leq
2(C_W-c_W)^2\{2(l_K+r_K)-1\}\E[B_f+\tilde S_f]\\
&\qquad{}+
6K\E[D_W(B_f+\tilde S_f)]
+8(C_W-c_W)^2K\E[\tilde D_f]+o(1)\\
&=o(1).
\end{align*}
Markov's inequality and the two first-deletion probability
bounds give
\begin{equation}
\label{eq:termVII-converge}
\frac1{\sqrt K}\sum_{t=1}^K
(\hat W_t-W_t)\boldsymbol\xi_{t,j}
\big\{f_t(H_t^{Y,Z})-\hat f_t(H_t^{Y,Z})\big\}
=o_p(1).
\end{equation}
Fixed \(d\) completes the proof for \(\mathrm{VI}\) and
\(\mathrm{VII}\).
\end{proof}

\begin{lemma}[Weight error times the product of nuisance errors]
\label{lem:weight-product-remainder}
Under the bounds in Assumptions~\ref{ass:population_weights}(i) and
\ref{ass:estimated_weights}(i), the product rate
\(A_fA_g=o_p(K^{-1})\) implies \(\mathrm{VIII}=o_p(1)\).
\end{lemma}

\begin{proof}
The triangle and Cauchy--Schwarz inequalities give
\begin{equation}
\label{eq:termVIII-product-bound}
\begin{aligned}
&\Big\|\frac1{\sqrt K}\sum_{t=1}^K(\hat W_t-W_t)
\big\{f_t(H_t^{Y,Z})-\hat f_t(H_t^{Y,Z})\big\}
\big\{g_t(H_t^{Y,Z})-\hat g_t(H_t^{Y,Z})\big\}\Big\|_2\\
&\quad\leq\frac{C_W-c_W}{\sqrt K}\sum_{t=1}^K
|f_t(H_t^{Y,Z})-\hat f_t(H_t^{Y,Z})|
\|g_t(H_t^{Y,Z})-\hat g_t(H_t^{Y,Z})\|_2\\
&\quad\leq(C_W-c_W)\sqrt K
\Big\{\frac1K\sum_{t=1}^K|f_t(H_t^{Y,Z})-\hat f_t(H_t^{Y,Z})|^2\Big\}^{1/2}
\Big\{\frac1K\sum_{t=1}^K\|g_t(H_t^{Y,Z})-\hat g_t(H_t^{Y,Z})\|_2^2\Big\}^{1/2}\\
&\quad=(C_W-c_W)\sqrt{K A_fA_g}\\
&\quad =o_p(1),
\end{aligned}
\end{equation}
by \eqref{eq:Kmsemse}.
\end{proof}

Under the stated assumptions, combining
Lemma~\ref{lem:leading-term} with the bounds for
terms~\(\mathrm{II}\)--\(\mathrm{VIII}\) gives
\begin{equation}
\label{eq:numerator}
T^{(K)}
=
\frac{1}{\sqrt K}\sum_{t=1}^K
W_t\varepsilon_t\boldsymbol\xi_t
+o_p(1)
\xrightarrow{d}\mathcal N_d(0,\Sigma).
\end{equation}

\subsection{Matrix covariance consistency}

\begin{lemma}[Consistency of the covariance estimator]
\label{lem:matrix-covariance-consistency}
Under the conditions of Theorem~\ref{thm:vector_gtcm},
\[
\hat\Sigma\xrightarrow{p}\Sigma.
\]
\end{lemma}

\begin{proof}
{For every \(1\leq j,k\leq d\), the martingale differences
\[
W_t^2\varepsilon_t^2\boldsymbol\xi_{t,j}\boldsymbol\xi_{t,k}
-\E\big[W_t^2\varepsilon_t^2\boldsymbol\xi_{t,j}\boldsymbol\xi_{t,k}
\given H_t^{X,Y,Z}\big]
\]
have uniformly bounded \(1+\delta/2\) moments by the moment calculation
in Lemma~\ref{lem:leading-term}. The martingale weak law, applied entrywise, gives
\begin{equation}
\label{eq:oracle-realized-minus-conditional}
\frac1K\sum_{t=1}^K
W_t^2\varepsilon_t^2\boldsymbol\xi_t\boldsymbol\xi_t^\top
-\frac1K\sum_{t=1}^K
\E\big[W_t^2\varepsilon_t^2\boldsymbol\xi_t\boldsymbol\xi_t^\top
\given H_t^{X,Y,Z}\big]\xrightarrow{p}0.
\end{equation}}
Combining \eqref{eq:oracle-realized-minus-conditional} with
\eqref{eq:oracle-conditional-covariance-limit} gives
\begin{equation}
\label{eq:oracle-second-moment-limit}
\frac{1}{K}\sum_{t=1}^K
W_t^2\varepsilon_t^2
\boldsymbol\xi_t\boldsymbol\xi_t^\top
\xrightarrow{p}\Sigma.
\end{equation}

It remains to replace the population residuals and weights by their estimated
counterparts. For every fixed \(a\in\mathbb R^d\), expand
\[
\hat\varepsilon_t
=
\varepsilon_t+
f_t(H_t^{Y,Z})-\hat f_t(H_t^{Y,Z}),
\quad
a^\top\hat{\boldsymbol\xi}_t
=
a^\top\boldsymbol\xi_t+
a^\top\{g_t(H_t^{Y,Z})-\hat g_t(H_t^{Y,Z})\}.
\]
The squared triangle inequality and boundedness of the weights give
\[
\begin{aligned}
&\frac1K\sum_{t=1}^K
\big\|\hat W_t\hat\varepsilon_t\hat{\boldsymbol\xi}_t
      -W_t\varepsilon_t\boldsymbol\xi_t\big\|_2^2\\
&\qquad\leq4B_W+4C_W^2\{B_f+B_g+K A_fA_g\}\\
&\qquad=o_p(1),
\end{aligned}
\]
The squared product of the nuisance errors contributes at most \(K A_fA_g\), \eqref{eq:residual-weighted-errors} and \eqref{eq:Kmsemse} and \eqref{eq:weight-consistency} give \(B_f+B_g+B_W=o_p(1)\) and \(K A_fA_g=o_p(1)\). Together with \eqref{eq:oracle-second-moment-limit}, Cauchy--Schwarz then gives
\begin{equation}
\label{eq:feasible-oracle-second-moment}
\begin{split}
\frac{1}{K}\sum_{t=1}^K
\hat W_t^2\hat\varepsilon_t^2
(a^\top\hat{\boldsymbol\xi}_t)^2
-
\frac{1}{K}\sum_{t=1}^K
W_t^2\varepsilon_t^2
(a^\top\boldsymbol\xi_t)^2
\xrightarrow{p}0.
\end{split}
\end{equation}
For the weight error at the population residuals,
\[
\Big\|\frac1K\sum_{t=1}^K(\hat W_t^2-W_t^2)
\varepsilon_t^2\boldsymbol\xi_t\boldsymbol\xi_t^\top\Big\|_F
\leq2C_W B_W^{1/2}
\Big\{\frac1K\sum_{t=1}^K\varepsilon_t^2\|\boldsymbol\xi_t\|_2^2\Big\}^{1/2}
=o_p(1),
\]
by \eqref{eq:weight-consistency} and the residual moment bound. Apply \eqref{eq:feasible-oracle-second-moment} to the finite collection
\(a=e_j\) and \(a=e_j+e_k\). Polarization and fixed \(d\) then give
\begin{equation}
\label{eq:feasible-oracle-second-moment-matrix}
\frac{1}{K}\sum_{t=1}^K
\hat W_t^2\hat\varepsilon_t^2
\hat{\boldsymbol\xi}_t\hat{\boldsymbol\xi}_t^\top
-
\frac{1}{K}\sum_{t=1}^K
W_t^2\varepsilon_t^2
\boldsymbol\xi_t\boldsymbol\xi_t^\top
\xrightarrow{p}0.
\end{equation}
Finally,
\[
\frac{1}{K}\sum_{t=1}^K
\hat W_t\hat\varepsilon_t\hat{\boldsymbol\xi}_t
=
\frac{T^{(K)}}{\sqrt K}
=O_p(K^{-1/2})
\]
by \eqref{eq:numerator}. Therefore the outer product of the sample mean
subtracted in \eqref{eq:covariance_estimator} is \(o_p(1)\).
Equations~\eqref{eq:oracle-second-moment-limit} and
\eqref{eq:feasible-oracle-second-moment-matrix} now imply
\(\hat\Sigma\xrightarrow{p}\Sigma\).
\end{proof}

Since \(\Sigma\) is positive definite, Lemma~\ref{lem:matrix-covariance-consistency},
\eqref{eq:numerator}, and Slutsky's theorem give
\[
\{T^{(K)}\}^\top\hat\Sigma^{-1}T^{(K)}
\xrightarrow{d}\chi_d^2.
\]

\section{Proof of Theorem \ref{thm:local_power}}
\label{app:local-power-proof}

\begin{proof}
All population quantities are evaluated under the \(K\)th local
alternative. By \eqref{eq:local_alt},
\[
\varepsilon_t=\zeta_t+\frac{\boldsymbol\xi_t^\top\tau}{\sqrt K},\qquad
\E[\zeta_t\given H_t^{X,Y,Z}]=0.
\]
The assumptions of Theorem~\ref{thm:local_power} give the bounds in
Appendix~\ref{app:nuisance-remainders} with \(\zeta_t\) replacing
\(\varepsilon_t\). Consequently,
\begin{equation}
\label{eq:local-centred-remainder}
\frac1{\sqrt K}\sum_{t=1}^K
\Big[\hat W_t\Big\{\hat\varepsilon_t-
\frac{\boldsymbol\xi_t^\top\tau}{\sqrt K}\Big\}
\hat{\boldsymbol\xi}_t-W_t\zeta_t\boldsymbol\xi_t\Big]
=o_p(1).
\end{equation}

The remaining difference in the local drift satisfies
\begin{align*}
&\Big\|\frac1K\sum_{t=1}^K
(\boldsymbol\xi_t^\top\tau)
\big(\hat W_t\hat{\boldsymbol\xi}_t-W_t\boldsymbol\xi_t\big)\Big\|_2\\
&\quad\leq\|\tau\|_2 A_W^{1/2}
\Big\{\frac1K\sum_t\|\boldsymbol\xi_t\|_2^4\Big\}^{1/2}
+C_W\|\tau\|_2 A_g^{1/2}
\Big\{\frac1K\sum_t\|\boldsymbol\xi_t\|_2^2\Big\}^{1/2} \\
&\quad =o_p(1),
\end{align*}
by \eqref{eq:weight-stability-moments}, $A_g = o_p(1)$ and
Assumption~\ref{ass:mixing}(iv). Thus
\begin{equation}
\label{eq:local-expansion}
T^{(K)}
=\frac1{\sqrt K}\sum_{t=1}^K W_t\zeta_t\boldsymbol\xi_t
+\Big(\frac1K\sum_{t=1}^K
W_t\boldsymbol\xi_t\boldsymbol\xi_t^\top\Big)\tau+o_p(1).
\end{equation}
Lemma~\ref{lem:leading-term} and Proposition~\ref{prop:Sigma_xi} give
\begin{equation}
\label{eq:local-drift-martingale}
\frac1{\sqrt K}\sum_{t=1}^K W_t\zeta_t\boldsymbol\xi_t
\xrightarrow{d}\mathcal N_d(0,\Sigma),\quad
\frac1K\sum_{t=1}^K W_t\boldsymbol\xi_t\boldsymbol\xi_t^\top
\xrightarrow{p}\Gamma.
\end{equation}
It follows that \(T^{(K)}\xrightarrow{d}\mathcal N_d(\Gamma\tau,\Sigma)\). The oracle second-moment argument in
Lemma~\ref{lem:matrix-covariance-consistency}, applied to \(\zeta_t\), gives
\begin{equation}
\label{eq:local-oracle-second-moment}
\frac1K\sum_tW_t^2\zeta_t^2
\boldsymbol\xi_t\boldsymbol\xi_t^\top\xrightarrow{p}\Sigma.
\end{equation}
Throughout this proof, $B_g,\tilde S_g,\tilde D_g$ and
$B_W,\tilde S_W,\tilde D_W$ are defined with $\zeta_t^2$
in place of $\varepsilon_t^2$. The corresponding conditions
are assumed along the local-alternative sequence.
The $f$-quantities retain the multiplier $\|\boldsymbol\xi_t\|_2^2$. For the test statistic, expanding the residuals yields
\begin{align}
&\frac1K\sum_t\big\|
\hat W_t\hat\varepsilon_t\hat{\boldsymbol\xi}_t
-W_t\zeta_t\boldsymbol\xi_t\big\|_2^2\notag\\
&\quad\leq5B_W+5C_W^2(B_g+B_f+K A_fA_g)
+\frac5{K^2}\sum_t\hat W_t^2
(\boldsymbol\xi_t^\top\tau)^2\|\hat{\boldsymbol\xi}_t\|_2^2.
\label{eq:local-feasible-score-distance}
\end{align}
The last term is negligible because
\begin{align*}
&\frac1{K^2}\sum_t\hat W_t^2
(\boldsymbol\xi_t^\top\tau)^2\|\hat{\boldsymbol\xi}_t\|_2^2\\
&\quad\leq2C_W^2\|\tau\|_2^2
\bigg\{\frac1{K^2}\sum_t\|\boldsymbol\xi_t\|_2^4
+A_g\frac1K\sum_t\|\boldsymbol\xi_t\|_2^2\bigg\}\\
&\quad =o_p(1).
\end{align*}
Here \(B_g,B_W\) use \(\zeta_t\), as specified in the theorem.
The other terms in \eqref{eq:local-feasible-score-distance} tend to zero
by the prediction and weight assumptions. Cauchy--Schwarz and \eqref{eq:local-oracle-second-moment} now give
\begin{align*}
&\Big\|\frac1K\sum_t
\hat W_t^2\hat\varepsilon_t^2
\hat{\boldsymbol\xi}_t\hat{\boldsymbol\xi}_t^\top
-\frac1K\sum_tW_t^2\zeta_t^2
\boldsymbol\xi_t\boldsymbol\xi_t^\top\Big\|_F\\
&\quad\leq O_p(1)
\bigg\{\frac1K\sum_t\big\|
\hat W_t\hat\varepsilon_t\hat{\boldsymbol\xi}_t
-W_t\zeta_t\boldsymbol\xi_t\big\|_2^2\bigg\}^{1/2}\\
&\quad =o_p(1).
\end{align*}
Since \(T^{(K)}=O_p(1)\),
\[
\hat\Sigma
=\frac1K\sum_t\hat W_t^2\hat\varepsilon_t^2
\hat{\boldsymbol\xi}_t\hat{\boldsymbol\xi}_t^\top
-\frac{T^{(K)}\{T^{(K)}\}^\top}{K}
\xrightarrow{p}\Sigma.
\]
Slutsky's theorem and positive definiteness of \(\Sigma\) imply
\[
Q^{(K)}\xrightarrow{d}\chi_d^2\{\lambda(\tau)\},\quad
\lambda(\tau)=\tau^\top\Gamma^\top\Sigma^{-1}\Gamma\tau.
\]
The limiting distribution is continuous, so under $H_{1,K}(\tau)$,
\[
\Pr\big(Q^{(K)}>\chi^2_{d,1-\alpha}\big)
\rightarrow
\Pr\big(\chi_d^2\{\lambda(\tau)\}>\chi^2_{d,1-\alpha}\big).
\]
\end{proof}

\section{Examples for the Weak Dependence Assumption}
\label{app:mixing-example}

Use $L^{Y,Z}_{a:b}=\{Y_{a:b},Z_{a:b}\}$, as in the main text.
Throughout this section, the response condition is under the null
\begin{equation}\label{eq:null-recalled}
 \E[Y_{t+1}\given H_t^{X,Y,Z}]=\E[Y_{t+1}\given H_t^{Y,Z}].
\end{equation}
The two requirements to be verified are
\begin{align}
 \sum_{t=1}^K\E \big[
  \E[\varepsilon_t\given H_t^{X,Y,Z},A_{t+r_K}^{X,Y,Z}]^2
 \big]&=o(K^{-1}),\\
 \sum_{t=1}^K\E \Big[
  \max_{\substack{1\leq s\leq K:\\|s-t|\geq l_K+r_K}}
  \big\|\E[\boldsymbol\xi_t\given H_t^{Y,Z},
          \phi(\bar X_s),A_{t+r_K}^{Y,Z}]\big\|_2^2
 \Big]&=o(K^{-1}).
\end{align}

\begin{lemma}[Stationary Gaussian autoregressions]
\label{lem:gaussian-ar-short-memory}
Let $V_t=(X_t,Y_t,Z_t^\top)^\top$ be a strictly stationary
Gaussian VAR$(l)$ process,
\begin{equation}
 V_t=A_1V_{t-1}+\cdots+A_lV_{t-l}+\nu_t,
\end{equation}
where the innovations $\nu_t$ are centred, independent and identically distributed,
with nonsingular covariance. Suppose $\phi(\bar X_t)$ is measurable with respect
to $L^X_{(t-l+1):t}$ and $\E\big[\|\phi(\bar X_t)\|_2^2\big]<\infty$.
Under \eqref{eq:null-recalled}, there is $\rho\in(0,1)$ such that
\eqref{eq:epsilon-mixing}--\eqref{eq:xi-mixing} hold whenever
\begin{equation}\label{eq:gap-rate}
 K^3\rho^{2r_K}\rightarrow0.
\end{equation}
For example, $r_K=\lceil C\log K\rceil,$ with constant $C>\frac{3}{2|\log\rho|}.$
\end{lemma}

\begin{proof}
Stationarity and nonsingular innovations imply exponential covariance decay and uniformly bounded covariance eigenvalues. By the inverse-decay result of \citet[Proposition~2]{jaffard1990localisees}, the inverse covariance matrices of the adjustment histories also decay exponentially, uniformly in the history length. Thus, for some \(\rho_0\in(0,1)\),
\begin{equation}\label{eq:gaussian-covariance-bounds}
|\Cov(V_i,V_j)|_{\mathrm{op}}\leq C\rho_0^{|i-j|},
\quad cI\preceq\Sigma\preceq CI.
\end{equation}
Here \(\Sigma\) denotes the covariance matrix of any finite collection of distinct coordinates. Set \(\rho=\sqrt{\rho_0}\), as used in \eqref{eq:gap-rate}; since \(\rho_0<\rho<1\), the same covariance-decay bound also holds with \(\rho\).

\smallskip\noindent\textit{\textbf{Response residual.}}
The VAR equation and the null give
\[
 \varepsilon_t=\nu_{Y,t+1},\quad
 \varepsilon_t \indep H_t^{X,Y,Z},\quad \E[\varepsilon_t]=0.
\]
The lower bound in \eqref{eq:gaussian-covariance-bounds} also holds for
conditional covariance matrices. Indeed, for any two disjoint Gaussian
coordinate blocks and any $v$,
\begin{equation}\label{eq:conditional-covariance-lower-bound}
 v^\top(\Sigma_{22}-\Sigma_{21}\Sigma_{11}^{-1}\Sigma_{12})v
 =\min_u (u^\top,v^\top)\Sigma(u^\top,v^\top)^\top
 \geq c\|v\|_2^2.
\end{equation}
Since the conditional mean of $\varepsilon_t$ given the history is zero,
Gaussian regression gives
\begin{align*}
 &\E\big[\E[\varepsilon_t\given H_t^{X,Y,Z},A_{t+r_K}^{X,Y,Z}]^2\big]\\
 &\quad=\Cov(\varepsilon_t,A_{t+r_K}^{X,Y,Z})
       \Var(A_{t+r_K}^{X,Y,Z}\given H_t^{X,Y,Z})^{-1}
       \Cov(A_{t+r_K}^{X,Y,Z},\varepsilon_t)\\
 &\quad\leq c^{-1}\|\Cov(\varepsilon_t,A_{t+r_K}^{X,Y,Z})\|_2^2 \\
 &\quad \leq C\sum_{j=t+r_K+1}^{K+1}\rho^{2(j-t-1)} \\
 & \quad \leq C\rho^{2r_K}.
\end{align*}
Here the stationary VAR representation gives
$\|\Cov(\nu_{Y,t+1},V_j)\|_2\leq C\rho^{j-t-1}$ for $j>t$.
Consequently,
\begin{equation}\label{eq:response-rate}
 \sum_{t=1}^K\E\big[
 \E[\varepsilon_t\given H_t^{X,Y,Z},A_{t+r_K}^{X,Y,Z}]^2\big]
 \leq CK\rho^{2r_K}=o(K^{-1}).
\end{equation}

\smallskip\noindent\textit{\textbf{Source residual.}}
Retain the full adjustment history and future. For the Gaussian vector
\begin{equation}\label{eq:source-blocks}
 \begin{pmatrix}
 H_t^{Y,Z}\\L^X_{(t-l+1):t}\\L^X_{(s-l+1):s}\\
 L^{Y,Z}_{(t+r_K+1):(K+1)}
 \end{pmatrix},~ \text{where}~|s-t|\geq l_K+r_K,
\end{equation}
write its covariance as $(\Sigma_{ab})_{a,b=1}^4$. The conditional
covariance and cross-covariance needed below are
\begin{align*}
 S&=\begin{pmatrix}\Sigma_{33}&\Sigma_{34}\\
                         \Sigma_{43}&\Sigma_{44}\end{pmatrix}
       -\begin{pmatrix}\Sigma_{31}\\\Sigma_{41}\end{pmatrix}
          \Sigma_{11}^{-1}(\Sigma_{13}\ \Sigma_{14}),\\
 \Psi&=(\Sigma_{23}\ \Sigma_{24})
          -\Sigma_{21}\Sigma_{11}^{-1}(\Sigma_{13}\ \Sigma_{14}).
\end{align*}
With the common decay constant chosen above, the inverse-decay bound is
\begin{equation}\label{eq:inverse-decay}
 \|(\Sigma_{11}^{-1})_{ij}\|_{\mathrm{op}}
 \leq C\rho_0^{|i-j|},
\end{equation}
uniformly over the history length.
The indices here identify time blocks. Combining
\eqref{eq:gaussian-covariance-bounds} and \eqref{eq:inverse-decay},
for a time $i$ in block 2 and a time $j$ in block 3 or 4,
\begin{align*}
 \|\Psi_{ij}\|_{\mathrm{op}}
 &\leq C\rho_0^{|i-j|}
   +C\sum_{u,v\leq t}
       \rho_0^{|i-u|+|u-v|+|v-j|}\\
 &\leq C\big(1+|i-j|\big)^2\rho_0^{|i-j|}\\
 &\leq C\rho_0^{|i-j|/2}
  = C\rho^{|i-j|}.
\end{align*}
Since block 2 has fixed
length and all coordinates in blocks 3 and 4 are separated from it,
\begin{equation}\label{eq:source-cross-covariance}
 \|\Psi\|_F^2
 \leq C\Big\{\rho^{2\big(|s-t|-l+1\big)}+
                    \sum_{j=t+r_K+1}^{K+1}\rho^{2(j-t)}\Big\}
 \leq C\rho^{2r_K}.
\end{equation}
By \eqref{eq:conditional-covariance-lower-bound},
\[
 S\succeq cI,\quad
 \Sigma_{22}-\Sigma_{21}\Sigma_{11}^{-1}\Sigma_{12}\succeq cI.
\]
Therefore, the largest conditional canonical correlation is at most
\begin{equation}\label{eq:canonical-bound}
 \|(\Sigma_{22}-\Sigma_{21}\Sigma_{11}^{-1}\Sigma_{12})^{-1/2}
           \Psi S^{-1/2}\|_{\mathrm{op}}
 \leq c^{-1}\|\Psi\|_F\leq C\rho^{r_K}.
\end{equation}
Conditional Gaussian maximal correlation \citep{veraar2009correlation}, applied
coordinatewise, now yields
\begin{align*}
 &\E \Big[\big\|\E[\boldsymbol\xi_t\given H_t^{Y,Z},
 L^X_{(s-l+1):s},A_{t+r_K}^{Y,Z}]\big\|_2^2\given H_t^{Y,Z}\Big]\\
 &\quad\leq C\rho^{2r_K}\E\big[\|\boldsymbol\xi_t\|_2^2\given H_t^{Y,Z}\big].
\end{align*}
Indeed $\boldsymbol\xi_t$ is conditionally centred; for jointly Gaussian
vectors the $L^2$ norm of the conditional-expectation operator on centred
functions equals their largest canonical correlation. Since
$\phi(\bar X_s)$ is measurable with respect to $L^X_{(s-l+1):s}$,
the tower property and conditional Jensen's inequality give
\begin{align}\label{eq:source-one-pair}
 &\E \Big[\big\|\E[\boldsymbol\xi_t\given H_t^{Y,Z},
                  \phi(\bar X_s),A_{t+r_K}^{Y,Z}]\big\|_2^2\Big]\notag\\
 &\quad\leq\E \Big[\big\|\E[\boldsymbol\xi_t\given H_t^{Y,Z},
                  L^X_{(s-l+1):s},A_{t+r_K}^{Y,Z}]\big\|_2^2\Big]\notag\\
 &\quad\leq C\rho^{2r_K}\E\big[\|\boldsymbol\xi_t\|_2^2\big] \notag\\
 &\quad\leq C\rho^{2r_K}\E\big[\|\phi(\bar X_t)\|_2^2\big]\notag \\
 &\quad\leq C\rho^{2r_K}.
\end{align}
Finally, the maximum is inside the expectation, so
\begin{align}\label{eq:source-max-rate}
 &\sum_{t=1}^K\E \Big[
 \max_{\substack{1\leq s\leq K:\\|s-t|\geq l_K+r_K}}
 \big\|\E[\boldsymbol\xi_t\given H_t^{Y,Z},
             \phi(\bar X_s),A_{t+r_K}^{Y,Z}]\big\|_2^2\Big]\notag\\
 &\quad\leq\sum_{t=1}^K
 \sum_{\substack{1\leq s\leq K:\\|s-t|\geq l_K+r_K}}
 \E \Big[\big\|\E[\boldsymbol\xi_t\given H_t^{Y,Z},
             \phi(\bar X_s),A_{t+r_K}^{Y,Z}]\big\|_2^2\Big]\notag\\
 &\quad\leq CK^2\rho^{2r_K} \notag\\
 &\quad=\frac{C}{K}\{K^3\rho^{2r_K}\}\notag\\
 &\quad =o(K^{-1}).
\end{align}
\end{proof}

\begin{lemma}[Moving averages]
\label{lem:gaussian-ma-short-memory}
Let
\[
 V_t=\nu_t+B_1\nu_{t-1}+\cdots+B_q\nu_{t-q},
\]
where the innovations are independent Gaussian, with finite second moments.
Under \eqref{eq:null-recalled}, the left-hand side of
\eqref{eq:epsilon-mixing} is zero whenever $r_K>q$.
For \eqref{eq:xi-mixing}, $cI\preceq\Sigma\preceq CI$ uniformly for finite collections of
distinct observed coordinates, and $\phi(\bar X_t)$ is measurable with respect to $L^X_{(t-q+1):t}$, and $\sup_t\E\big[\|\phi(\bar X_t)\|_2^2\big]<\infty.$
Then \eqref{eq:xi-mixing} holds under \eqref{eq:gap-rate} for some
$\rho\in(0,1)$. 
\end{lemma}

\begin{proof}
If $r_K>q$, the innovations generating $A_{t+r_K}^{X,Y,Z}$ occur strictly
after $t+1$. Hence
\begin{align*}
 A_{t+r_K}^{X,Y,Z}& \indep (Y_{t+1},H_t^{X,Y,Z}),\\
 \E[\varepsilon_t\given H_t^{X,Y,Z},A_{t+r_K}^{X,Y,Z}]
 &=\E[\varepsilon_t\given H_t^{X,Y,Z}]=0.
\end{align*}
For the source term, use the blocks in \eqref{eq:source-blocks}. Even if
$l_K>q$ and $r_K>q$, finite dependence gives only
\[
 \Sigma_{23}=0,\quad \Sigma_{24}=0,\quad
 \Psi=-\Sigma_{21}\Sigma_{11}^{-1}(\Sigma_{13}\ \Sigma_{14}),
\]
which need not be zero. Conditioning on the full $H_t^{Y,Z}$ is essential.
Under the additional Gaussian and covariance assumptions,
\[
 \Cov(V_i,V_j)=0~ \big(|i-j|>q\big),\quad
 \|(\Sigma_{11}^{-1})_{ij}\|_{\mathrm{op}}\leq C\rho^{|i-j|}.
\]
The second inequality follows from the same inverse-decay theorem.
Thus \eqref{eq:source-cross-covariance}--\eqref{eq:source-max-rate} give
\[
 \sum_{t=1}^K\E\Big[
 \max_{\substack{1\leq s\leq K:\\|s-t|\geq l_K+r_K}}
 \big\|\E[\boldsymbol\xi_t\given H_t^{Y,Z},
                  \phi(\bar X_s),A_{t+r_K}^{Y,Z}]\big\|_2^2\Big]
 \leq CK^2\rho^{2r_K}=o(K^{-1}).
\]
\end{proof}

\section{Examples for the Stability Assumption}
\label{app:stability-example}

We give sufficient conditions for the convergence requirements
\begin{equation}
\label{eq:stability-recalled}
S_q=o_p(K^{-1}),
\quad
(l_K+r_K)\E[\tilde S_q]+K\E[\tilde D_q]=o(1),
\quad q\in\{f,g\},
\end{equation}
and the bounded second moments in
Assumption~\ref{ass:stability}.
The probability and expectation results below have separate
assumptions. We write
\[
\bar Y_{t,p}=(Y_t,\ldots,Y_{t-p+1})^\top,
\quad p=p_K=O(\log K).
\]
The adjustment variables $Z$ are suppressed. 

Deletion is performed at the level of primitive observations.
For
\[
\Delta_t=\{t+1,\ldots,t+r_K\}\cap\{1,\ldots,K+1\},
\quad l_K+r_K=O(\log K),
\]
let $r_t$ contain every regression row whose response or
predictors use an observation indexed by $\Delta_t$.
The deletion sets are determined by the chosen temporal
footprints and satisfy
\begin{equation}
\label{eq:linear-purged-size}
\max_t|r_t|=O(\log K),
\quad
\max_{t,t'}|r_t\cup r_{t'}|
\leq 2\max_t|r_t|=O(\log K).
\end{equation}
All deleted fits are evaluated at the original predictors.
Required initial lags and the final response are observed and
satisfy the stated moment conditions.

The proofs are written for the scalar source regression
$\phi(\bar X_t)=X_t$.
For a fixed-dimensional dictionary, the same arguments apply
coordinatewise under the stated conditions on
$\phi(\bar X_t)$.
For the outcome regression, the training response $X_j$ is
replaced by $Y_{j+1}$.

\subsection{Linear regression}
\label{app:linear-weighted-stability}

Consider the linear source regression
\[
X_t=g_t(\bar Y_{t,p})+\xi_t
   =\bar Y_{t,p}^\top\beta+\xi_t,
\quad
\E[\xi_t\given H_t^{Y,Z}]=0,
\]
where the deterministic coefficient vector may depend on $K$
but is constant across $t$.
Ordinary least squares gives
\[
\hat g_t(\bar Y_{t,p})=\bar Y_{t,p}^\top\hat\beta.
\]

\begin{lemma}[Probability stability of linear regression]
\label{lem:linear-regression-stability}
Suppose that, for some $\delta>0$,
\begin{equation}
\label{eq:linear-moment}
\sup_{K,j}
\E\big[
\|\phi(\bar X_j)\|_2^{4+\delta}
+|Y_j|^{4+\delta}
\big]<\infty,
\end{equation}
and, for some $c>0$,
\begin{equation}
\label{eq:linear-gram}
\Pr\Big[
\lambda_{\min}\big(
\frac1K\sum_{j=1}^K
\bar Y_{j,p}\bar Y_{j,p}^\top
\big)\geq c
\Big]\rightarrow1.
\end{equation}
Then ordinary least squares satisfies
\[
S_q
=
O_p\big(
p^3(\log K)^2K^{-2+4/(4+\delta)}
\big)
=o_p(K^{-1}),
\quad q\in\{f,g\}.
\]
\end{lemma}

\begin{proof}
We first consider $\phi(\bar X_t)=X_t$.
The union bound and Markov's inequality give
\begin{equation}
\label{eq:linear-max-bound}
\max_{1\leq j\leq K}|X_j|
=O_p(K^{1/(4+\delta)}),
\quad
\max_{1\leq j\leq K}\|\bar Y_{j,p}\|_2
=O_p(\sqrt p\,K^{1/(4+\delta)}).
\end{equation}
For the second bound, use
\[
\max_j\|\bar Y_{j,p}\|_2
\leq
\sqrt p\max_{2-p\leq j\leq K}|Y_j|,
\]
and note that $K+p-1=O(K)$.
Also,
\begin{equation}
\label{eq:linear-average-bound}
\frac1K\sum_jX_j^2=O_p(1),
\quad
\frac1K\sum_j\|\bar Y_{j,p}\|_2^2=O_p(p).
\end{equation}

The normal equations and the least-squares projection identity
imply
\[
\sum_j\bar Y_{j,p}
(X_j-\bar Y_{j,p}^\top\hat\beta)=0,
\quad
\sum_j(\bar Y_{j,p}^\top\hat\beta)^2
\leq\sum_jX_j^2.
\]
On the event in \eqref{eq:linear-gram},
\[
c\|\hat\beta\|_2^2
\leq
\hat\beta^\top
\big(
\frac1K\sum_j\bar Y_{j,p}\bar Y_{j,p}^\top
\big)\hat\beta
\leq
\frac1K\sum_jX_j^2.
\]
Hence $\|\hat\beta\|_2=O_p(1)$. For a single deletion,
\begin{align*}
\max_t
\Big\|
\sum_{j\in r_t}\bar Y_{j,p}\bar Y_{j,p}^\top
\Big\|_{\mathrm{op}}
&\leq
\max_t|r_t|\max_j\|\bar Y_{j,p}\|_2^2\\
&=
O_p\big(p\log K\,K^{2/(4+\delta)}\big)\\
&=o_p(K).
\end{align*}
Weyl's inequality and \eqref{eq:linear-gram} therefore give
\begin{equation}
\label{eq:linear-deleted-inverse}
\max_t
\Big\|
\big(
\sum_{j\notin r_t}
\bar Y_{j,p}\bar Y_{j,p}^\top
\big)^{-1}
\Big\|_{\mathrm{op}}
=O_p(K^{-1}).
\end{equation}

Subtracting the full and retained normal equations yields
\[
\hat\beta^{-t}-\hat\beta
=
\big(
\sum_{j\notin r_t}
\bar Y_{j,p}\bar Y_{j,p}^\top
\big)^{-1}
\sum_{j\in r_t}
\bar Y_{j,p}
(\bar Y_{j,p}^\top\hat\beta-X_j).
\]
Moreover,
\begin{align*}
\max_j|\bar Y_{j,p}^\top\hat\beta-X_j|
&\leq
\max_j\|\bar Y_{j,p}\|_2\|\hat\beta\|_2
+\max_j|X_j|=
O_p(\sqrt p\,K^{1/(4+\delta)}).
\end{align*}
Consequently,
\begin{align}
\max_t\|\hat\beta^{-t}-\hat\beta\|_2
&\leq
\max_t
\Big\|
\big(
\sum_{j\notin r_t}
\bar Y_{j,p}\bar Y_{j,p}^\top
\big)^{-1}
\Big\|_{\mathrm{op}}
\max_t|r_t|
\max_j\|\bar Y_{j,p}\|_2
\max_j|\bar Y_{j,p}^\top\hat\beta-X_j|
\notag\\
&=
O_p\big(
p\log K\,K^{-1+2/(4+\delta)}
\big).
\label{eq:linear-coefficient-stability}
\end{align}
It follows that
\begin{align*}
S_g
&=
\frac1K\sum_t
|\bar Y_{t,p}^\top(\hat\beta-\hat\beta^{-t})|^2\\
&\leq
\big(
\frac1K\sum_t\|\bar Y_{t,p}\|_2^2
\big)
\max_t\|\hat\beta-\hat\beta^{-t}\|_2^2\\
&=
O_p\big(
p^3(\log K)^2K^{-2+4/(4+\delta)}
\big)\\
&=
K^{-1}O_p\big(
(\log K)^5K^{-\delta/(4+\delta)}
\big)\\
&=o_p(K^{-1}).
\end{align*}
Replacing $X_j$ by $Y_{j+1}$ proves the result for $f$.
Applying the scalar argument to each dictionary coordinate and
summing proves the result for fixed-dimensional $g$.
\end{proof}

\begin{lemma}[Expected residual-weighted stability of linear regression]
\label{prop:linear-weighted-stability}
Suppose that $Y_j$ and each coordinate of $\phi(\bar X_j)$ have uniformly bounded sub-Gaussian norms,
and also that all doubly deleted Gram matrices are
nonsingular almost surely for sufficiently large $K$ and,
for some $\delta>0$,
\begin{equation}
\label{eq:linear-weighted-inverse}
\max_{1\leq t\leq K}\E\Big[
\max_{1\leq t'\leq K}
\Big\|
\big(
\frac1K
\sum_{j\notin r_t\cup r_{t'}}
\bar Y_{j,p}\bar Y_{j,p}^\top
\big)^{-1}
\Big\|_{\mathrm{op}}^{3+\delta}
\Big]
=O(1).
\end{equation}
Then, for $q\in\{f,g\}$,
\begin{align}
\E[\tilde S_q+\tilde D_q]
&\leq
C\frac{p^3(\log K)^4}{K^2},
\label{eq:linear-weighted-stability-rate}
\end{align}
In particular,
\[
(l_K+r_K)\E[\tilde S_q]+K\E[\tilde D_q]=o(1).
\]
\end{lemma}

\begin{proof}
We first consider the scalar source regression
$\phi(\bar X_t)=X_t$.
Uniform sub-Gaussian moment bounds imply, for every fixed
$a\geq1$,
\begin{equation}
\label{eq:linear-weighted-max-moments}
\begin{aligned}
\E[\max_j|X_j|^{2a}]
&\leq C(\log K)^a,\\
\E[\max_j\|\bar Y_{j,p}\|_2^{2a}]
&\leq Cp^a(\log K)^a,\\
\E\big[
\big(
\frac1K\sum_jX_j^2
\big)^a
\big]
&\leq C.
\end{aligned}
\end{equation}
The first two bounds follow from a union bound and the third
from Jensen's inequality.
The initial lags are included in the corresponding maximum
over the primitive $Y$ observations.

The definitions of the population residuals and conditional
Jensen's inequality give
\[
\sup_{K,t}
\E\big[
|\varepsilon_t|^{4a}
+\|\boldsymbol\xi_t\|_2^{4a}
\big]\leq C.
\]
Since
$\E\|\bar Y_{t,p}\|_2^{4a}\leq Cp^{2a}$,
Cauchy--Schwarz yields
\begin{equation}
\label{eq:linear-weighted-evaluation-moments}
\begin{aligned}
\max_t
\E\big[
\{\varepsilon_t^2\|\bar Y_{t,p}\|_2^2\}^a
\big]
&\leq Cp^a,\\
\max_t
\E\big[
\{\|\boldsymbol\xi_t\|_2^2
\|\bar Y_{t,p}\|_2^2\}^a
\big]
&\leq Cp^a,\\
\max_t\E[\|\bar Y_{t,p}\|_2^{2a}]
&\leq Cp^a.
\end{aligned}
\end{equation}

The full Gram matrix dominates every retained Gram matrix.
For each fixed $t$, its normalized inverse is therefore
bounded in operator norm by
\[
\max_{t'}
\Big\|
\big(
\frac1K
\sum_{j\notin r_t\cup r_{t'}}
\bar Y_{j,p}\bar Y_{j,p}^\top
\big)^{-1}
\Big\|_{\mathrm{op}}.
\]
The least-squares projection identity also gives
\begin{equation}
\label{eq:linear-weighted-beta-bound}
\|\hat\beta\|_2^2
\leq
\Big\|
\big(
\frac1K\sum_j
\bar Y_{j,p}\bar Y_{j,p}^\top
\big)^{-1}
\Big\|_{\mathrm{op}}
\frac1K\sum_jX_j^2.
\end{equation}
In particular,
\[
\max_j|\bar Y_{j,p}^\top\hat\beta-X_j|
\leq
\max_j\|\bar Y_{j,p}\|_2\|\hat\beta\|_2
+\max_j|X_j|.
\]

Subtracting the full and retained normal equations gives
\[
\hat\beta^{-t,t'}-\hat\beta
=
\big(
\sum_{j\notin r_t\cup r_{t'}}
\bar Y_{j,p}\bar Y_{j,p}^\top
\big)^{-1}
\sum_{j\in r_t\cup r_{t'}}
\bar Y_{j,p}
(\bar Y_{j,p}^\top\hat\beta-X_j).
\]
Taking $t'=t$ includes single deletions.
Thus, for each $t$,
\begin{align}
&\max_{t'}
\|\hat\beta^{-t,t'}-\hat\beta\|_2
\notag\\
&\quad\leq
\frac{2\max_s|r_s|}{K}
\max_{t'}
\Big\|
\big(
\frac1K
\sum_{j\notin r_t\cup r_{t'}}
\bar Y_{j,p}\bar Y_{j,p}^\top
\big)^{-1}
\Big\|_{\mathrm{op}}
\times
\max_j\|\bar Y_{j,p}\|_2
\big(
\max_j\|\bar Y_{j,p}\|_2\|\hat\beta\|_2
+\max_j|X_j|
\big).
\label{eq:linear-weighted-coefficient-bound}
\end{align}

Substitute \eqref{eq:linear-weighted-beta-bound} into
\eqref{eq:linear-weighted-coefficient-bound}.
The largest inverse power after raising to $2+\delta/3$ is
\[
\frac32(2+\delta/3)
=3+\delta/2<3+\delta.
\]
H\"older's inequality, the assumed inverse-moment bound and
\eqref{eq:linear-weighted-max-moments} therefore imply
\begin{equation}
\label{eq:linear-weighted-coefficient-moments}
\max_t
\E\big[
\max_{t'}
\|\hat\beta^{-t,t'}-\hat\beta\|_2^{2+\delta/3}
\big]
\leq
C\big\{
p(\log K)^2K^{-1}
\big\}^{2+\delta/3}.
\end{equation}
The constants are uniform in $t$.
The predictor maxima contribute at most $p\log K$, while
$\max_s|r_s|=O(\log K)$ supplies the remaining logarithmic
factor.

For the additional deletion,
\begin{align*}
\max_{t'}
\|\hat\beta^{-t}-\hat\beta^{-t,t'}\|_2
&\leq
\|\hat\beta^{-t}-\hat\beta\|_2
+\max_{t'}
\|\hat\beta^{-t,t'}-\hat\beta\|_2\\
&\leq
2\max_{t'}
\|\hat\beta^{-t,t'}-\hat\beta\|_2.
\end{align*}
All fits are evaluated at the original predictors, so
\begin{equation}
\label{eq:linear-weighted-prediction-deletion}
\tilde S_g+\tilde D_g
\leq
\frac5K\sum_t
\varepsilon_t^2\|\bar Y_{t,p}\|_2^2
\max_{t'}
\|\hat\beta^{-t,t'}-\hat\beta\|_2^2.
\end{equation}
Applying H\"older separately at each $t$, with exponents
$1+6/\delta$ and $1+\delta/6$, gives
\begin{align*}
\E[\tilde S_g+\tilde D_g]
&\leq
\frac5K\sum_t
\bigg\{
\E\Big[
\{\varepsilon_t^2\|\bar Y_{t,p}\|_2^2\}^{1+6/\delta}
\Big]
\bigg\}^{\delta/(6+\delta)}\\
&\quad{}\times
\bigg\{
\E\Big[
\max_{t'}
\|\hat\beta^{-t,t'}-\hat\beta\|_2^{2+\delta/3}
\Big]
\bigg\}^{6/(6+\delta)}\\
&\leq
Cp\big\{
p(\log K)^2K^{-1}
\big\}^2\\
&=
C\frac{p^3(\log K)^4}{K^2},
\end{align*}
using
\eqref{eq:linear-weighted-evaluation-moments} and
\eqref{eq:linear-weighted-coefficient-moments}.
For the outcome regression, replace the training response
$X_j$ by $Y_{j+1}$ and the evaluation multiplier by
$\|\boldsymbol\xi_t\|_2^2$.
The corresponding bounds in
\eqref{eq:linear-weighted-evaluation-moments} prove both
conclusions for $f$.
For fixed-dimensional $g$, apply the scalar argument
coordinatewise and combine the finitely many first- and
second-moment bounds.

Finally, $p=O(\log K)$ and $l_K+r_K=O(\log K)$ imply
\[
(l_K+r_K)\E[\tilde S_q]+K\E[\tilde D_q]
\leq
C\frac{p^3(\log K)^4}{K}
\big\{1+(l_K+r_K)/K\big\}
=o(1),
\]
\end{proof}

\subsection{Time-varying coefficient kernel regression}
\label{app:tvcm-weighted-stability}

Consider
\[
X_t=g_t(\bar Y_{t,p})+\xi_t
   =\bar Y_{t,p}^\top\boldsymbol\beta(t)+\xi_t,
\quad
\E[\xi_t\given H_t^{Y,Z}]=0,
\]
where the deterministic coefficients may vary with $t$ and
$K$.
Let
\[
W_h(j-t)=\mathcal K((j-t)/(Kh)),
\]
where $\mathcal K$ is nonnegative, bounded and supported on
$[-1,1]$.
The locally weighted least-squares estimator is
\begin{equation}
\label{eq:tvcm-full-estimator}
\hat{\boldsymbol\beta}(t)
=
\big(
\sum_jW_h(j-t)
\bar Y_{j,p}\bar Y_{j,p}^\top
\big)^{-1}
\sum_jW_h(j-t)\bar Y_{j,p}X_j,
\end{equation}
on nonsingular samples, and
\[
\hat g_t(\bar Y_{t,p})
=
\bar Y_{t,p}^\top\hat{\boldsymbol\beta}(t).
\]
Deleted fits use the same kernel and bandwidth and omit
$r_t$ or $r_t\cup r_{t'}$.

\begin{lemma}[Probability stability of kernel regression]
\label{lem:tvcm-kernel-stability}
Suppose that, for some $\delta>0$,
\begin{equation}
\label{eq:tvcm-moment}
\sup_{K,j}
\E\big[
\|\phi(\bar X_j)\|_2^{4+\delta}
+|Y_j|^{4+\delta}
\big]<\infty.
\end{equation}
Suppose $h=h_K$ and
\begin{equation}
\label{eq:tvcm-bandwidth}
(\log K)^{5/2}
K^{-\delta/\{2(4+\delta)\}}
=o(h).
\end{equation}
For some $c>0$, suppose also that
\begin{equation}
\label{eq:tvcm-local-conditions}
\begin{aligned}
&\Pr\Big[
\min_t\lambda_{\min}\big(
\frac1{Kh}\sum_jW_h(j-t)
\bar Y_{j,p}\bar Y_{j,p}^\top
\big)\geq c
\Big]\rightarrow1.
\end{aligned}
\end{equation}
Then
\[
S_q
=
O_p\big(
p^3(\log K)^2h^{-2}
K^{-2+4/(4+\delta)}
\big)
=o_p(K^{-1}),
\quad q\in\{f,g\}.
\]
\end{lemma}

\begin{proof}
We first consider the scalar source regression
$\phi(\bar X_t)=X_t$.
The uniform moment bounds, Markov's inequality and the union
bound give
\[
\max_j|X_j|=O_p(K^{1/(4+\delta)}),
\quad
\max_j\|\bar Y_{j,p}\|_2
=O_p(\sqrt p\,K^{1/(4+\delta)}).
\]
Also,
\[
\E\big[
\frac1K\sum_t\|\bar Y_{t,p}\|_2^2
\big]\leq Cp,
\quad
\frac1K\sum_t\|\bar Y_{t,p}\|_2^2=O_p(p).
\]

Condition~\eqref{eq:tvcm-bandwidth} implies $Kh\to\infty$.
Since the kernel is bounded and supported on $[-1,1]$,
\[
\sup_t\frac1{Kh}\sum_jW_h(j-t)\leq C.
\]
The uniform fourth moments imply
$\E\|\bar Y_{t,p}\|_2^4\leq Cp^2$.
Consequently, Cauchy--Schwarz gives
\begin{align*}
&\E\bigg[
\frac1K\sum_t\|\bar Y_{t,p}\|_2^2
\big(
\frac1{Kh}\sum_jW_h(j-t)X_j^2
\big)
\bigg]\\
&\quad\leq
\frac1K\sum_t\frac1{Kh}\sum_jW_h(j-t)
\{\E\|\bar Y_{t,p}\|_2^4\}^{1/2}
\{\E|X_j|^4\}^{1/2}\\
&\quad\leq Cp.
\end{align*}
Markov's inequality therefore yields
\begin{equation}
\label{eq:tvcm-average-response-bound}
\frac1K\sum_t\|\bar Y_{t,p}\|_2^2
\big(
\frac1{Kh}\sum_jW_h(j-t)X_j^2
\big)
=O_p(p).
\end{equation}

For a single deletion,
\begin{align*}
&\max_t
\Big\|
\sum_{j\in r_t}W_h(j-t)
\bar Y_{j,p}\bar Y_{j,p}^\top
\Big\|_{\mathrm{op}}\\
&\quad\leq
\|\mathcal K\|_\infty
\max_t|r_t|\max_j\|\bar Y_{j,p}\|_2^2\\
&\quad=
O_p\big(p\log K\,K^{2/(4+\delta)}\big)\\
&\quad=o_p(Kh),
\end{align*}
where the last line follows from
\eqref{eq:tvcm-bandwidth} and $p=O(\log K)$.
The uniform lower eigenvalue condition and Weyl's inequality
thus imply that, with probability tending to one,
\[
\min_t
\lambda_{\min}\big(
\sum_{j\notin r_t}W_h(j-t)
\bar Y_{j,p}\bar Y_{j,p}^\top
\big)
\geq cKh/2.
\]
On the same event,
\begin{equation}
\label{eq:tvcm-deleted-inverse}
\max_t
\Big\|
\big(
\sum_{j\notin r_t}W_h(j-t)
\bar Y_{j,p}\bar Y_{j,p}^\top
\big)^{-1}
\Big\|_{\mathrm{op}}
\leq\frac{2}{cKh}.
\end{equation}

The weighted least-squares projection identity gives
\[
\sum_jW_h(j-t)
\{\bar Y_{j,p}^\top\hat{\boldsymbol\beta}(t)\}^2
\leq
\sum_jW_h(j-t)X_j^2.
\]
Hence, on the full Gram event, for every $t$,
\[
\|\hat{\boldsymbol\beta}(t)\|_2^2
\leq
\frac1c\frac1{Kh}\sum_jW_h(j-t)X_j^2.
\]

Subtracting the full and retained weighted normal equations
gives
\begin{align*}
&\hat{\boldsymbol\beta}^{-t}(t)
-\hat{\boldsymbol\beta}(t)\\
&\quad=
\big(
\sum_{j\notin r_t}W_h(j-t)
\bar Y_{j,p}\bar Y_{j,p}^\top
\big)^{-1}\times
\sum_{j\in r_t}W_h(j-t)\bar Y_{j,p}
\{\bar Y_{j,p}^\top\hat{\boldsymbol\beta}(t)-X_j\}.
\end{align*}
Using \eqref{eq:tvcm-deleted-inverse}, boundedness of the
kernel and $\max_t|r_t|=O(\log K)$, we obtain, simultaneously
for every $t$ on the preceding event,
\begin{align*}
&\|\hat{\boldsymbol\beta}^{-t}(t)
-\hat{\boldsymbol\beta}(t)\|_2\leq
\frac{C\log K}{Kh}
\max_j\|\bar Y_{j,p}\|_2
\big(
\max_j\|\bar Y_{j,p}\|_2
\|\hat{\boldsymbol\beta}(t)\|_2
+\max_j|X_j|
\big).
\end{align*}
Therefore,
\begin{align*}
S_g
&=
\frac1K\sum_t
\big|
\bar Y_{t,p}^\top
\{\hat{\boldsymbol\beta}(t)
-\hat{\boldsymbol\beta}^{-t}(t)\}
\big|^2\\
&\leq
\frac{C(\log K)^2}{K^2h^2}
\bigg[
\max_j\|\bar Y_{j,p}\|_2^4
\frac1K\sum_t\|\bar Y_{t,p}\|_2^2
\big(
\frac1{Kh}\sum_jW_h(j-t)X_j^2
\big)\\
&\hspace{45mm}{}
+\max_j\|\bar Y_{j,p}\|_2^2
\max_j|X_j|^2
\frac1K\sum_t\|\bar Y_{t,p}\|_2^2
\bigg].
\end{align*}
The maximum bounds and
\eqref{eq:tvcm-average-response-bound} now imply
\begin{align*}
S_g
&=
O_p\big(
p^3(\log K)^2h^{-2}
K^{-2+4/(4+\delta)}
\big)\\
&=
K^{-1}O_p\bigg[
\big\{
p^{3/2}\log K\,
K^{-\delta/\{2(4+\delta)\}}/h
\big\}^2
\bigg]\\
&=o_p(K^{-1}),
\end{align*}
because $p=O(\log K)$ and
\eqref{eq:tvcm-bandwidth} holds.

Replacing $X_j$ by $Y_{j+1}$ proves the outcome-regression
result under the same moment assumptions.
For fixed-dimensional $g$, apply the scalar argument to each
coordinate of $\phi(\bar X_j)$ and sum the finitely many
coordinatewise bounds.
\end{proof}

\begin{lemma}[Expected residual-weighted stability of kernel regression]
\label{lem:tvcm-weighted-stability}
Suppose that $Y_j$ and each coordinate of $\phi(\bar X_j)$ have uniformly bounded sub-Gaussian norms, and $h=h_K$ such that
\begin{equation}
\label{eq:tvcm-weighted-bandwidth}
\frac{p^{3/2}(\log K)^2}{\sqrt K}=o(h).
\end{equation}
Suppose also that all doubly deleted local Gram matrices are
nonsingular almost surely for sufficiently large $K$ and,
for some $\delta>0$,
\begin{equation}
\label{eq:tvcm-weighted-inverse}
\max_{1\leq t\leq K}
\E\bigg[
\max_{1\leq t'\leq K}
\Big\|
\big(
\frac1{Kh}
\sum_{j\notin r_t\cup r_{t'}}
W_h(j-t)\bar Y_{j,p}\bar Y_{j,p}^\top
\big)^{-1}
\Big\|_{\mathrm{op}}^{3+\delta}
\bigg]
=O(1).
\end{equation}
Then, for $q\in\{f,g\}$,
\begin{align}
\E[\tilde S_q+\tilde D_q]
&\leq
C\frac{p^3(\log K)^4}{K^2h^2},
\label{eq:tvcm-weighted-stability-rate}
\end{align}
In particular,
\[
(l_K+r_K)\E[\tilde S_q]+K\E[\tilde D_q]=o(1).
\]
\end{lemma}

\begin{proof}
We first consider $\phi(\bar X_t)=X_t$.
The sub-Gaussian maximum bounds in
\eqref{eq:linear-weighted-max-moments} remain valid.
Since $Kh\to\infty$ and the kernel is bounded and supported
on $[-1,1]$,
\[
\sup_t\frac1{Kh}\sum_jW_h(j-t)\leq C.
\]
Weighted Jensen's inequality therefore gives, for every fixed
$a\geq1$,
\begin{equation}
\label{eq:tvcm-weighted-local-response}
\max_t
\E\big[
\big(
\frac1{Kh}\sum_jW_h(j-t)X_j^2
\big)^a
\big]
\leq C.
\end{equation}
The same bound holds with $X_j$ replaced by $Y_{j+1}$.

As in the linear-regression proof, the population residuals
have uniformly bounded moments of every fixed order.
Thus
\begin{equation}
\label{eq:tvcm-weighted-evaluation-moments}
\begin{aligned}
\max_t
\E\big[
\{\varepsilon_t^2\|\bar Y_{t,p}\|_2^2\}^a
\big]
&\leq Cp^a,\\
\max_t
\E\big[
\{\|\boldsymbol\xi_t\|_2^2
\|\bar Y_{t,p}\|_2^2\}^a
\big]
&\leq Cp^a,\\
\max_t\E[\|\bar Y_{t,p}\|_2^{2a}]
&\leq Cp^a.
\end{aligned}
\end{equation}

The weighted projection identity gives
\begin{equation}
\label{eq:tvcm-weighted-beta-bound}
\|\hat{\boldsymbol\beta}(t)\|_2^2
\leq
\Big\|
\big(
\frac1{Kh}\sum_jW_h(j-t)
\bar Y_{j,p}\bar Y_{j,p}^\top
\big)^{-1}
\Big\|_{\mathrm{op}}
\frac1{Kh}\sum_jW_h(j-t)X_j^2.
\end{equation}
At each $t$, the full inverse is bounded by the maximum
retained inverse in \eqref{eq:tvcm-weighted-inverse}.

Subtracting the full and retained weighted normal equations
gives
\begin{align*}
&\hat{\boldsymbol\beta}^{-t,t'}(t)
-\hat{\boldsymbol\beta}(t)\\
&\quad=
\big(
\sum_{j\notin r_t\cup r_{t'}}
W_h(j-t)\bar Y_{j,p}\bar Y_{j,p}^\top
\big)^{-1}\times
\sum_{j\in r_t\cup r_{t'}}
W_h(j-t)\bar Y_{j,p}
\{\bar Y_{j,p}^\top\hat{\boldsymbol\beta}(t)-X_j\}.
\end{align*}
Hence
\begin{align}
&\max_{t'}
\|\hat{\boldsymbol\beta}^{-t,t'}(t)
-\hat{\boldsymbol\beta}(t)\|_2
\notag\\
&\quad\leq
\frac{2\|\mathcal K\|_\infty\max_s|r_s|}{Kh}
\max_{t'}
\Big\|
\big(
\frac1{Kh}
\sum_{j\notin r_t\cup r_{t'}}
W_h(j-t)\bar Y_{j,p}\bar Y_{j,p}^\top
\big)^{-1}
\Big\|_{\mathrm{op}}
\notag\\
&\quad\quad\quad\quad\quad\times
\max_j\|\bar Y_{j,p}\|_2
\big(
\max_j\|\bar Y_{j,p}\|_2
\|\hat{\boldsymbol\beta}(t)\|_2
+\max_j|X_j|
\big).
\label{eq:tvcm-weighted-coefficient-bound}
\end{align}
Applying H\"older as in the linear-regression proof yields
\begin{equation}
\label{eq:tvcm-weighted-coefficient-moments}
\max_t
\E\big[
\max_{t'}
\|\hat{\boldsymbol\beta}^{-t,t'}(t)
-\hat{\boldsymbol\beta}(t)\|_2^{2+\delta/3}
\big]
\leq
C\big\{
p(\log K)^2(Kh)^{-1}
\big\}^{2+\delta/3}.
\end{equation}
The largest inverse power is again
$3+\delta/2<3+\delta$.
The local response moments are bounded uniformly in $t$ by
\eqref{eq:tvcm-weighted-local-response}; no maximum over $t$
inside that expectation is required.

The triangle inequality and
$\hat{\boldsymbol\beta}^{-t,t}(t)
=\hat{\boldsymbol\beta}^{-t}(t)$ give
\begin{align*}
\tilde S_g+\tilde D_g
&\leq
\frac5K\sum_t
\varepsilon_t^2\|\bar Y_{t,p}\|_2^2
\max_{t'}
\|\hat{\boldsymbol\beta}^{-t,t'}(t)
-\hat{\boldsymbol\beta}(t)\|_2^2.
\end{align*}
Apply H\"older to each summand, with exponents
$1+6/\delta$ and $1+\delta/6$.
Equations
\eqref{eq:tvcm-weighted-evaluation-moments} and
\eqref{eq:tvcm-weighted-coefficient-moments} imply
\[
\E[\tilde S_g+\tilde D_g]
\leq
Cp\big\{
p(\log K)^2(Kh)^{-1}
\big\}^2
=
C\frac{p^3(\log K)^4}{K^2h^2}.
\]



The outcome and vector-source cases follow by the same
substitutions as in
Lemma~\ref{prop:linear-weighted-stability}.
Finally,
\begin{align*}
(l_K+r_K)\E[\tilde S_q]+K\E[\tilde D_q]
&\leq
C\frac{p^3(\log K)^4}{Kh^2}
\big\{1+(l_K+r_K)/K\big\}\\
&=o(1)
\end{align*}
by \eqref{eq:tvcm-weighted-bandwidth}.
The same bandwidth condition and $p=O(\log K)$ imply
$p^3(\log K)^2/(Kh)\to0$.
\end{proof}

\section{Asymptotic Power}
\label{app:power-results}

\subsection{Complementary Results}

\begin{proposition}
\label{prop:Sigma_xi}
Under Assumption~\ref{ass:errors}(i) and
Assumption~\ref{ass:population_weights}, suppose that, for a finite matrix
\(\Gamma\in\mathbb R^{d\times d}\),
\begin{equation}
\label{eq:local-drift-mean-limit}
\frac1K\sum_{t=1}^K\E[W_t\boldsymbol\xi_t\boldsymbol\xi_t^\top]
\rightarrow\Gamma.
\end{equation}
Then \(K^{-1}\sum_{t=1}^K
W_t\boldsymbol\xi_t\boldsymbol\xi_t^\top\xrightarrow{p}\Gamma\).
\end{proposition}

\begin{proof}
The second-moment and covariance bounds give
\begin{align*}
&\E\Big\|\frac1K\sum_{t=1}^K
\big\{W_t\boldsymbol\xi_t\boldsymbol\xi_t^\top
-\E[W_t\boldsymbol\xi_t\boldsymbol\xi_t^\top]\big\}\Big\|_F^2\\
&\quad\leq\frac1K\max_{1\leq t\leq K}
\E[W_t^2\|\boldsymbol\xi_t\|_2^4]\\
&\qquad+\frac{2d}{K^2}\sum_{h=1}^{K-1}(K-h)
\max_{1\leq t\leq K-h}\Big\|\Cov\big(
W_t\operatorname{vec}(\boldsymbol\xi_t\boldsymbol\xi_t^\top),
W_{t+h}\operatorname{vec}(\boldsymbol\xi_{t+h}\boldsymbol\xi_{t+h}^\top)
\big)\Big\|_F\\
&\quad \rightarrow0,
\end{align*}
by Assumption~\ref{ass:population_weights}(ii).
Chebyshev's inequality and \eqref{eq:local-drift-mean-limit} complete the proof.
\end{proof}

\subsection{Local optimality of precision weighting}

Under the local alternative \eqref{eq:local_alt}, Theorem~\ref{thm:local_power}
gives the noncentrality parameter
\begin{equation}
    {\tau^\top\Gamma^\top\Sigma^{-1}\Gamma\tau.}
    \label{eq:weighted-local-noncentrality}
\end{equation}
We compare this quantity across different choices of the population weight. The key additional condition is that the centred innovation
variance, conditional on the full history, is determined by \(H_t^{Y,Z}\).
This allows the same scalar predictable weight to be optimal for every local
direction.

\begin{proposition}[Optimality of precision weighting]
\label{prop:vector-precision-optimal}
Suppose the conditions of Theorem~\ref{thm:local_power} hold with \(W_t=\sigma_t^{-2}\), where \(W_t\) and \(\hat W_t\) satisfy the assumptions in Appendix~\ref{app:weight-assumptions}. Suppose in addition that
\begin{equation}
    \E[
        \zeta_t^2
        \mid H_t^{X,Y,Z}]
    =
    \sigma_t^2,
    \label{eq:variance-history-only}
\end{equation}
where $\sigma_t^2$ is $H_t^{Y,Z}$-measurable and bounded away from zero and
infinity. Suppose also that
\begin{equation}
    \Gamma^\star
    =
    \lim_{K\to\infty}
    \frac{1}{K}\sum_{t=1}^K
    \E\big[
        \sigma_t^{-2}
        \boldsymbol{\xi}_t\boldsymbol{\xi}_t^\top
    \big]
    \label{eq:Gamma-star}
\end{equation}
is positive definite. For any bounded positive $H_t^{Y,Z}$-measurable weight
$\tilde W_t$ satisfying the conditions of Theorem~\ref{thm:local_power}, let
\begin{equation}
    \tilde\Gamma
    =
    \lim_{K\to\infty}
    \frac{1}{K}\sum_{t=1}^K
    \E\big[
        \tilde W_t
        \boldsymbol{\xi}_t\boldsymbol{\xi}_t^\top
    \big],
    ~
    \tilde\Sigma
    =
    \lim_{K\to\infty}
    \frac{1}{K}\sum_{t=1}^K
    \E\big[
        \tilde W_t^2\zeta_t^2
        \boldsymbol{\xi}_t\boldsymbol{\xi}_t^\top
    \big].
    \label{eq:alternative-weight-limits}
\end{equation}
Then
\begin{equation}
    {
    \tilde\Gamma^\top
    \tilde\Sigma^{-1}
    \tilde\Gamma
    \preceq
    \Gamma^\star.}
    \label{eq:precision-loewner-bound}
\end{equation}
Equality is attained by any weight of the form
\begin{equation}
    W_t^\star=c\sigma_t^{-2},
    ~ c>0.
    \label{eq:optimal-precision-weight}
\end{equation}
Consequently, precision weighting maximizes the local asymptotic power for
every $\tau\in\mathbb R^d$.
\end{proposition}

\begin{proof}
By \eqref{eq:variance-history-only} and iterated expectation,
\begin{equation}
    \tilde\Sigma
    =
    \lim_{K\to\infty}
    \frac{1}{K}\sum_{t=1}^K
    \E\big[
        \tilde W_t^2\sigma_t^2
        \boldsymbol{\xi}_t\boldsymbol{\xi}_t^\top
    \big].
\end{equation}
Hence, for any $\boldsymbol{a},\boldsymbol{b}\in\mathbb R^d$,
\begin{equation}
\begin{aligned}
    0
    &\leq
    \lim_{K\to\infty}
    \frac{1}{K}\sum_{t=1}^K
    \E\big[
        \big(
            \boldsymbol{a}^\top
            \sigma_t\tilde W_t\boldsymbol{\xi}_t
            +
            \boldsymbol{b}^\top
            \sigma_t^{-1}\boldsymbol{\xi}_t
        \big)^2
    \big] \\
    &=
    \boldsymbol{a}^\top\tilde\Sigma\boldsymbol{a}
    +
    2\boldsymbol{a}^\top\tilde\Gamma\boldsymbol{b}
    +
    \boldsymbol{b}^\top\Gamma^\star\boldsymbol{b}.
\end{aligned}
\end{equation}
Taking $    \boldsymbol{a}
    =
    -\tilde\Sigma^{-1}
    \tilde\Gamma\boldsymbol{b}$ gives
\begin{equation}
    \boldsymbol{b}^\top
    \big(
        \Gamma^\star
        -
        {\tilde\Gamma^\top}
        \tilde\Sigma^{-1}
        \tilde\Gamma
    \big)
    \boldsymbol{b}
    \geq 0
\end{equation}
for every $\boldsymbol{b}\in\mathbb R^d$, proving
\eqref{eq:precision-loewner-bound}. For $W_t^\star=c\sigma_t^{-2}$, the
corresponding limits are $c\Gamma^\star$ and $c^2\Gamma^\star$, respectively.
Therefore,
\[
    {
    (c\Gamma^\star)^\top
    (c^2\Gamma^\star)^{-1}
    (c\Gamma^\star)
    =
    \Gamma^\star,}
\]
so equality holds in \eqref{eq:precision-loewner-bound}.
\end{proof}

Finally, the population precision weight used by the procedure agrees with the oracle weight in \eqref{eq:optimal-precision-weight} with \(c=1\), in empirical mean square.
Under \eqref{eq:local_alt} and \eqref{eq:variance-history-only},
\begin{equation}
\begin{aligned}
    \E\big[
        \varepsilon_t^2
        \mid H_t^{Y,Z}
    \big]
    &=
    \sigma_t^2
    +
    \frac{1}{K}
    \E\big[
        (
            \boldsymbol{\xi}_t^\top\tau
        )^2
        \mid H_t^{Y,Z}
    \big].
\end{aligned}
\end{equation}
Since \(\sigma_t^2\geq c_2^2\), Jensen's inequality gives
\begin{equation}
\begin{aligned}
&\frac1K\sum_{t=1}^K
\Big(\big\{\E[\varepsilon_t^2\given H_t^{Y,Z}]\big\}^{-1}
-\sigma_t^{-2}\Big)^2\\
&\quad\leq\frac{c_2^{-8}\|\tau\|_2^4}{K^3}
\sum_{t=1}^K\E[\|\boldsymbol\xi_t\|_2^4\given H_t^{Y,Z}]
=O_p(K^{-2}),
\end{aligned}
\end{equation}
by Assumption~\ref{ass:mixing}(iv). Together with
\eqref{eq:variance-history-only}, this gives the same local drift and
limiting covariance as \(\sigma_t^{-2}\). Thus, provided it satisfies
Assumption~\ref{ass:population_weights}, the population precision weight
targets the locally optimal weight to first order.

\section{Additional simulation results}
\label{app:implementation}

\subsection{Simulation design and implementation}
\label{sec:simulation-common-details}

All three designs use \(K=2500\), 1000 Monte Carlo replications, and nominal level 0.05. We vary the ratio of the upper to lower conditional variance defined below:
\begin{equation}
\label{eq:simulation-variance-function}
\sigma_t^2=\sigma_L^2+(\sigma_H^2-\sigma_L^2)
\{1+\exp(-4Y_{t-2})\}^{-1},
\quad
\sigma_L^2=\frac{0.40}{1+R_\sigma},
\quad
\sigma_H^2=R_\sigma\sigma_L^2.
\end{equation}
Thus \(R_\sigma\in\{4,8,16\}\) controls the degree of heteroscedasticity.

Conditional means and precision weights are estimated by
generalised random forests \citep{athey2019generalized}, using 300 trees for
the mean regressions and 500 trees for the precision regression. Predictions are obtained from the complete fitted forests. The estimated precisions are truncated below at \(10^{-4}\) and above at their empirical 99th percentile. The adaptive procedure searches 31 leading rectangles of the eight-lag by five-degree dictionary, each containing at most 20 coordinates, and is calibrated with 1999 Gaussian draws.

All comparison methods use the same eight exposure lags and the
outcome history appropriate to the design. We consider the classical
Granger \(F\)-test \citep{granger1969investigating}, a linear VAR test
\citep{lutkepohl2005new}, an HC3-robust dynamic linear Wald test
\citep{mackinnon1985some}, and the weighted GCM with GAM nuisance regressions
\citep{scheidegger2022weighted}. Their specifications were fixed before the
Monte Carlo results were examined.

\subsection{Additional simulation designs}
\label{sec:simulation-additional-results}

The primary nonlinear shared-history design is reported in the
main paper. Here we describe the two additional designs and report their
complete heteroscedasticity grids.

\medskip
\noindent{\it \textbf{Nonlinear distributed-signal design}.}
\label{sec:simulation-nonlinear-lags}

The second design has a nonlinear exposure effect distributed
over lags 1--3 and a linear, smoothly time-varying conditioning history. The
data were generated from
\begin{align}
X_t={}&0.20X_{t-1}-0.05X_{t-2}
+\{0.55+0.10\cos(2\pi t/1000)\}Y_{t-1}+u_t,
\label{eq:sim-nonlinear-x-supp}\\
Y_t={}&\{0.50+0.25\sin(2\pi t/1000)\}Y_{t-1}-0.10Y_{t-2}
+\beta \sum_{j=1}^3\frac{X_{t-j}^2}{1+X_{t-j}^2}+\varepsilon_t.
\label{eq:sim-nonlinear-y-supp}
\end{align}
Here \(u_t\sim N(0,0.20)\),
\(\varepsilon_t\sim N(0,\sigma_t^2)\), \(\beta=\delta/\sqrt K\), and the
variance follows \eqref{eq:simulation-variance-function}. The specified-alternative benchmark uses the three transformed coordinates $X_{t-1}, X_{t-2}$ and $X_{t-3}$. By contrast, the 40D procedures use the polynomial dictionary, which can approximate but does not contain the transformed coordinates exactly.

The procedures remain well calibrated under the null. Under the
alternatives, adaptive aggregation recovers much of the power lost by using
the full dictionary, and precision weighting becomes increasingly helpful
as heteroscedasticity grows. The specified-alternative benchmark is the most
powerful, as expected, but requires knowledge of the form and location of the
signal.

\medskip
\noindent{\it \textbf{Delayed linear-signal design}.}
\label{sec:simulation-delayed-lags}

The third design has a linear exposure effect at lags 4--5 and
a longer, smoothly time-varying conditioning history. We used the same
equation for \(X_t\) as in the nonlinear distributed-signal design and
generated the outcome according to
\begin{align}
Y_t={}&\{0.40+0.15\sin(2\pi t/1000)\}Y_{t-1}
+0.15Y_{t-2}-0.10Y_{t-3}+0.05Y_{t-4}+\beta 
\Big(\frac{X_{t-4}+X_{t-5}}{\sqrt 2}\Big)
+\varepsilon_t.
\label{eq:sim-delayed-y-supp}
\end{align}
The remaining construction is unchanged. This setting challenges the
adaptive procedure because the smallest leading-lag rectangle containing the
signal must also include the uninformative first three exposure lags. The
specified-alternative benchmark uses only \(X_{t-4}\) and \(X_{t-5}\), while
the 40D procedures use the full lag--degree dictionary.

The proposed procedures again maintain the nominal level.
Adaptation improves on the full 40D test even though the signal begins only
at lag 4, and precision weighting provides a further gain as
heteroscedasticity increases. The estimated precision weights remain close
to their oracle counterparts in both additional designs.

Figure~\ref{fig:simulation-additional-dgps-full} reports the
complete factorial results for the two additional designs. The null rows
confirm that the power gains are not obtained at the cost of size control;
the alternative rows separate the contributions of precision weighting and
adaptive aggregation.

\begin{figure}[htbp]
\centering
\includegraphics[width=0.99\linewidth]
{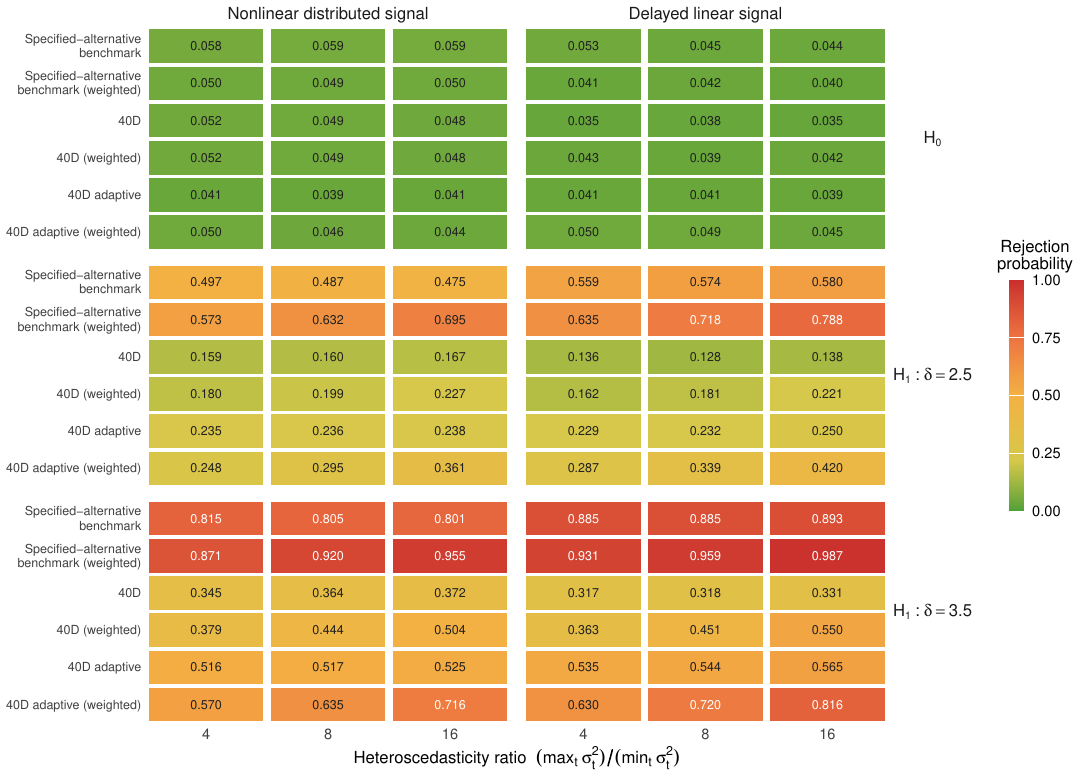}
\caption{Empirical rejection rates for the nonlinear
distributed-signal and delayed linear-signal designs, based on 1000 Monte
Carlo replications. Rows of panels correspond to the null, \(\delta=2.5\),
and \(\delta=3.5\); columns of panels correspond to the two designs. Within
each panel, columns give \(R_\sigma\in\{4,8,16\}\). The
specified-alternative benchmark uses the three active nonlinear coordinates
in the nonlinear distributed-signal design and the two active linear
coordinates in the delayed linear-signal design. ``Weighted'' denotes
estimated precision weighting.}
\label{fig:simulation-additional-dgps-full}
\end{figure}

\section{Additional analyses of the OhioT1DM cohorts}
\label{sec:ohiot1dm_supplement}

\subsection{Cohorts and analysis specification}
\label{sec:ohiot1dm_cohort_details}

Within each participant, the training and test files were joined
chronologically and duplicate boundary events were removed. Measurements
were aligned to a five-minute grid, on which all lags were constructed before
incomplete analysis rows were excluded. The 2018 source was Basis step count;
the 2020 source was five-minute mean Empatica acceleration magnitude. The response was the next five-minute glucose change. The 51-dimensional conditioning vector comprised current glucose; lags \(0,\ldots,7\) of glucose change, basal insulin, bolus insulin, carbohydrate intake and sleep fraction; total insulin and carbohydrate intake over each of the lag ranges \(8{:}11\), \(12{:}23\), \(24{:}47\) and \(48{:}71\); and 24-hour sine and cosine terms. Thus its dimension was \(1+5\times8+2\times4+2=51\). At the current time and seven preceding lags, the source was represented by standardized polynomial terms of degrees one through five, giving 40 coordinates.

Each trajectory was analysed separately. Conditional means and the
conditional response variance were estimated by 300-tree generalized random
forests. Precision weights were truncated and normalized within participant,
and the adaptive statistic was calibrated over the 40 leading lag--degree
subsets using 9,999 Gaussian draws. No nuisance fits, residuals, weights or
covariance estimates were shared between participants.

Table~\ref{tab:ohiot1dm_primary_participant} reports the participant-level
results for the adaptive precision-weighted procedure used for cohort
inference. All six 2018 values attain the Monte Carlo resolution limit. The
2020 evidence is more heterogeneous, plausibly reflecting its different
activity measurement, lower data availability or genuine between-person
variation.

\begin{table}[t]
\centering
\caption{Participant-level adaptive precision-weighted GTCM results}
\label{tab:ohiot1dm_primary_participant}
\begin{threeparttable}
\begin{tabular}{llrr}
\toprule
Cohort & Participant & \(K\) & \(p\)-value \\
\midrule
2018 Basis & 559 & 11{,}484 & 0.0001 \\
           & 563 & 12{,}484 & 0.0001 \\
           & 570 & 12{,}694 & 0.0001 \\
           & 575 & 12{,}267 & 0.0001 \\
           & 588 & 13{,}888 & 0.0001 \\
           & 591 & 11{,}804 & 0.0001 \\
\addlinespace
2020 Empatica & 540 & 5{,}329  & 0.0007 \\
              & 544 & 9{,}317  & 0.3232 \\
              & 552 & 4{,}799  & 0.0696 \\
              & 567 & 5{,}483  & 0.3257 \\
              & 584 & 10{,}194 & 0.0530 \\
              & 596 & 3{,}231  & 0.0011 \\
\bottomrule
\end{tabular}

\begin{tablenotes}[flushleft]
\footnotesize
\item \textit{Note:}
Entries are participant-specific adaptive precision-weighted \(p\)-values.
With 9,999 calibration draws, the minimum attainable \(p\)-value is 0.0001.
\end{tablenotes}
\end{threeparttable}
\end{table}

\subsection{Information accumulation and stability across random windows}
\label{sec:ohiot1dm_random_windows}

The full-record rejection does not reveal whether evidence accumulates across
the follow-up or is driven by a particular part of the records. To distinguish
these possibilities, we sampled 20 contiguous windows containing a fraction
\(\rho\in\{0.2,0.4,0.6,0.8\}\) of each 2018 participant's eligible
observations; the full record, \(\rho=1\), was analysed once. For draw \(b\),
let \(p_{i\rho b}\) be the adaptive precision-weighted \(p\)-value obtained
from participant \(i\)'s sampled window. The cohort evidence was
\begin{equation}
 T_{\rho b}=-2\sum_{i=1}^{6}\log p_{i\rho b},\quad
 p^{\mathrm{global}}_{\rho b}
 =\Pr\{\chi^2_{12}\geq T_{\rho b}\},\quad
 E_{\rho b}=-\log_{10}p^{\mathrm{global}}_{\rho b}.
 \label{eq:ohiot1dm_window_evidence}
\end{equation}
Starting positions were drawn independently across participants and fractions,
and all nuisance regressions, feature standardizations and precision weights
were recomputed within each window. Thus this analysis complements the
full-record test by assessing temporal localization and the accumulation of
evidence. Because sampled windows overlap, their distribution measures
sensitivity to window location, not repeated-sampling uncertainty.

\begin{figure}[t]
\centering
\includegraphics[width=0.80\linewidth]{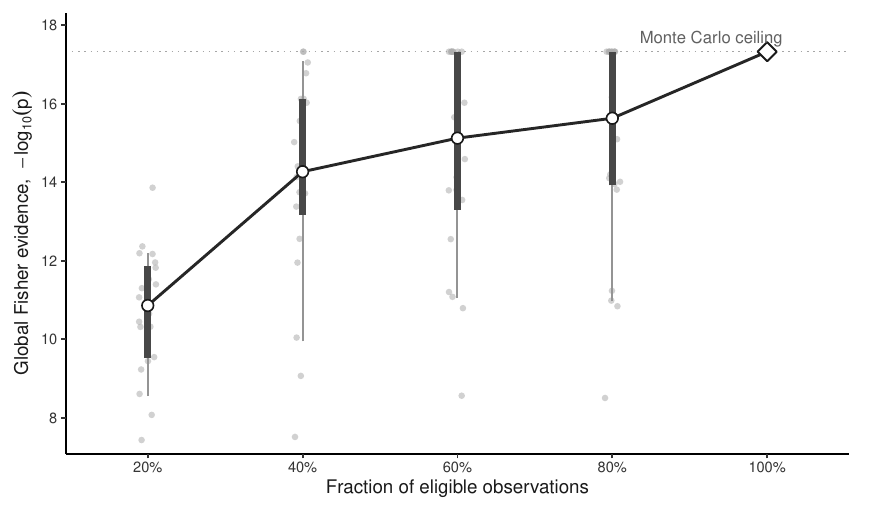}
\caption{Information accumulation for the adaptive precision-weighted GTCM
in the 2018 OhioT1DM cohort. Each pale point is the global Fisher evidence
from one draw of six participant-specific windows. At 20--80\% retention,
open circles show medians over 20 draws, thick bars show interquartile
ranges, and thin bars show 10th--90th percentile ranges. The diamond is the
single full-record result. The dotted line is the evidence ceiling when all
six participant \(p\)-values attain the Monte Carlo minimum \(10^{-4}\).
The 5\% rejection threshold, \(-\log_{10}(0.05)=1.30\), lies below the
plotted range. Bars describe variation across overlapping windows and are
not confidence intervals.}
\label{fig:ohiot1dm_random_window_stability}
\end{figure}

Figure~\ref{fig:ohiot1dm_random_window_stability} shows that median evidence
increases from 10.86 at 20\% retention to 15.63 at 80\% retention. Every
partial-record combination rejects the cohort global null, although the
spread across window locations remains appreciable. The full-record evidence
is 17.33 and reaches the ceiling induced by all six participant \(p\)-values
attaining the Monte Carlo minimum. Thus the cohort conclusion is not confined
to one portion of the records, while the increasing median is consistent
with information accumulating as longer stretches are retained.

\subsection{Independent pseudo-exposure calibration}
\label{sec:ohiot1dm_pseudo_exposure}

The very small observed cohort \(p\)-values do not by themselves establish
that the complete empirical procedure is well calibrated. We therefore used
an independent pseudo-exposure experiment as an application-specific negative
control. For repetition \(b\), observed activity values for participant \(i\)
were replaced according to
\begin{equation}
X^{(b)}_{it}\mathrel{\mathop{\sim}^{\mathrm{ind}}}\widehat F_i,
 \quad t\in\mathcal O_i,
 \label{eq:ohiot1dm_pseudo_exposure}
\end{equation}
where \(\widehat F_i\) is the participant's empirical activity distribution
and \(\mathcal O_i\) is the set of observed activity times. This preserves the
participant-specific marginal distribution and activity-missingness pattern,
including zero inflation where present, while removing temporal association
with glucose and the recorded history. The glucose response and conditioning
history were unchanged.

The 40 source features and their nuisance regressions were reconstructed in
each repetition, and participants were tested separately using the adaptive
precision-weighted GTCM. If \(p_i^{(b)}\) denotes the resulting participant
\(p\)-value, the cohort test and its empirical rejection rate were
\begin{equation}
 T^{(b)}=-2\sum_{i=1}^{6}\log p_i^{(b)},\quad
 p_{\mathrm{global}}^{(b)}=\Pr\{\chi^2_{12}\geq T^{(b)}\},
 \quad
 \widehat\alpha=\frac{1}{500}\sum_{b=1}^{500}
 \mathbf 1\{p_{\mathrm{global}}^{(b)}\leq0.05\}.
 \label{eq:ohiot1dm_pseudo_calibration}
\end{equation}
This calibration analysis complements the observed-data rejection by checking
whether the same fitting and aggregation pipeline produces excess discoveries
when the tested temporal information has been deliberately removed.

\begin{table}[t]
\centering
\caption{Cohort-level calibration under independent pseudo-exposures}
\label{tab:ohiot1dm_pseudo_exposure}
\begin{threeparttable}
\begin{tabular}{lccc}
\toprule
Cohort & Rejections & Rejection rate (95\% interval) & Median global \(p\) \\
\midrule
2018 Basis    & 15/500 & 0.030 (0.017--0.049) & 0.555 \\
2020 Empatica & 18/500 & 0.036 (0.021--0.056) & 0.537 \\
\bottomrule
\end{tabular}

\begin{tablenotes}[flushleft]
\footnotesize
\item \textit{Note:}
Intervals are exact binomial 95\% confidence intervals. Rejection is defined at the cohort level at \(\alpha=0.05\), after Fisher combination of the six participant-specific tests.
\end{tablenotes}
\end{threeparttable}
\end{table}

Table~\ref{tab:ohiot1dm_pseudo_exposure} shows no inflation above the nominal
5\% level.  The 2018 estimate is mildly conservative, while the 2020 interval
contains 0.05.  This experiment is an application-specific negative control:
it assesses calibration after deliberately removing temporal information
from the activity process, rather than positing a complete generative model
for the observed longitudinal records.
\section{Complementary Lemmas}
\label{app:complementary-lemmas}

\subsection{Conditional second-moment criterion}

\begin{lemma}[Conditional second-moment criterion]
\label{lem:conditional-second-moment}
Let \(\theta_K\geq0\) and let \(\mathcal G_K\) be a sigma-field.
If
\[
\E\big[\theta_K\given\mathcal G_K\big]=o_p(1),
\]
then \(\theta_K=o_p(1)\). The same conclusion holds if, for some
\(\eta>0\), there is a sequence \(\bar\theta_K\geq0\) such that
\(\bar\theta_K=o_p(1)\),
\[
\sup_K\E\big[\bar\theta_K^{1+\eta}\big]<\infty,
\qquad
\E\big[\theta_K\big]\leq\E\big[\bar\theta_K\big].
\]
Consequently, for a fixed-dimensional random vector
\(\boldsymbol\theta_K\),
\[
\Big\|
\E\big[\boldsymbol\theta_K\boldsymbol\theta_K^\top
\given\mathcal G_K\big]
\Big\|_{\mathrm{op}}=o_p(1)
\quad\Longrightarrow\quad
\boldsymbol\theta_K=o_p(1).
\]
\end{lemma}

\begin{proof}
For every \(\epsilon>0\), conditional Markov's inequality gives
\[
\Pr(\theta_K>\epsilon\given\mathcal G_K)
\leq
\min\Big\{1,
\frac{\E\big[\theta_K\given\mathcal G_K\big]}{\epsilon}\Big\}.
\]
The right-hand side is bounded and converges to zero in probability, so its
expectation converges to zero. This proves the first statement. For the
second, the bounded \((1+\eta)\)-moment and \(\bar\theta_K=o_p(1)\) give
\(\E\big[\bar\theta_K\big]\to0\), and ordinary Markov's inequality then gives
\[
\Pr(\theta_K>\epsilon)
\leq \frac{\E\big[\theta_K\big]}{\epsilon}
\leq \frac{\E\big[\bar\theta_K\big]}{\epsilon}
\longrightarrow0.
\]
For the vector statement, take
\(\theta_K=\|\boldsymbol\theta_K\|_2^2\) and use
\[
\E\big[\|\boldsymbol\theta_K\|_2^2\given\mathcal G_K\big]
=
\operatorname{tr} \Big\{
\E\big[\boldsymbol\theta_K\boldsymbol\theta_K^\top
\given\mathcal G_K\big]
\Big\}
\leq
d\Big\|
\E\big[\boldsymbol\theta_K\boldsymbol\theta_K^\top
\given\mathcal G_K\big]
\Big\|_{\mathrm{op}}.
\]
\end{proof}

\subsection{A weak law of large numbers for mixing sequences}

\begin{lemma}[A weak law of large numbers for mixing sequences]
\label{lemma:WLLN-MS}
Suppose \(\{\zeta_t:t\geq1\}\) is strongly mixing, with coefficients
\[
\alpha(h)=
\sup_t\sup_{A\in\mathcal F(H_t^\zeta),\,
B\in\mathcal F(A_{t+h}^\zeta)}
\big|\Pr(A\cap B)-\Pr(A)\Pr(B)\big|,
\]
and \(\alpha(h)\to0\) as \(h\to\infty\). If
\(\E[\zeta_t]=0\) and
\(\sup_t\E[|\zeta_t|^{1+\eta}]\leq c\) for some \(c,\eta>0\), then
\[
\frac1K\sum_{t=1}^K\zeta_t=o_p(1).
\]
\end{lemma}

\begin{proof}
Fix \(M>0\). The bounded transformations
\(\zeta_t\mathbf 1\{|\zeta_t|\leq M\}\) have mixing coefficients no
larger than \(\alpha(h)\). The covariance inequality for bounded strongly
mixing variables therefore gives
\begin{align*}
&\Var\Big\{
\frac1K\sum_{t=1}^K
\Big(\zeta_t\mathbf 1\{|\zeta_t|\leq M\}
-\E[\zeta_t\mathbf 1\{|\zeta_t|\leq M\}]\Big)
\Big\}\\
&\quad\leq
\frac{CM^2}{K}
+\frac{CM^2}{K}\sum_{h=1}^{K-1}\alpha(h)
=o(1),
\end{align*}
where the last equality follows from Ces\`aro convergence. Chebyshev's
inequality consequently makes the centred truncated average \(o_p(1)\)
for every fixed \(M\).

The moment condition also gives, uniformly in \(t\),
\[
\E\big[|\zeta_t|\mathbf 1\{|\zeta_t|>M\}\big]
\leq
M^{-\eta}\E[|\zeta_t|^{1+\eta}]
\leq cM^{-\eta}.
\]
Because \(\E[\zeta_t]=0\), the absolute mean of the truncated average is
at most \(cM^{-\eta}\). Markov's inequality also gives
\[
\Pr\Big\{
\Big|\frac1K\sum_{t=1}^K
\zeta_t\mathbf 1\{|\zeta_t|>M\}\Big|>\epsilon
\Big\}
\leq \frac{c}{\epsilon M^\eta}.
\]
Letting first \(K\to\infty\) and then \(M\to\infty\) proves the result.
\end{proof}

\subsection{Validity of the adaptive Monte Carlo calibration}

{
\begin{proposition}
\label{prop:adaptive_calibration}
Suppose that the conditions of Theorem~\ref{thm:vector_gtcm} hold. Write
$[d]=\{1,\ldots,d\}$, and let
$\mathcal S\subseteq 2^{[d]}\setminus\{\varnothing\}$ be the fixed finite
nonempty collection of candidate coordinate subsets used by
Algorithm~\ref{alg:bootstrap_subset}. Let $B\in\mathbb N$ be the number of
independent Gaussian Monte Carlo draws used by the algorithm. Then, for
every $\alpha\in(0,1)$,
\begin{equation}
\label{eq:adaptive_size_control}
\Pr\{p_{\mathrm{adj}}\leq\alpha\}
\longrightarrow
\frac{\lfloor\alpha(B+1)\rfloor}{B+1}
\leq\alpha.
\end{equation}
For fixed \(B\), the adjusted p-value is discrete and the test is generally
conservative. If \(B=B_K\to\infty\), then
\(p_{\mathrm{adj}}\xrightarrow{d}\operatorname{Unif}(0,1)\) and
$\Pr\{p_{\mathrm{adj}}\leq\alpha\}\to\alpha$. Thus the complete subset
search, rather than any one selected subset alone, is calibrated under the
null.
\end{proposition}}

\begin{proof}
Since $\hat\Sigma\xrightarrow{p}\Sigma$ and $\Sigma$ is positive definite,
all covariance submatrices indexed by the fixed finite collection
$\mathcal S$ are nonsingular with probability tending to one.

First suppose that $B$ is fixed. By Theorem~\ref{thm:vector_gtcm},
consistency of $\hat\Sigma$, and independence of the Monte Carlo draws, the
observed score and its $B$ Gaussian replicates converge jointly to $B+1$
independent $\mathcal N_d(0,\Sigma)$ random vectors. Since $\mathcal S$ is
finite and the relevant covariance submatrices are nonsingular, the
adaptive statistic is a continuous function of the score and covariance
matrix. Its limiting distribution has no atoms, since it is
the maximum of finitely many transformed chi-squared variables, each
with a continuous distribution. The observed and
simulated adaptive statistics therefore converge jointly to $B+1$ i.i.d.\
continuous random variables.

Consequently, the limiting rank of the observed statistic among the
$B+1$ values is uniform on $\{1,\ldots,B+1\}$. By the definition of
$p_{\mathrm{adj}}$,
\begin{equation}
    \Pr \Big(p_{\mathrm{adj}}\leq\alpha\Big)
    \longrightarrow
    \frac{\lfloor \alpha(B+1)\rfloor}{B+1},
\end{equation}
which proves \eqref{eq:adaptive_size_control}. Now let $B=B_K\to\infty$. Conditional on $T^{(K)}$ and $\hat\Sigma$, the
Monte Carlo fraction entering $p_{\mathrm{adj}}$ is an average of
$B_K$ independent Bernoulli variables, with conditional variance bounded
by $1/(4B_K)$. Thus
\begin{equation}
    p_{\mathrm{adj}}
    =
    \Pr\big(
        T_{\mathrm{ad}}^{*}{\geq} T_{\mathrm{ad}}^{(K)}
        \given
        T^{(K)},\hat\Sigma
    \big)
    +o_p(1),
\end{equation}
where $T_{\mathrm{ad}}^{(K)}$ and $T_{\mathrm{ad}}^{*}$ denote the observed
and Gaussian Monte Carlo adaptive statistics, respectively.

By Theorem~\ref{thm:vector_gtcm}, consistency of $\hat\Sigma$, and
continuity of the limiting adaptive statistic, the conditional probability
on the right converges in distribution to $1-F(T)$, where $F$ is the
continuous distribution function of the limiting adaptive statistic
$T$. Hence, by the probability integral transform,
\begin{equation}
    p_{\mathrm{adj}}
    \xrightarrow{d}
    \operatorname{Unif}(0,1).
\end{equation}
Therefore
$\Pr(p_{\mathrm{adj}}\leq\alpha)\to\alpha$ for every
$\alpha\in(0,1)$.
\end{proof}


\end{document}